\documentclass[11pt]{article}

\usepackage[a4paper,margin=1in]{geometry}
\usepackage{amsmath,amssymb,amsthm,mathtools}
\usepackage{bm}
\usepackage{graphicx}
\usepackage{booktabs}
\usepackage{microtype}
\usepackage{tikz}
\usetikzlibrary{arrows.meta,positioning,calc}
\usepackage[colorlinks=true,linkcolor=blue!55!black,citecolor=blue!55!black]{hyperref}
\usepackage{subcaption}
\usepackage{orcidlink}
\usepackage{authblk}

\DeclareMathOperator{\ri}{ri}
\DeclareMathOperator*{\argmax}{arg\,max}
\DeclareMathOperator*{\argmin}{arg\,min}

\newcommand{\R}{\mathbb{R}}
\newcommand{\E}{\mathbb{E}}
\newcommand{\F}{\mathcal{F}}
\newcommand{\one}{\mathbf{1}}
\newcommand{\Ni}{\mathcal{N}_i}
\newcommand{\Nip}{\mathcal{N}_i'}
\newcommand{\Dkl}[2]{D_{\mathrm{KL}}\!\left(#1\,\middle\|\,#2\right)}

\theoremstyle{plain}
\newtheorem{theorem}{Theorem}[section]
\newtheorem{proposition}[theorem]{Proposition}
\newtheorem{lemma}[theorem]{Lemma}
\newtheorem{corollary}[theorem]{Corollary}
\theoremstyle{definition}
\newtheorem{definition}[theorem]{Definition}
\newtheorem{assumption}{Assumption}
\theoremstyle{remark}
\newtheorem{remark}[theorem]{Remark}

\makeatletter
\renewenvironment{proof}[1][\proofname]{\par
    \pushQED{\qed}%
    \normalfont \topsep6\p@\@plus6\p@\relax
    \trivlist
    \item[\hskip\labelsep
        \bfseries    
    #1\@addpunct{.}]\ignorespaces
}{%
  \popQED\endtrivlist\@endpefalse
}
\makeatother

\title{\bfseries Contagious Misinterpretations: Social Learning with Shared Interpretations}

\author[1]{Pol Rodríguez Farrés\, \orcidlink{0009-0005-3877-9720}\thanks{Contact: \href{mailto:pol.rodriguez@iiia.csic.es}{\texttt{pol.rodriguez@iiia.csic.es}}}}
\author[2]{Jesus Cerquides}

\affil[1]{IIIA-CSIC, Bellaterra, Spain}
\affil[2]{CSIC, Madrid, Spain}

\date{\today}

\begin{document}

\maketitle

\begin{abstract}
When people discuss which of several hypotheses is true, they share not only their beliefs but also their interpretation of each one, that is, the distribution of signals each hypothesis is expected to produce. This interpretation is an opinion about what each hypothesis means and what evidence it predicts, and it is revised only through conversation, not by confronting it with the observed data. Standard non-Bayesian social learning captures only belief-sharing and keeps the likelihoods private. We study what happens when agents share both beliefs and likelihoods under static networks. We show that interpretations converge geometrically to a common reading shaped by the most influential agents. Then, beliefs concentrate, almost surely and at an exponential rate, on the hypothesis that, under this consensus interpretation, best explains the signals privately observed by all agents taken together.

As a result, sharing interpretations can lead a population to adopt, unanimously and confidently, a hypothesis it would have rejected if it had shared only beliefs. Whether this helps or hurts when a true hypothesis exists depends on how accurately agents initially read the data they observe. When every agent initially interprets the true hypothesis exactly as the data it would observe, sharing interpretations cannot strictly improve learning and can strictly harm it when agents observe different data distributions. When interpretations are logically coherent and some agents misread their own data, sharing can either rescue the population, leading it to learn the true hypothesis, if an accurate reading dominates the interpretation network, or corrupt it, leading it to a wrong hypothesis, if a misreading one does.
\end{abstract}


\section{Introduction}\label{sec:intro}

Communication within real communities conveys more than conclusions or beliefs. Scientists and experts share not only their assessments of which hypotheses are correct, but also the methods, models, and interpretive frameworks through which evidence is understood. The same holds on social platforms. Users pass along not only endorsements of claims, through likes, shares, or expressions of agreement, but also interpretations of those claims: the commentary attached to a forwarded post, the explanation of why a given piece of evidence supports a conclusion, a community note, a quoted repost that reframes the material. Exposure to such communication can therefore affect an agent in two distinct ways. It can change the agent's assessment of which hypothesis is true, but it can also change the agent's interpretation of the evidence itself, that is, how likely a given observation is under each hypothesis.

Standard models of non-Bayesian social learning capture only the first of these. In them, agents exchange beliefs over a network and combine what they hear with their own observations, but the rule that turns an observation into support for a hypothesis stays private, i.e., no one communicates how they interpret the evidence, which is especially important when the population receives the same data. This paper asks what happens when interpretations are shared as well.

We study a population of agents learning among finitely many hypotheses on the basis of private, noisy signals. Each agent carries two things at once, a belief about which hypothesis is most likely to be true, and an interpretation of each hypothesis, that is, a reading of how probable each possible observation would be under each hypothesis. In every period, agents pool their neighbors' beliefs over one network and their neighbors' interpretations over a second, and each agent then performs a single Bayesian update using its freshly pooled interpretation. A person may adopt someone's framing while rejecting their conclusions, or the reverse. The two networks therefore need not coincide, and the influence a given agent exerts in one need not match the influence it exerts in the other. A key simplification is that interpretations are updated exclusively through conversation, never by confronting them with data. This isolates the effect of socially transmitted readings, but it also implies that, in our model, no amount of evidence can change the population's expectations about what each hypothesis predicts.

The analysis is organized around four findings. First, under our assumptions, we show that interpretations reach a deterministic consensus. Because the interpretation update step never consults a signal or a belief, its dynamics are decoupled from them, and it turns out to be exactly solvable. Moreover, every agent's interpretation model converges geometrically to a single shared reading of the evidence, which will be close to the initial readings of the most influential agents in that network. Second, we show that beliefs concentrate on the hypotheses that best explain the global data observed under this shared reading, where the global data weights each agent's observed data by its influence in the belief network. The main result shows that every agent eventually rules out every suboptimal hypothesis almost surely and at an exponential rate, even when no hypothesis reproduces anyone's signals exactly.

The remaining two findings concern how these pieces relate. Third, the two networks turn out to act on different objects. What the population comes to mean by each hypothesis is fixed entirely by the interpretation channel, i.e. the starting interpretations and who is influential in spreading them. Whose individual data are weighed against those meanings is fixed entirely by the belief side and the influential agents in that channel. The fourth finding shows that sharing interpretations can reverse the outcome of belief-only learning and can explain dynamics that cannot be explained with standard social learning models. There are populations that end up believing, unanimously and quickly, a hypothesis that the very same population, pooling beliefs over the very same network but keeping its interpretations private, would have unanimously rejected. Then, when people share their interpretations of the evidence and hypotheses, the resulting population dynamics of beliefs can differ substantially from those that arise when people share only their beliefs, either helping or harming agents’ learning of the truth.

These results let us take a position on the two questions behind the model: \emph{What does the group converge to, and what does it depend on?}
\emph{Does sharing interpretations help or hurt?}

We address these questions across settings with and without a true hypothesis, under homogeneous and heterogeneous signal distributions, and with either a correctly specified population or one in which, some or all agents, initially misinterpret their observations.

\paragraph{A running example.}
Apart from the theoretical results, we carry one concrete scenario through the whole paper. It is a deliberate simplification, meant only to make the model and its conclusions easier to follow. The example takes place during an advanced vaccination campaign, where a population tries to decide whether the vaccine is effective or not by observing whether severe hospitalizations occur among vaccinated or unvaccinated individuals. During the decision process, people keep observing hospitalizations and share with others two very different kinds of thing. On the one hand, they share their belief in the effectiveness of the vaccine, for example, through conversations with family, colleagues, and clinicians. On the other hand, they share their interpretation of each hypothesis, i.e., what they would expect hospitals to look like if the vaccine works, and what they would expect if it does not. 

Whether the vaccine is effective has an objective answer. Since vaccination coverage is known, an effective vaccine is one that leaves vaccinated people with fewer severe cases than their share of the population. But people do not hold this definition in mind when they look at the news, they hold a picture. For some, an effective vaccine means vaccinated patients almost never reach hospital, and anything less is a disappointment. For others, a clear reduction is enough. These pictures are personal standards, not facts, and none of them is the correct one. Yet they shape how evidence is interpreted: people favor the hypothesis that best matches their mental representation of the data. As a result, someone with a very demanding standard may look at objectively effective evidence, a vaccine that clearly protects but not perfectly, and conclude that it does not work. Because these standards are opinions, they also travel. People soften or harden them by listening to others, mostly through the media and social platforms. These interpretations may be communicated through the media and social platforms. The voice that shapes how you read the evidence need not be the voice whose verdict you repeat, and a framing that travels widely online can reach far more people than its author ever persuades in person. Each idea presented in the article is illustrated in this simple scenario, where the full setup is presented in Section~\ref{subsec:example}.

\section{Related work}

Our model sits at the meeting point of several literatures.

\paragraph{Social learning and distributed hypothesis testing.}

A large literature studies how agents connected through a network communicate and aggregate information, answering questions such as what the population ends up believing, whether it reaches consensus, and how fast. In economics this is studied under the name of social learning. The most basic model is DeGroot's \cite{degroot1974reaching}, in which agents update beliefs by iteratively averaging with their neighbors. The Friedkin--Johnsen model \cite{friedkin1990social} extends it by letting agents retain a fixed attachment to their initial opinion. Some of these models incorporate signals only through agents' initial beliefs \cite{golub2010naive,demarzo2003persuasion}, while others let agents receive signals continuously during the decision process. A parallel line of work studies perfectly rational agents, known as Bayesian social learning \cite{gale2003bayesian,acemoglu2011bayesian}, in which agents perform exact inference over both their own signals and what their neighbors' actions or messages reveal about the signals those neighbors saw. This is informationally demanding, analytically difficult, and not a realistic model of how people actually reason. Non-Bayesian social learning bridges the two approaches mentioned. For example, \cite{jadbabaie2012non} combined DeGroot-style belief aggregation with Bayesian updates on fresh private signals at each time period, and \cite{molavi2018theory} provided behavioral foundations for this class of rules. They postulate imperfect recall, under which agents treat their neighbors' current beliefs as sufficient statistics for the whole history of their observations, and show that, together with three further restrictions on how agents process their neighbors' information, the social learning rule must take a log-linear form. In other words, log-linear learning is the rule that departs from Bayesian rationality only through imperfect recall. This last work shapes the belief update in this article \eqref{eq:S2},\eqref{eq:S3}, where agents listen to their neighbors' last period beliefs and observe a private signal at each time period.

The same updates were developed in the engineering literature under the name distributed hypothesis testing, in which the question asked changes. Rather than asking what the population believes in the long run, this literature asks how fast beliefs concentrate and how reliably they concentrate correctly, and it proves sharp, non-asymptotic rates for both \cite{shahrampour2015distributed,lalitha2018social}. \cite{rahimian2017bayesian} study a related variant in which agents cannot recall past communication and must treat their current belief as a full summary of it. These works commonly treat each agent's likelihood (how they read evidence as support for one hypothesis or another) as private and well-specified, i.e. they expect to see a distribution of signals under the true hypothesis exactly as the distribution of data they will observe. Our model departs here by both considering misspecification relative to the true signal distribution and by treating the likelihood itself as an opinion, something agents can also communicate and revise, alongside their beliefs.

\paragraph{Misspecification with a fixed model.}

An agent is misspecified when it is simply wrong, i.e. none of the likelihoods associated with the hypotheses matches the true distribution of data observed. A line of work asks what misspecification does to beliefs when the wrong model is held fixed forever. \cite{nedic2017fast} prove, for time-varying undirected graphs with doubly-stochastic weights, that agents still concentrate on the hypotheses that best explain the observations of the whole network, at a geometric and non-asymptotic rate, even when no hypothesis matches any agent's data. Our convergence analysis follows the structure of theirs, adapted to static, possibly directed, row-stochastic networks, where each agent is weighted by its stationary influence as in \cite{lalitha2018social}. Berk-Nash equilibrium \cite{esponda2016berk} characterizes where individual learning settles under such a fixed, wrong model, and this kind of misspecification has been shown to produce persistent errors and self-confirming belief traps even in otherwise sufficiently informative environments \cite{bohren2016informational,heidhues2018unrealistic}. \cite{acemoglu2016fragility} show that even a vanishing amount of uncertainty about how evidence should be read can stop beliefs from ever fully agreeing. In a social setting, agents who each read the same signals through different fixed models can fail to converge to the truth at all \cite{frick2020misinterpreting,bohren2021learning}, and the same mechanism helps explain the spread of misinformation, where more sophisticated agents can be more, not less, vulnerable to systematically misreading evidence \cite{mostagir2022learning}. What all of this work has in common is that each agent's likelihood is static.

\paragraph{Uncertain models: likelihoods disciplined by data.}

A second line instead lets the model itself move. Agents estimate their likelihoods from finite data alongside the belief dynamics \cite{hare2020non,hare2021general}. Here misspecification is transient rather than permanent, since the model is disciplined by evidence and can in principle be corrected by it.

Again, our setting differs from both branches in the same way. Rather than a likelihood that is fixed forever or one that is learned from data, we let likelihoods be socially aggregated, exchanged between agents as opinions and interpretations of the world and never confronted with data directly. This separation from the data makes sense when agents know neither the true distribution of the signals nor the true hypothesis, and therefore cannot determine which hypothesis's reading should be drawn toward the observed data. With this convention, we isolate the pure effect of socially transmitted readings, a case that neither a fixed nor a data-driven model of the likelihood can capture, and we study it in both the well-specified and misspecified versions of the underlying environment.

\paragraph{Opinion pooling.}

Finally, we model the likelihood-aggregation step \eqref{eq:S1} as a logarithmic opinion pool, one of the two canonical rules for combining probability distributions, alongside the linear pool introduced by Stone \cite{stone1961opinion} and surveyed in \cite{dietrich2016probabilistic}. Its axiomatic and external-Bayesian properties were developed by \cite{genest1986combining, genest1986characterization}. What is new here is running two such log-linear aggregation processes at once, one on beliefs and the other on likelihoods, over two possibly distinct networks. The question then becomes how their separate influence structures jointly determine what the population ends up learning.

\paragraph{Roadmap.}
Section \ref{sec:model} presents mathematical notation and the model: the agents, signals, the two networks (beliefs and interpretations), and the dual-channel update rule. Sections~\ref{sec:likelihood}-\ref{sec:characterizing_population} solve the long-run dynamics of the update. We first show that interpretations reach a deterministic consensus for all agents (Section~\ref{sec:likelihood}), then show that beliefs concentrate almost surely and exponentially on the hypothesis that best explains the global and heterogeneous data under that consensus interpretation (Section~\ref{sec:beliefs}), and finally show that the outcome can be reduced to the inference of a single effective agent, allowing a geometric visualization of the learning (Section~\ref{sec:characterizing_population}). Section~\ref{sec:comparison} then asks whether sharing interpretations helps or hurts, comparing the model against the benchmark model where only beliefs are shared: Section~\ref{subsec:misspecified_benchmark} establishes that sharing interpretations can overturn belief-only learning when no true hypothesis exists; Section~\ref{subsec:wellspecified} shows that in the well-specified case, with a true hypothesis, sharing is weakly dominated and strictly harmful under heterogeneous observed data across agents; and Section~\ref{subsec:coherent_sharing_good_or_bad} considers the case where all agents have coherent interpretations with the definition of each hypothesis and the world, and a ground truth exists. In this scenario, when part or all the population misreads its own data, we show that sharing interpretations can either rescue the whole population to the correct hypothesis or corrupt it to a wrong one. Section~\ref{sec:discussion} discusses implications and future work, and Section~\ref{sec:conclusions} concludes.

\section{Model}\label{sec:model}

\subsection{Notation and Preliminaries}\label{subsec:notation}

We first define some notation and definitions used in the following sections. 

For a finite set $X$, let $\Delta(X)$ denote the probability simplex on $X$ and $\ri\Delta(X)$ its relative interior (the full-support distributions). For $p\in\Delta(X)$ and $q\in\ri\Delta(X)$, the Kullback-Leibler (KL) divergence is

\begin{equation}
    \Dkl{p}{q}=\sum_{x\in X}p(x)\log\frac{p(x)}{q(x)}\in[0,\infty),
\end{equation}

it is finite whenever $q$ has full support. Throughout, KL divergence plays the role of a (directed) measure of fit between a data-generating law and a model: $\Dkl{p}{q}$ is small precisely when $q$ is a good explanation of data drawn from $p$.

We write $\one\in\R^n$ for the all-ones vector and $\delta_{ij}$ for the Kronecker delta. For matrices $M,M'\in\R^{n\times n}$, $M\circ M'$ is the Hadamard (entrywise) product. 

We use the $\ell_\infty$ norm $\|v\|_\infty:=\max_i|v_i|$ for $v\in\R^n$ and $\|M\|_\infty:=\max_i\sum_j|M_{ij}|$ (the induced maximum absolute row-sum norm) for $M\in\R^{n\times n}$.

We use one standard consequence of the Perron-Frobenius theorem, proved in Appendix~\ref{app:pf}. If $M\in\R^{n\times n}$ is row-stochastic and primitive, it has a unique stationary distribution $\pi_M\in\ri\Delta(N)$ with $\pi_M^{\top}M=\pi_M^{\top}$, and $M^t\to\one\pi_M^{\top}$ geometrically, at a rate governed by the second-largest eigenvalue modulus $\rho_M<1$. Precisely, for a nonnegative sequence $(X_t)_{t\ge0}$ we write $X_t=O(\rho_M^{\,t})$ to mean that for every $\hat\rho\in(\rho_M,1)$ there is a constant $C_{\hat\rho}<\infty$ with $X_t\le C_{\hat\rho}\,\hat\rho^{\,t}$ for all $t\ge0$, and that $\hat\rho=\rho_M$ is also admissible whenever $M$ is diagonalizable. With this convention $\|M^t-\one\pi_M^{\top}\|_\infty=O(\rho_M^{\,t})$ (Lemma~\ref{lem:pf-rate}), and in particular $\sum_{t\ge0}\|M^t-\one\pi_M^{\top}\|_\infty<\infty$. Every geometric rate below is understood in this sense.

\subsection{Agents, hypotheses, and signals}

We consider a network of $n$ agents, with node set $N=\{1,\dots,n\}$, that learn a finite set of unknown hypotheses $\Theta=\{\theta_1,\dots,\theta_m\}$, $m\ge 2$. At each discrete time $t\ge 1$, agent $i$ privately observes a signal $\omega_{it}\in S$, where $S$ is a finite and common signal set for all agents. The signals are independent across agents and i.i.d. across time, and each agent $i$ draws them from an unknown signal-generating distribution $\omega_{it}\sim f_i$, where $f_i\in\Delta(S)$. We write $\F_t:=\sigma(\omega_{is}:i\in N,\ s\le t)$ for the natural filtration, with $\F_0$ trivial.

Each agent $i$ carries two objects at every time $t$:

\begin{itemize}
    \item a \emph{belief} $\mu_{it}\in\Delta(\Theta)$, its current assessment of the hypotheses;
    \item a \emph{likelihood model} (also referred to as its interpretation) $l_{it}(\cdot\mid\theta)\in\Delta(S)$ for each $\theta\in\Theta$, its assessment of how probable each signal would be if $\theta$ were true, the agent's interpretation of the hypothesis.
\end{itemize}

As a conditional law of the signal given the hypothesis, the likelihood model satisfies $\sum_{s\in S} l_{it}(s\mid\theta)=1$ for every $\theta$ and $t$. Similarly, beliefs are a probability distribution over $\Theta$, so they satisfy $\sum_{\theta\in\Theta}\mu_{it}(\theta)=1$.

Two features of this setup deserve emphasis. First, the data $f_i$ that agents observe may differ across agents, so the population can hold genuinely dispersed information rather than repeated draws from one common truth. Second, no hypothesis's likelihood need correspond to any $f_i$: the agent may be misspecified, and we do not require that any initial interpretation of $\theta\in\Theta$ reproduce any agent's data-generating law.

\subsection{Communication networks and standing assumptions}

Agents communicate over two (possibly distinct) weighted directed networks: a belief-aggregation network with row-stochastic adjacency matrix $A=[a_{ij}]\in\R^{n\times n}$, and a likelihood-aggregation network with row-stochastic adjacency matrix $A'=[a'_{ij}]\in\R^{n\times n}$. We write $\Ni:=\{j:a_{ij}>0\}$ and $\Nip:=\{j:a'_{ij}>0\}$ for the in-neighborhoods of agent $i$ (each including $i$ when the diagonal entry is positive). The two matrices need not share support. The agents $j$ whose beliefs agent $i$ listens to may differ from those whose interpretations it listens to, and even on a common support the weights may differ.

\begin{assumption}[Standing assumptions]\label{ass:standing}
Throughout the paper:

\begin{enumerate}
    \item[(A1)] $A$ and $A'$ are row-stochastic and primitive (irreducible (strongly connected) graph and aperiodic). We denote their unique stationary distributions by $\pi^A$ and $\pi^{A'}$ (with $\pi^A_i>0$, $\pi^{A'}_j>0$ for all $i,j$) and their second-largest eigenvalue moduli by $\rho_A,\rho_{A'}\in[0,1)$.
    
    \item[(A2)] Initial likelihood profiles have full support: $l_{j,0}(\cdot\mid\theta)\in\ri\Delta(S)$ for every $j\in N$ and $\theta\in\Theta$. Initial beliefs have full support: $\mu_{j,0}\in\ri\Delta(\Theta)$.
    
    \item[(A3)] Signal-generating distributions also have full support: $f_i\in\ri\Delta(S)$ for every $i\in N$.
\end{enumerate}
\end{assumption}

Assumption (A1) is the usual connectivity-and-aperiodicity requirement that makes each network mix to a well-defined stationary distribution. Assumption (A2) guarantees that every log-ratio and KL divergence below is finite and, since the updates preserve relative interiors, that $l_{it}(\cdot\mid\theta)\in\ri\Delta(S)$ and $\mu_{it}\in\ri\Delta(\Theta)$ for all $t$. In particular no hypothesis is ever assigned probability zero.

\subsection{The update rule}\label{sec:update}

At each discrete time $t\ge 1$, every agent performs three steps.

\paragraph{Step 1: likelihood pooling (updated model).}

For each hypothesis $\theta$, agent $i$ updates its likelihood model by the normalized weighted geometric mean, the logarithmic opinion pool on $\Delta(S)$, of its neighbors' previous-period models:

\begin{equation}
    l_{it}(s\mid\theta) = \frac{\displaystyle\prod_{j\in\Nip}l_{j,t-1}(s\mid\theta)^{a'_{ij}}} {Z^{\theta}_{i,t-1}}, \qquad Z^{\theta}_{i,t-1} :=\sum_{s'\in S}\prod_{j\in\Nip}l_{j,t-1}(s'\mid\theta)^{a'_{ij}} .
    \tag{S1}\label{eq:S1}
\end{equation}

The normalizer $Z^\theta_{i,t-1}$ enforces $\sum_s l_{it}(s\mid\theta)=1$. By construction, this update depends only on neighbors' previous models and never on any signal nor belief.

\paragraph{Step 2: belief pooling (interim belief).}

Agent $i$ forms an interim belief $\hat\mu_{it}\in\ri\Delta(\Theta)$ by a log-linear pool of its neighbors' previous public beliefs:

\begin{equation}
    \hat\mu_{it}(\theta)\propto \displaystyle\prod_{j\in\Ni}\mu_{j,t-1}(\theta)^{a_{ij}}.
    \tag{S2}\label{eq:S2}
\end{equation}

\paragraph{Step 3: Bayesian revision (public belief).}

Agent $i$ observes $\omega_{it}\sim f_i$ and updates the interim belief by Bayes' rule, under its freshly pooled updated model \eqref{eq:S1}:

\begin{equation}
    \mu_{it}(\theta)= \frac{l_{it}(\omega_{it}\mid\theta)\,\hat\mu_{it}(\theta)} {\displaystyle\sum_{k=1}^{m} l_{it}(\omega_{it}\mid\theta_k)\, \hat\mu_{it}(\theta_k)} .
    \tag{S3}\label{eq:S3}
\end{equation}

Substituting \eqref{eq:S1}-\eqref{eq:S2} into \eqref{eq:S3} yields the fully unrolled one-step update

\begin{equation}
    \mu_{it}(\theta) \propto \frac{\prod_{j\in\Nip}l_{j,t-1}(\omega_{it}\mid\theta)^{a'_{ij}}} {Z^{\theta}_{i,t-1}}\, \prod_{j\in\Ni}\mu_{j,t-1}(\theta)^{a_{ij}},
    \label{eq:unrolled-update}
\end{equation}

in which the interim-belief normalizer of \eqref{eq:S2} cancels between numerator and denominator. Figure~\ref{fig:dual} summarizes the resulting dual-channel structure, in which beliefs propagate over $A$, interpretations propagate over $A'$, and the two meet only in the Bayesian step, where the pooled model serves as the expectation of the signal under each hypothesis.

\begin{figure}[t]
    \centering
    \includegraphics[width=\textwidth]{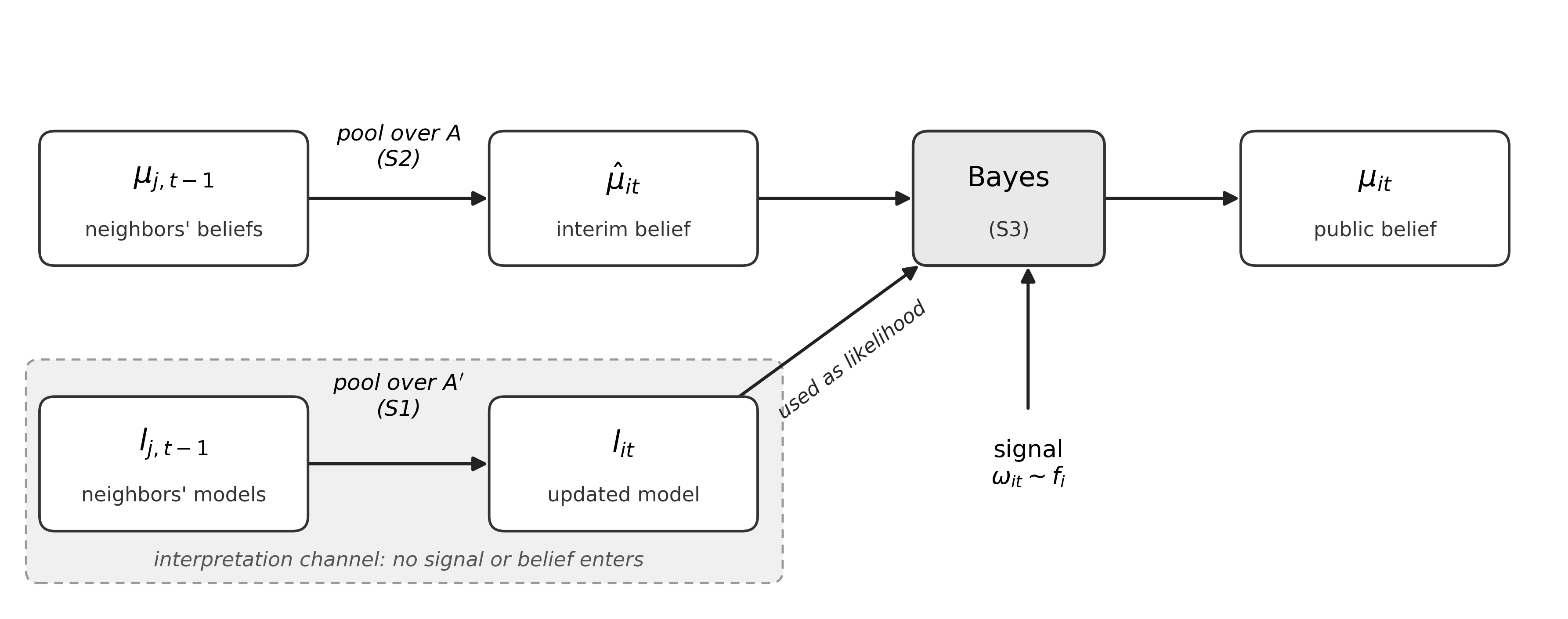}
    \caption{The dual-channel update at agent $i$. Beliefs are aggregated over the belief network $A$ (top, S2) and interpretations over the likelihood network $A'$ (bottom, S1). The two lanes stay independent until they meet at the Bayesian step (S3), where the pooled model serves as the likelihood and the private signal $\omega_{it}$ enters. The interpretation channel never consults a signal or belief, which is what makes its dynamics autonomous and exactly solvable.}
    \label{fig:dual}
\end{figure}

Two aspects of this rule are worth mentioning, because they shape everything that follows.

The interpretation channel uses neither data nor beliefs. The likelihood update \eqref{eq:S1} revises interpretations by social exchange alone, never by private evidence. This is a deliberate modeling choice, parallel to the imperfect-recall convention that is standard in the non-Bayesian learning literature. Interpretations are exchanged as opinions about the meaning of each hypothesis or what each agent expects to observe, and each agent treats its current model as a sufficient statistic for the history of such exchanges. The choice has a decisive structural consequence, established in Section~\ref{sec:likelihood}. The interpretation dynamics are autonomous and deterministic. Moreover, this decoupling is exactly what renders the joint dynamics solvable and isolates, in closed form, the effect of the interpretation channel on collective belief.

The classical model where only beliefs are shared is the special case with $A'=I$. In that case, $l_{it}=l_{i,0}$ for all $t$ and the dynamics reduce exactly to classical misspecified non-Bayesian learning with log-linear belief aggregation and heterogeneous, fixed likelihood models.

\subsection{A running example: reading the pandemic}\label{subsec:example}

We fix one scenario and return to it as each concept appears. It is small enough to check by hand, yet it reproduces a typical episode that the classical model, where agents exchange beliefs but not interpretations, cannot produce.

\paragraph{The scenario.}
During an advanced vaccination campaign where vaccination coverage is already high, a population decides whether the vaccine is effective or not. Two hypotheses are considered, $\Theta=\{\theta_E,\theta_{NE}\}$. The first, $\theta_E$, says the vaccine is effective, so vaccinated individuals face a lower risk of severe disease than unvaccinated ones (though breakthrough cases still occur). The second, $\theta_{NE}$, says it is not effective, or even harmful, i.e., vaccinated individuals face the same or a higher risk. These definitions are objective and shared by everyone.

The population consists of two different communities, possibly located in different regions, but both with vaccination coverage $\alpha=70\%$ of the community. Each community is represented as a single agent in the network, that communicates in two distinct ways. On the one hand, they share beliefs through family, colleagues, and clinicians (network $A$). On the other hand, they communicate their interpretation of each hypothesis, that is, their personal standard of what the evidence should look like if the vaccine works and if it does not. They share those opinions through media and social platforms (network $A'$). Throughout the example, we assume a common interpretation within each community.

\paragraph{Signals.}
The evidence that the population observes is the evidence shared by health institutions and hospitals, where we consider full disclosure and continuous data sharing with the population. We simplify the scenario to just two possible signals: whether a severe case (hospitalization or death) occurs in a vaccinated or an unvaccinated person.

\begin{equation*}
    S=\{\mathsf u,\ \mathsf v\},\qquad
    \mathsf u=\text{severe case, unvaccinated},\qquad
    \mathsf v=\text{severe case, vaccinated}.
\end{equation*}

Agent (community) $i$'s private evidence $\omega_{it}\sim f_i$ is the sequence of such statuses that it observes. We can think of the heterogeneous signal distributions across agents as arising from communities in different regions, where the impact of the disease can differ. 

To present an example with no trivial truth, we specify that each community observes a different distribution of signals. Community 1 observes $f_1=(0.3,0.7)$, i.e. they will observe that $30\%$ of severe cases occur among unvaccinated individuals. In contrast, Community 2 will observe that, of all severe cases, unvaccinated cases account for $70\%$ ($f_2=(0.7,0.3)$).

\paragraph{Interpretations.}

A likelihood model $l_{i,t}(\cdot\mid\theta)$ is community $i$'s interpretation (standard) at time $t$, i.e., the portion of severe cases it expects to see if $\theta$ holds. These standards are subjective. They are not estimates of the true data, and none is identified as correct. They express how much protection a community expects from a vaccine that works, and what it fears from one that does not. For the illustration, we consider that Community 1 is moderate in its requirements for considering a vaccine effective, whereas Community 2 is extreme: either it works almost perfectly or, in contrast, it is harmful. The interpretation values are presented in Table~\ref{tab:ex-frameworks}.

\begin{table}[h!]
  \centering
  \begin{tabular}{lcccc}
    \toprule
    & \multicolumn{2}{c}{Community $1$}
    & \multicolumn{2}{c}{Community $2$}\\
    \cmidrule(lr){2-3}\cmidrule(lr){4-5}
    & $l_{1,0}(\cdot\mid\theta_E)$ & $l_{1,0}(\cdot\mid\theta_{NE})$
    & $l_{2,0}(\cdot\mid\theta_E)$ & $l_{2,0}(\cdot\mid\theta_{NE})$\\
    \midrule
    $\mathsf u$ (severe: unvacc.) & $0.7$ & $0.3$ & $0.95$ & $0.2$\\
    $\mathsf v$ (severe: vacc.)   & $0.3$ & $0.7$ & $0.05$ & $0.8$\\
    \bottomrule
  \end{tabular}
  \caption{Initial likelihoods of each agent (representing the communities) over the observable signals under each hypothesis defined in the running example.}
  \label{tab:ex-frameworks}
\end{table}

Community 1 pictures an effective vaccine as one under which about $30\%$ of severe cases are vaccinated, which is equivalent to interpret a vaccine as effective if it prevents roughly $82\%$ of severe disease. In contrast, it interprets an ineffective one as leaving severe cases split exactly as the population is, $70\%$ cases as vaccinated ($0\%$ prevention). Community 2 is more drastic, and pictures an effective vaccine as nearly perfect, with only $5\%$ of severe cases vaccinated (they consider as effective when the vaccine presents about $98\%$ prevention), and an ineffective one as actively harmful (possibly due to secondary effects), with $80\%$ of severe cases vaccinated, more than the coverage. Community 1's standard for $\theta_{NE}$ matches the data it observes exactly, so it is well specified. Community 2's standards match neither hypothesis to its data, so it is misspecified.

Although standards are subjective, not every interpretation is logically acceptable: it must be coherent with the objective definition of the hypothesis it describes. Coherence is set by the coverage $\alpha$. If the vaccine had no effect on the population, severe cases would simply mirror the population, observing $\alpha\%$ of cases among the vaccinated. An effective vaccine must therefore produce a portion of vaccinated cases strictly below $\alpha$, and an ineffective or harmful one a portion at or above $\alpha$. Then, a coherent interpretation (subjective standard) respects these bounds: $l(\mathsf v\mid\theta_E)<\alpha\le l(\mathsf v\mid\theta_{NE})$. However, coherence does not make an interpretation accurate. In fact, within these bounds, a community whose standards of effectiveness are too high can identify an objectively low effective evidence as being more close to its expectation of what an ineffective one looks like. 

In the specific example, both communities hold coherent interpretations in Table~\ref{tab:ex-frameworks}. Also, note that Community 1 would actually observe $70\%$ of severe cases among the vaccinated, exactly the coverage, which by definition means the vaccine is not effective there. Community 2 instead observes $30\%$, well below the coverage, which objectively means it is effective. Therefore, there is no common ground truth, the two communities face objectively different situations, for instance because of different variants or risk profiles. The interesting question is what the population concludes when the two communities talk.

\section{Likelihood Dynamics: Consensus of Interpretations}
\label{sec:likelihood}

Because the likelihood update \eqref{eq:S1} never uses a signal nor a belief, $l_{it}(\cdot\mid\theta)$ is a deterministic function of the initial profiles $\{l_{j,0}\}_{j\in N}$ and the matrix $A'$, for every $i,\theta,t$. This section solves these dynamics and identifies the limit of \eqref{eq:S1}. Every object in this section is deterministic.

\subsection{Log-ratios undergo linear consensus}

The mechanism is cleanest at the level of log-likelihood ratios, where the nonlinear geometric pool becomes an ordinary linear averaging step.

\begin{lemma}[Log-ratio consensus]\label{lem:logratio}
For each $\theta\in\Theta$ and $s,s'\in S$, define $[\bm r_t(s,s'\mid\theta)]_i:=\log\dfrac{l_{i,t}(s\mid\theta)}{l_{i,t}(s'\mid\theta)}$. Then $\bm r_t(s,s'\mid\theta)=A'\,\bm r_{t-1}(s,s'\mid\theta)$ for all $t\ge 1$, and hence:

\begin{equation}
    \bm r_t=A'^{\,t}\bm r_0\to\one\,\pi^{A'\top}\bm r_0.
\end{equation}
\end{lemma}

In other words, Lemma~\ref{lem:logratio} says that aggregating likelihoods geometrically is, in log-ratio coordinates, exactly a DeGroot (linear consensus) update. Every agent's vector of log-likelihood ratios is therefore driven to the common value $\pi^{A'\top}\bm r_0$, a $\pi^{A'}$-weighted average of the initial log-ratios. The normalization $\sum_s l_{it}(s\mid\theta)=1$ then pins down a unique limiting distribution. We make this explicit with a closed form.

\subsection{Closed form and the consensus model}

At finite time, we can express the interpretation an agent will hold in every hypothesis $\theta$ as follows.

\begin{proposition}[Closed form for the likelihood models]\label{prop:l-closed}
For every $i\in N$, $\theta\in\Theta$, $s\in S$, and $t\ge 0$,

\begin{equation}
    l_{i,t}(s\mid\theta) = \frac{\displaystyle\prod_{j=1}^{n}l_{j,0}(s\mid\theta)^{[A'^{\,t}]_{ij}}} {\Psi^{\theta}_{i,t}} \quad , \quad  \Psi^{\theta}_{i,t}= \sum_{s'\in S}\prod_{j=1}^{n}l_{j,0}(s'\mid\theta)^{[A'^{\,t}]_{ij}} .
  \label{eq:l-closed}
\end{equation}

Since $[A'^{0}]_{ij}=\delta_{ij}$ and $\sum_s l_{i,0}(s\mid\theta)=1$, one has $\Psi^{\theta}_{i,0}=1$.
\end{proposition}

The closed form says that agent $i$'s model at time $t$ is a geometric pool of all initial models weighted by the cumulative influence that agent $j$ has exerted on agent $i$ through time $t$, as captured by the entries $[A'^{\,t}]_{ij}$. Notice that the model does not depend only on $i$'s neighbors, but on all the agents in the network.

As $t$ grows, every row of $A'^{\,t}$ approaches the stationary distribution $\pi^{A'}$, and the weights therefore become common to all agents. As a consequence, each agent's interpretation of each hypothesis will converge to a common consensus interpretation for everyone.

\begin{definition}[Consensus interpretation]\label{def:l-consensus}
Under Assumption~\ref{ass:standing}, the consensus interpretation of each hypothesis $\theta\in\Theta$, $l^*(\cdot\mid\theta)$, is the logarithmic opinion pool of the agents' initial likelihood models, weighted by their influence $\pi^{A'}_j$:

\begin{equation}
    l^*(s\mid\theta) :=\frac{\displaystyle\prod_{j=1}^n l_{j,0}(s\mid\theta)^{\pi^{A'}_j}} {\Psi^{\theta}_{\infty}}\in\ri\Delta(S), \qquad \Psi^{\theta}_{\infty} :=\sum_{s'\in S}\prod_{j=1}^n l_{j,0}(s'\mid\theta)^{\pi^{A'}_j}.
    \label{eq:lstar}
\end{equation}
\end{definition}

\begin{theorem}[Consensus of interpretations]\label{thm:l-consensus}
Under Assumption~\ref{ass:standing}, every agent's interpretation of the hypotheses converges to the consensus interpretation $l^*(\cdot\mid\theta)$ from Definition~\ref{def:l-consensus}, and does so geometrically and uniformly across agents and signals, in the sense of Section~\ref{subsec:notation}:

\begin{equation}
    \max_{i\in N,\,s\in S}\bigl|l_{i,t}(s\mid\theta)-l^*(s\mid\theta)\bigr|
    =O\!\left(\rho_{A'}^{\,t}\right)
    \qquad\text{for every }\theta\in\Theta.
    \label{eq:l-conv}
\end{equation}
\end{theorem}

In other words, although agents initially differ about the distributions of signals they expect under each hypothesis, this disagreement dies out at a geometric rate ($\rho_{A'}^t$). Eventually, the entire population will share one common way of reading the world. More specifically, agents who are influential in the interpretation network (large $\pi_j^{A'}$) will pull the shared meaning of each hypothesis toward their own initial reading.

\paragraph{In the running example.}
Suppose Community 2's demanding standard has the larger reach in the interpretation channel $A'$, i.e. the view that a vaccine should be considered as effective if it is almost perfect and that one that is not may be harmful, circulates more widely. Following this intuition, we can define the network $A'=\begin{pmatrix} \frac{1}{2} & \frac{1}{2} \\ \frac{1}{18} & \frac{17}{18} \end{pmatrix}$, where Community 1 puts the same weight on its own interpretation and on the one coming from Community 2, whereas the second community listens hardly at all to the Community 1 and strongly believes it has the correct interpretation of the world. This matrix has the stationary distribution $\pi^{A'}=(\frac{1}{10},\frac{9}{10})$, from which we can conclude that in the interpretation channel, Community 2 is the most influential. Therefore, when communicating their interpretations, Theorem~\ref{thm:l-consensus} implies that both communities will end up having the same interpretation of each hypothesis:

\begin{equation*}
    l^*(\cdot\mid\theta_E)\approx(0.94,0.061),\qquad
    l^*(\cdot\mid\theta_{NE})\approx(0.21,0.79).
\end{equation*}

Both pooled interpretations are dragged toward Community 2's readings, purely because that community is influential in $A'$. This is the sense in which influential interpreters fix what the population comes to mean by each hypothesis.

\section{Belief Dynamics and Consistency}\label{sec:beliefs}

Having solved the interpretation channel in Section~\ref{sec:likelihood}, we now turn to the belief dynamics. One question organizes the whole section: on which hypotheses does every agent's belief concentrate as $t\to\infty$? The answer is given by a single score per hypothesis called group confidence. Beliefs will then concentrate on the hypotheses with the highest score (Theorem \ref{thm:consistency}).

\begin{remark}
    The belief recursion inherits its skeleton from the distributed non-Bayesian learning of \cite{nedic2017fast}. The one difference lies in the object sitting in the likelihood slot of the Bayesian step \eqref{eq:S3}. In \cite{nedic2017fast} each agent's model is fixed. Here it is the evolving pool $l_{j,t}$ of Section~\ref{sec:likelihood}, so its limit must be carried through the Bayesian step. We state the results here and collect the derivation of the belief recursion and all proofs in Appendix~\ref{app:proofs_beliefs}.
\end{remark}

\subsection{Group confidence and the optimal hypothesis set}

\begin{definition}[Group confidence]\label{def:confidence}
The group confidence of $\theta\in\Theta$ is
 
\begin{equation}
    C(\theta):=-\sum_{i=1}^n\pi^A_i\,\Dkl{f_i}{l^*(\cdot\mid\theta)} .
    \label{eq:confidence}
\end{equation}
\end{definition}

$C(\theta)\leq 0$ scores how well the common consensus interpretation $l^{*}$ of Theorem~\ref{thm:l-consensus} reads each agent's observed signal distribution, weighted by their influence in the belief network. In other words, it scores each hypothesis by weighting each agent by their influence in the network and by how confident they are in that hypothesis (how close each agent's observed data is to the consensus interpretation).

A larger $C(\theta)$ means a better collective fit, and $C(\theta)=0$ would mean the consensus reads $\theta$ as a perfect explanation of every agent's data. Moreover, notice that the two networks ($A$ and $A'$) enter $C$ in separate roles: $\pi^{A}$ sets whose data count ($f_i$), while $\pi^{A'}$ determines what to expect under the hypothesis $\theta$ ($l^{*}$).

With this scoring, we can now define the optimal hypothesis set as the hypotheses in which the population is more confident in.

\begin{definition}[Optimal hypothesis set]\label{def:optimal}
\begin{equation}
    \Theta^*:=\argmax_{\theta\in\Theta}C(\theta),
    \qquad
    C^*:=\max_{\theta\in\Theta}C(\theta) .
    \label{eq:optimal-set}
\end{equation}
 
Since $\Theta$ is finite, $\Theta^*$ is nonempty.
\end{definition}

Hypotheses with equal group confidence are not separated at an exponential rate. The results below say nothing about how beliefs are shared among the elements of $\Theta^*$, and Lemma~\ref{lem:twins} gives a case in which they remain split forever. If all hypotheses tie, $\Theta^*=\Theta$, and the results below are vacuous. We therefore assume throughout that $\Theta^*\subsetneq\Theta$.

Therefore, the population’s beliefs in a non-optimal hypothesis decrease over time and concentrate on the optimal set.

\begin{theorem}[Consistency]\label{thm:consistency}
Under Assumption~\ref{ass:standing}, the update \eqref{eq:S1}-\eqref{eq:S3} satisfies that for every agent $i\in N$, their beliefs on every non-optimal hypothesis ($\theta\notin \Theta^*$) tend to zero as $t\to\infty$ ($\lim_{t\to\infty}\mu_{i,t}(\theta)=0$ a.s.). Moreover, for every $\theta_v\notin\Theta^*$, every $\theta_w\in\Theta^*$, and every $i\in N$,
 
\begin{equation}
    \frac1t\,\ \log \frac{\mu_{i,t}(\theta_v)}{\mu_{i,t}(\theta_w)} \xrightarrow[t\to\infty]{\mathrm{a.s.}}\ C(\theta_v)-C^*\ <\ 0 .
    \label{eq:phi-rate}
\end{equation}
\end{theorem}

The proof is in Appendix \ref{app:consistency-proof}. In words, almost surely, the log-belief ratio of a suboptimal hypothesis against an optimal one falls linearly in $t$, at the rate $C^*-C(\theta_v)>0$, which is the $\pi^A$-weighted gap in fit between the two hypotheses under the consensus model. The convergence of beliefs and its rate do not depend on the initial beliefs, the finite-time disagreement between interpretations (of order $\rho_{A'}^{\,t}$ by Theorem~\ref{thm:l-consensus}), and the speed at which the belief network mixes (through $\rho_A$). All of these wash out in the long run. What is left is a division of labor between the two networks. The interpretation network affects the outcome only through $\pi^{A'}$, which fixes what each hypothesis is taken to mean, through $l^{*}$. The belief network affects it only through $\pi^{A}$, which fixes whose data are weighed against that meaning, through the weights in $C$. The two channels that stayed separate in the update (Figure~\ref{fig:dual}) stay separate in the limit, each governing a different aspect of what the population learns.
 
This result accounts for many possible optimal hypotheses. However, when one hypothesis is strictly best, the beliefs do not just discard the others, they concentrate fully on it.
 
\begin{corollary}[Complete learning under a unique optimum]\label{cor:singleton} 
Suppose $\Theta^*=\{\theta^\star\}$, i.e.\ $C(\theta^\star)>C(\theta)$ for every $\theta\neq\theta^\star$. Then $\lim_{t\to\infty}\mu_{i,t}(\theta^\star)=1$ a.s. for every $i\in N$.
\end{corollary}
 
The condition asks only that $\theta^\star$ be strictly the best collective explanation, i.e.\ that its $\pi^A$-weighted total divergence $\sum_i\pi^A_i\Dkl{f_i}{l^*(\cdot\mid\theta^\star)}$ be strictly smaller than that of every other hypothesis. It does not ask that any hypothesis's reading coincide exactly with anyone's data. In this sense the population can reach a firm, unanimous conclusion without any of its members ever holding a precise model.

When two hypotheses have the same initial interpretation for every agent, nothing in the dynamics can ever separate them.

\begin{lemma}[Twin hypotheses are never separated]\label{lem:twins}
Let $\theta\neq\theta'$ satisfy $l_{j,0}(\cdot\mid\theta)=l_{j,0}(\cdot\mid\theta')$ for every $j\in N$. Run \eqref{eq:S1}-\eqref{eq:S3} with $A$ satisfying (A1) and any row-stochastic $A'$. This includes $A'=I$, the belief-only benchmark, and every $A'$ satisfying (A1). Start from initial conditions satisfying (A2), and let $c:=\max_{j\in N}\bigl|\log\bigl(\mu_{j,0}(\theta)/\mu_{j,0}(\theta')\bigr)\bigr|$. Then, for every realization of the signals, every $i\in N$ and every $t\ge0$,

\begin{equation*}
    \Bigl|\log\frac{\mu_{i,t}(\theta)}{\mu_{i,t}(\theta')}\Bigr|\le c
    \qquad\text{and}\qquad
    \mu_{i,t}(\theta')\le\frac{1}{1+e^{-c}}<1 .
\end{equation*}

In particular, neither $\mu_{i,t}(\theta')$ nor, by symmetry, $\mu_{i,t}(\theta)$ converges to $1$ on any sample path.
\end{lemma}

\paragraph{In the running example.}

To complete the model we specify the belief network. As in the interpretation channel, Community~2 is the more influential one. Community~1 weights its own beliefs and those of Community~2 equally, while Community~2 puts more weight on its own beliefs, although it listens more to Community~1's beliefs than to its interpretations:

\begin{equation*}
    A=\begin{pmatrix} \tfrac12 & \tfrac12\\[2pt] \tfrac{3}{10} & \tfrac{7}{10}\end{pmatrix},
  \qquad
  \pi^{A}=\Bigl(\tfrac38,\tfrac58\Bigr).
\end{equation*}

Hospital data, unlike beliefs and interpretations, are never exchanged. Each community keeps observing its own cases, $f_1=(0.3,0.7)$ and $f_2=(0.7,0.3)$. By the time beliefs settle, both communities read these data through the same consensus standard of Section \ref{sec:likelihood}, which is close to Community~2's demanding one: an effective vaccine should leave about $6\%$ of severe cases among the vaccinated, and an ineffective one about $79\%$. The group confidence of Definition \ref{def:confidence} then works as follows. Each community compares the data it observes with this common standard, and the population's score for each hypothesis adds up these comparisons, each weighted by the community's influence in the belief network. Since everyone holds the same standard, the belief channel decides only whose comparison counts more, and here Community~2's counts $5/8$ against Community~1's $3/8$.

Influence is not the only thing that matters, however, because the two comparisons differ greatly in strength. Under the consensus standard, Community~1's data ($70\%$ vaccinated) are almost exactly what an ineffective vaccine is expected to produce and very far from what an effective one should produce. Community~2's data ($30\%$ vaccinated) fit the effective picture better than the ineffective one, but only moderately, since they fall well short of the near-perfect protection the standard demands. In fact, Community~1's data favor $\theta_{NE}$ by a margin about five times larger than the margin by which Community~2's data favor $\theta_E$, and a belief weight of $5/8$ against $3/8$ is not enough to compensate. The resulting scores are

\begin{equation*}
    C(\theta_E)\approx-0.68,\qquad C(\theta_{NE})\approx-0.36 ,
\end{equation*}

so $\Theta^*=\{\theta_{NE}\}$, and by Corollary~\ref{cor:singleton} the beliefs of both communities almost surely concentrate on the hypothesis that the vaccine is not effective. Theorem \ref{thm:consistency} adds that, in the long run and whatever the initial beliefs, the log-odds against $\theta_E$ grow on average by $C(\theta_{NE})-C(\theta_E)\approx0.33$ per period, so belief in effectiveness vanishes exponentially fast.

Without any communication, each community would have kept its own standard and reached its own conclusion. Community~1 would believe the vaccine is not effective, since its standard for $\theta_{NE}$ matches its data exactly. Community~2 would believe it is effective, since its data are closer to its picture of near-perfect protection than to its picture of a harmful vaccine. Once they communicate, the population adopts Community~1's belief, even though Community~2 is the more influential community in both channels. The reason lies in the interpretation channel: Community~2's demanding standard, once shared by everyone, makes Community~1's data strong evidence against effectiveness and its own data only weak evidence in favor. The following Section \ref{sec:characterizing_population} makes this trade-off visual, and Section \ref{subsec:misspecified_benchmark} compares this result with the one obtained if the Communities only shared their beliefs.

\section{Characterizing What the Population Learns}\label{sec:characterizing_population}

The dominance condition of Corollary~\ref{cor:singleton} looks like it involves all $n$ individual observed distributions at once, but in fact, it collapses to an inference problem for a single effective agent. The reason is that KL divergence splits into a hypothesis-independent entropy term and a cross-entropy term that is linear in each $f_i$. The linear part is all that matters for comparing hypotheses, and linearity lets us aggregate the heterogeneous truths and the belief weights into a single distribution.

\subsection{Reduction to an effective single agent}

\begin{proposition}[Effective single-agent reduction]\label{prop:effective_single_agent}
Define the pooled data

\begin{equation}
  \bar f:=\sum_{i=1}^n\pi^A_i\,f_i\ \in\ \ri\Delta(S) .
  \label{eq:pooled-truth}
\end{equation}

Then, for every $\theta_v,\theta_w\in\Theta$,

\begin{equation}
  C(\theta_w)-C(\theta_v) =\Dkl{\bar f}{l^*(\cdot\mid\theta_v)}-\Dkl{\bar f}{l^*(\cdot\mid\theta_w)},
  \label{eq:confidence-gap-pooled}
\end{equation}

and consequently

\begin{equation}
  \Theta^* =\argmax_{\theta\in\Theta}C(\theta) =\argmin_{\theta\in\Theta}\Dkl{\bar f}{l^*(\cdot\mid\theta)} .
  \label{eq:effective-argmin}
\end{equation}

Thus $\Theta^*$ depends on the individual data $\{f_i\}$ and the belief weights $\pi^A$ only through $\bar f$, and on the initial interpretations $\{l_{j,0}\}$ and the interpretation weights $\pi^{A'}$ only through $\{l^*(\cdot\mid\theta)\}_\theta$.
\end{proposition}

The reduction is genuinely useful: the collective outcome is exactly what a single Bayesian agent with data-generating law $\bar f$ and model family $\{l^*(\cdot\mid\theta)\}$ would learn. It also makes the roles of the two networks fully explicit. The belief side ($\{f_i\}$ and $\pi^A$) defines only the data distribution, while the interpretation side ($\{l_{j,0}\}$ and $\pi^{A'}$) defines how likely each signal is for each hypothesis. Two immediate consequences are worth stating.

\begin{corollary}[A single common truth]\label{cor:singletruth}
If there exists a single common truth, i.e., all signal distributions are the same $f_1=\dots=f_n=f$, then $\bar f=f$ regardless of $\pi^A$, and $\Theta^*=\argmin_\theta\Dkl{f}{l^*(\cdot\mid\theta)}$. In this case the belief influence ($\pi^A$) has no effect on which hypotheses are learned. The outcome is governed entirely by the common truth and the consensus interpretations.
\end{corollary}

\begin{corollary}[Concrete dominance conditions]\label{cor:concrete}
$\Theta^*=\{\theta^\star\}$ if and only if, for every $\theta\neq\theta^\star$, either of the following equivalent conditions holds:

\begin{align}
  &\text{\emph{(i)}}\quad
  \sum_{s\in S}\bar f(s)\,\log\frac{l^*(s\mid\theta^\star)}{l^*(s\mid\theta)}>0;
  \label{eq:concrete-llr}\\[4pt]
  &\text{\emph{(ii)}}\quad
  \sum_{j=1}^n\pi^{A'}_j\sum_{s\in S}\bar f(s)\,
  \log\frac{l_{j,0}(s\mid\theta^\star)}{l_{j,0}(s\mid\theta)}
  \;>\;\log\frac{\Psi^{\theta^\star}_{\infty}}{\Psi^{\theta}_{\infty}} .
  \label{eq:concrete-l0}
\end{align}
In that case the population fully learns $\theta^\star$, i.e.\ $\mu_{i,t}(\theta^\star)\to1$ a.s.\ for all $i$ (Corollary~\ref{cor:singleton}).
\end{corollary}

Corollary~\ref{cor:concrete} characterizes when $\theta^\star$ is the unique optimum, which is sufficient for full learning. If some rival $\theta$ reverses the strict inequality in \eqref{eq:concrete-llr}, then $\theta^\star\notin\Theta^*$ and $\mu_{i,t}(\theta^\star)\to0$ a.s.\ by Theorem~\ref{thm:consistency}. In the remaining boundary case (equality in \eqref{eq:concrete-llr} for some rival and no reversal) our results do not determine the limit of beliefs (see Lemma~\ref{lem:twins}). Condition \eqref{eq:concrete-l0} rewrites the same test entirely in the primitives of the model, the initial interpretations, the interpretation weights $\pi^{A'}$, and the pooled data $\bar f$, and will be the workhorse of the comparisons that follow.

\subsection{Geometry: interpretations set the boundaries}

The reduction has a simple geometric picture on the simplex $\Delta(S)$, the set of all possible data distributions, where the pooled data $\bar f$ is a single point.

\begin{proposition}[Interpretations draw the boundaries, data place the point]\label{prop:geometry}
For each $\theta\in\Theta$, let the decision cell of $\theta$ be the set of data distributions that the consensus interpretation explains strictly better by $\theta$ than by any other hypothesis,

\begin{equation}
  \mathcal R(\theta) :=\Bigl\{p\in\Delta(S):\ \Dkl{p}{l^*(\cdot\mid\theta)}<\Dkl{p}{l^*(\cdot\mid\theta')}\ \ \forall\theta'\neq\theta\Bigr\}.
\end{equation}

By Proposition~\ref{prop:effective_single_agent} and Corollary~\ref{cor:singleton}, if $\bar f\in\mathcal R(\theta)$, every agent's belief in $\theta$ converges to one almost surely. Moreover:

\begin{enumerate}
    \item[\emph{(i)}] Each cell is cut out by linear inequalities in $p$,
    
    \begin{equation*}
    \mathcal R(\theta)=\Bigl\{p\in\Delta(S):\ \textstyle\sum_{s}p(s)\,\nu_{\theta,\theta'}(s)>0\ \ \forall\theta'\neq\theta\Bigr\},
    \qquad
    \nu_{\theta,\theta'}(s):=\log\frac{l^*(s\mid\theta)}{l^*(s\mid\theta')},
    \end{equation*}
    
    so it is convex. Its boundaries depend only on the initial interpretations $\{l_{j,0}\}$ and the interpretation influence $\pi^{A'}$, through \eqref{eq:lstar}.
    
    \item[\emph{(ii)}] The cells do not overlap, and every point outside the walls
    
    \begin{equation*}
    \mathcal W:=\bigcup_{\theta\neq\theta'}\Bigl\{p\in\Delta(S):\ \textstyle\sum_sp(s)\,\nu_{\theta,\theta'}(s)=0\Bigr\},
    \end{equation*}
    
    where two hypotheses fit $p$ equally well, belongs to exactly one cell. If the consensus interpretations of different hypotheses are distinct, the walls are pieces of hyperplanes and have zero $(|S|-1)$-dimensional Lebesgue measure.
    
    \item[\emph{(iii)}] The data $\{f_i\}$ and the belief influence $\pi^A$ do not affect the cells. They only determine where the point $\bar f=\sum_i\pi^A_i f_i$ falls.
\end{enumerate}
\end{proposition}

\begin{remark}[Shape of the cells]\label{rem:cells}
Each cell is relatively open in $\Delta(S)$. It is either empty or the intersection of $\Delta(S)$ with $|\Theta|-1$ open half-spaces, and it is empty whenever its hypothesis shares its consensus interpretation with another hypothesis. A nonempty cell is never closed, so it is not a polytope. Finally, a cell may contain points on the wall between two other hypotheses that both fit worse than its own.
\end{remark}

This geometry makes the division of labor between the two channels visible. Apart from the walls, the simplex is split into decision cells, one per hypothesis, and the population comes to believe the hypothesis whose cell contains the pooled data $\bar f$. The interpretation channel draws the cells. Their boundaries move only if the initial interpretations change or if influence in the interpretation network shifts toward different agents. The belief channel and the data do not touch the boundaries, they only decide where the point $\bar f$ falls. Within the model, then, an intervention on the interpretation channel reshapes the boundaries themselves, whereas an intervention on the belief channel, or on the data each agent observes, only relocates the point. Figure~\ref{fig:simplex} shows an example of such a partition with three signals.

When the pooled data lie on a wall, two hypotheses fit them equally well, $\Dkl{\bar f}{l^*(\cdot\mid\theta_v)}=\Dkl{\bar f}{l^*(\cdot\mid\theta_w)}$. If these are the best-fitting hypotheses, neither is discarded at an exponential rate, and our results do not determine how beliefs split between them. This is not observational equivalence in the usual sense, because the two hypotheses may still predict different signal distributions.

\begin{figure*}[!ht]
    \centering
    \includegraphics[width=0.9\textwidth]{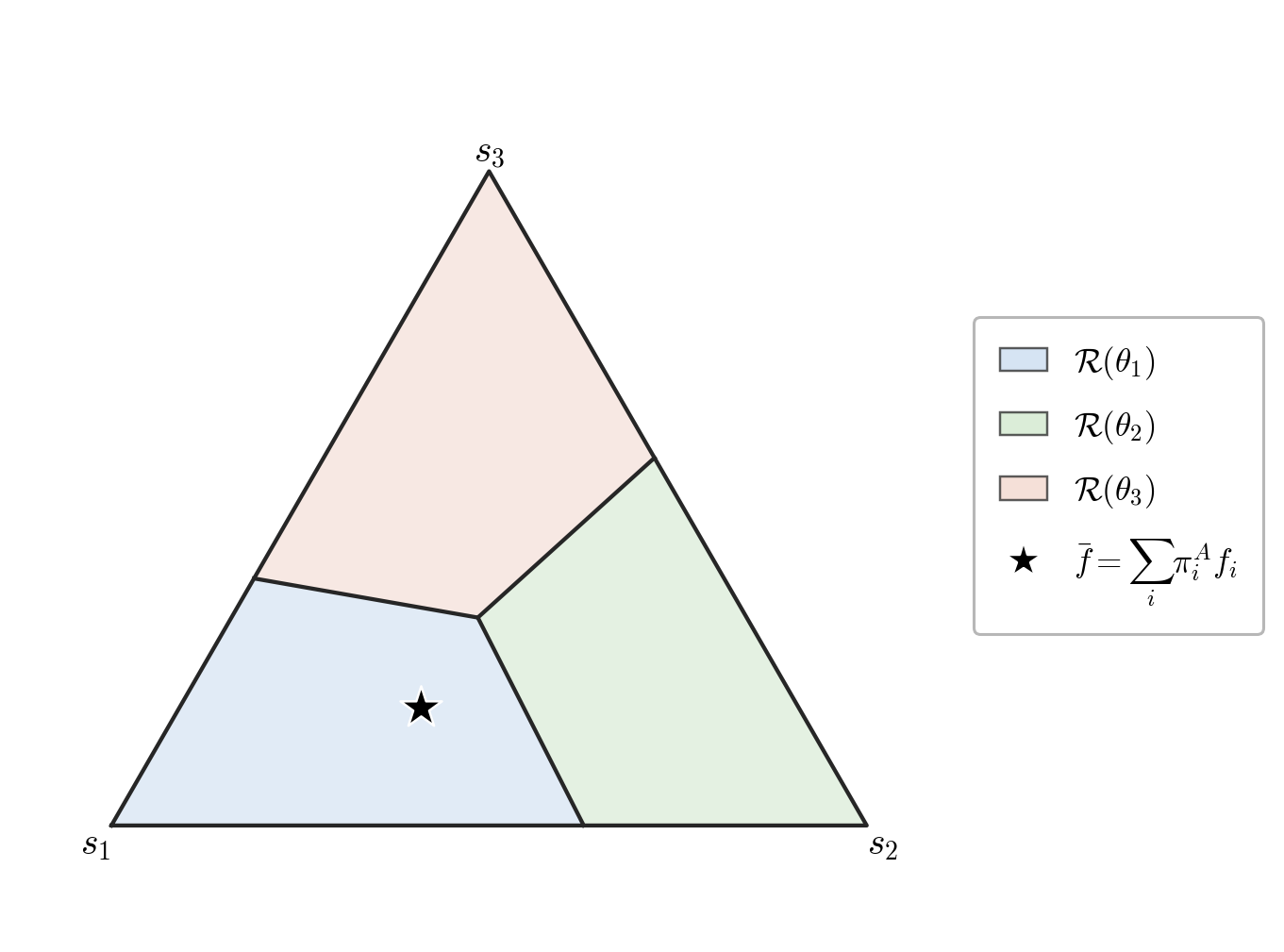}
    \caption{Decision geometry on $\Delta(S)$ with $|S|=3$. Apart from the walls, the simplex is divided into decision cells, and the population learns the hypothesis whose cell contains the pooled data $\bar f$. The cells and the walls between them, here the shaded cell $\mathcal R(\theta^\star)$ in which $\theta^\star$ is learned, are set entirely by the consensus interpretations $l^{*}$, hence by $\{l_{j,0}\}$ and $\pi^{A'}$. The heterogeneous data $\{f_i\}$ and belief influences $\pi^A$ enter only by placing the point $\bar f$ within this fixed geometry.\label{fig:simplex}}
\end{figure*}

\paragraph{In the running example.}
To represent our running example with the geometry presented before, we separate the interpretations from the pooled data. The first will give us the partition of the simplex into sections, each corresponding to one hypothesis. The second will simply put the pooled data in one of the sections, deciding which hypothesis the population ends up believing.

With only two signals, every possible data distribution is described by a single number, the fraction $p$ of severe cases that occur among the unvaccinated, so the simplex $\Delta(S)$ is just the segment $p\in[0,1]$. Any reading of the two hypotheses splits this segment at a single point, the wall $p^*$, where the data fit both hypotheses equally well. To the right of the wall ($p>p^*$) the reading favors $\theta_E$, and to the left ($p<p^*$) it favors $\theta_{NE}$. The position of the wall therefore measures how demanding a standard is: it is the portion of severe cases among the unvaccinated that is needed before the vaccine is judged effective. It is obtained by setting the averaged log-likelihood ratio in \eqref{eq:concrete-llr} to zero, and it depends only on the reading, not on any data.

Figure~\ref{fig:simplex_running_example}(a) applies this to each community in isolation and to the population as a whole. Community~1's moderate standard places its wall at $p^*=0.5$. Its data, with $30\%$ of severe cases among the unvaccinated, fall well to the left, so on its own it would believe the vaccine is not effective. Community~2's demanding standard pushes its wall to $p^*\approx0.64$, but its data ($70\%$) still fall to the right, so on its own it would believe the vaccine is effective. When the communities communicate, the two channels play separate roles. The interpretation channel sets the wall: since Community~2 dominates it, the consensus wall moves close to Community~2's, at $p^*\approx0.63$. The belief channel sets the point being classified, the pooled data $\bar f$, a weighted average of the two communities' observations with weights $\pi^A=(3/8,5/8)$, which gives $\bar f(\mathsf u)=0.55$. No community observes this mixture. It is the data that the population, as a whole, effectively weighs. Since $0.55<0.63$, the point falls on the $\theta_{NE}$ side, and the population comes to believe the vaccine is not effective, as found in Section~\ref{sec:beliefs}. Community~2's influence in the belief channel pulls the point toward its own data, but its influence in the interpretation channel pulls the wall even further to the right.

Figure~\ref{fig:simplex_running_example}(b) repeats this classification for every possible distribution of influence, keeping the initial interpretations and the data fixed. Since the outcome depends on the networks only through their stationary distributions, and with two agents each distribution is fixed by Community~2's weight, every pair of networks $(A,A')$ corresponds to a point in the square: Community~2's influence over interpretations, $\pi^{A'}_2$, on the horizontal axis, and its influence over beliefs, $\pi^A_2$, on the vertical one. Moving right makes the shared standard more demanding, which shifts the wall from $0.5$ toward $0.64$. Moving up gives more weight to Community~2's data, which shifts the pooled data from $0.3$ toward $0.7$. The population believes in $\theta_E$ above the black curve, where the pooled data lie to the right of the wall, and in $\theta_{NE}$ below it. On the curve itself the two hypotheses fit the pooled data equally well, i.e., neither is then discarded at an exponential rate. The running example sits at $(0.9,0.625)$, below the curve. With the same interpretation network as the example studied, the population would need $\pi^A_2>0.826$ to believe the vaccine is effective. More generally, because the wall never falls below $0.5$ nor rises above about $0.64$, the population believes in $\theta_{NE}$ whenever Community 2's belief influence is $\pi^A_2<0.5$, and in $\theta_E$ whenever $\pi^A_2>0.851$, independently of who dominates the interpretation channel. Only between these two values does the interpretation channel decide the outcome.

\begin{figure*}[th!]
    \centering
    \includegraphics[width=0.9\textwidth]{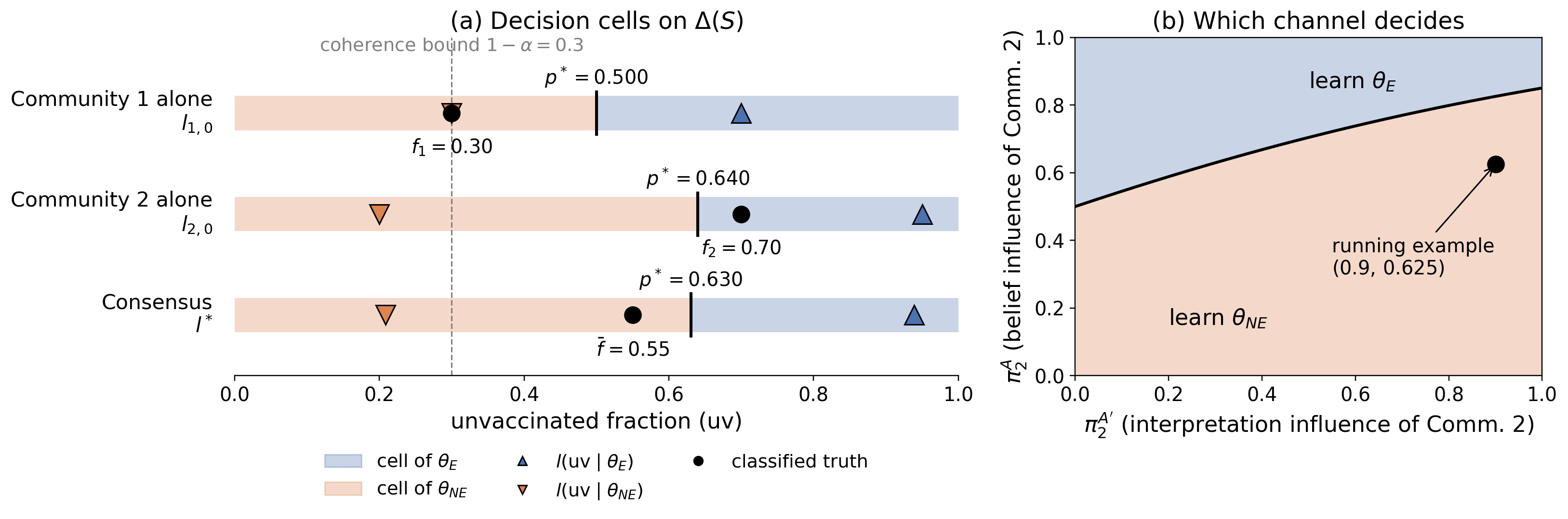}
    \caption{\textbf{(a)} Decision cells on $\Delta(S)$, parametrized by the fraction $p$ of severe cases occurring among unvaccinated individuals. Each row shows the partition induced by a likelihood model: Community~1's initial model $l_{1,0}$, Community~2's initial model $l_{2,0}$, and the consensus model $l^*$ obtained with $\pi^{A'}=(0.1,0.9)$. The shaded intervals are the cells in which $\theta_E$ (blue) or $\theta_{NE}$ (orange) is learned, and the vertical bar marks the wall $p^*$, where both hypotheses fit the data equally well in KL divergence. Triangles mark each model's prediction $l(\mathsf{u}\mid\theta)$ for each hypothesis, and the black dot is the data being classified: $f_1$ and $f_2$ for isolated communities, and the pooled data $\bar f=\sum_i\pi^A_i f_i$, with $\pi^A=(3/8,5/8)$, for the interacting population. The dashed line is the coherence bound $1-\alpha=0.3$ implied by the vaccination coverage $\alpha=0.7$. In isolation, Community~1 learns $\theta_{NE}$ and Community~2 learns $\theta_E$. Because Community~2 is influential in the interpretation channel, the consensus wall moves to $p^*\approx0.630$, and $\bar f=0.55$ falls in the cell of $\theta_{NE}$. 
    \textbf{(b)} Collective outcome as a function of Community~2's interpretation influence $\pi^{A'}_2$ and belief influence $\pi^A_2$. The black curve, $\pi^A_2=(p^*(\pi^{A'}_2)-0.3)/0.4$, is the set of weights for which $\bar f$ lies exactly on the wall. It rises from $0.5$ at $\pi^{A'}_2=0$ to $\approx0.851$ at $\pi^{A'}_2=1$. The interpretation channel moves the wall and the belief channel moves the classified point, so the two weights enter the outcome in separate roles (Proposition~\ref{prop:geometry}). The black dot marks the running example, $(\pi^{A'}_2,\pi^A_2)=(0.9,0.625)$.\label{fig:simplex_running_example}}
\end{figure*}

\section{Sharing Interpretations Versus Sharing Beliefs Only}
\label{sec:comparison}

We now confront the model with the classical benchmark in which agents pool beliefs over the same matrix $A$ but never exchange interpretations (equivalently $A'=I$). Each agent then keeps its own fixed model $l_{i,0}$ over time. We now present the consequences of sharing interpretations in three different scenarios. The first is the misspecified scenario, the most general case analyzed so far, in which there exists no true hypothesis and agents' likelihoods need not match the signal distributions. The second is the well-specified scenario where we specify a true hypothesis, unknown to the population, whose initial interpretations of that true hypothesis actually match the signal distribution. Finally, we keep a ground truth but relax well-specification, so that each interpretation needs only to be coherent with its hypothesis, and some agents may misread their own data even though those data are generated under the ground truth.

\subsection{The Misspecified Benchmark}
\label{subsec:misspecified_benchmark}

In the misspecified scenario, the one studied before, we assume that there is no true hypothesis. Therefore, we ask which hypotheses the network will end up believing. Moreover, the signal distribution need not match any agent's likelihood for any hypothesis ($f_i$ may differ from $l_{i,0}(\cdot\mid\theta)$, $\forall\theta\in\Theta$). 

In the classical model where no interpretations are shared, the theory of \cite{nedic2017fast} defines the group confidence and the optimal hypotheses on which the network concentrates (under Assumption~\ref{ass:standing}) as follows.

\begin{definition}[No-sharing group confidence]\label{def:CNS}
\begin{equation}
    C^{NS}(\theta):=-\sum_{i=1}^n\pi^A_i\,\Dkl{f_i}{l_{i,0}(\cdot\mid\theta)},
    \qquad
    \Theta^{NS}:=\argmax_{\theta\in\Theta}C^{NS}(\theta).
    \label{eq:CNS}
\end{equation}
\end{definition}

\begin{remark}[Fixed likelihood profiles and the $A'=I$ benchmark consistency]\label{rem:fixed-profile}
When $A'=I$ each agent retains its own model, $l_{i,t}=l_{i,0}$ for all $t$, so $A'$ is not primitive and Theorem~\ref{thm:l-consensus} does not apply. The belief-consistency analysis of Section~\ref{sec:beliefs} nonetheless carries over to fixed likelihoods (Appendix~\ref{app:fixed-profile}). Hence every $\mu_{i,t}$ concentrates on $\Theta^{NS}$ exactly as in Theorem~\ref{thm:consistency}, with $C^{NS}$ in place of $C$. For a doubly-stochastic $A$, where $\pi^A=\one/n$, this reduces to the consistency guarantee of \cite{nedic2017fast} on a static network.
\end{remark}

Comparing this definition with the one in \ref{eq:confidence} and \ref{eq:optimal-set} we see that the two functionals differ in a single yet consequential way. $C^{NS}$ evaluates the individual confidence of each agent's data $f_i$ against agent $i$'s own interpretation on a given hypothesis, weighting each agent's confidence. $C$ instead defines the confidence on a given hypothesis as how well the consensus interpretation explains each agent's data. Then, individual confidence may vary when the data is compared to a different interpretation, and with this, the collective optimal hypotheses set.

\begin{theorem}[Likelihood sharing can overturn belief-only learning]
\label{thm:reversal}
There exist $(A,A')$ satisfying (A1), data $\{f_i\}$ satisfying (A3), full-support initial profiles $\{l_{i,0}\}$ and initial beliefs satisfying (A2), and hypotheses $\theta_1\neq\theta_2$, such that

\begin{equation}
    C^{NS}(\theta_1)>C^{NS}(\theta_2) \qquad\text{while}\qquad C(\theta_1)<C(\theta_2).
    \label{eq:reversal}
\end{equation}

Moreover, for any configuration satisfying Assumption~\ref{ass:standing}, writing $\mathfrak D_{ij}(\theta):=\Dkl{f_i}{l_{j,0}(\cdot\mid\theta)}\in[0,\infty)$, the reversal \eqref{eq:reversal} holds if and only if
\begin{equation}
  \underbrace{\sum_{i=1}^n \pi^A_i\bigl(\mathfrak D_{ii}(\theta_2)-\mathfrak D_{ii}(\theta_1)\bigr)}_{\text{matched (belief-only) ranking}}
  \ >\ 0\ >\
  \underbrace{\pi^{A\top}\bigl(\mathfrak{D}(\theta_2)-\mathfrak{D}(\theta_1)\bigr)\pi^{A'}
  + \bigl(\log \Psi^{\theta_2}_\infty-\log \Psi^{\theta_1}_\infty\bigr)}_{\text{pooled (interpretation-sharing) ranking}} .
  \label{eq:reversal-box}
\end{equation}

Consequently, when $\Theta^{NS}=\{\theta_1\}$ and $\Theta^{*}=\{\theta_2\}$, Corollary~\ref{cor:singleton} gives $\mu_{i,t}(\theta_2)\to1$ a.s.\ for every $i$ under likelihood sharing, whereas the benchmark dynamics over the same belief network $A$ concentrate on $\theta_1$ ($\mu_{i,t}(\theta_1)\to1$).
\end{theorem}

The two sides of \eqref{eq:reversal-box} are exactly the two group-confidence comparisons. The left expression is the diagonal (matched) contrast: it sums each agent's own-model KL advantage of $\theta_1$ over $\theta_2$, weighted by $\pi^A$, and its sign is the belief-only ranking. The right expression is the full bilinear form in $\pi^A$ and $\pi^{A'}$, plus the diversity-bonus gap ($\log \Psi^{\theta_2}_\infty-\log \Psi^{\theta_1}_\infty$) between the two hypotheses, and its sign is the pooled ranking. A reversal is precisely a sign disagreement between them.

\paragraph{In the running example.}

The running example is an instance of the reversal in Theorem~\ref{thm:reversal}. To see this, consider the same two communities and the same belief network $A$, but suppose they share only their beliefs. Each community then keeps its own standard and compares its own data with it, and the no-sharing group confidence of Definition~\ref{def:CNS} weights these comparisons by belief influence, as before. Community~1 reads its data ($30\%$ of severe cases among the unvaccinated) with its moderate standard, which favors $\theta_{NE}$, but only by a moderate margin. Community~2 reads its data ($70\%$) with its demanding standard, which favors $\theta_E$ by a slightly smaller margin. Since Community~2's comparison counts $5/8$ against $3/8$, the scores are

\begin{equation*}
    C^{NS}(\theta_E)\approx-0.33,\qquad C^{NS}(\theta_{NE})\approx-0.36,
\end{equation*}

so $\Theta^{NS}=\{\theta_E\}$ and the population comes to believe that the vaccine is effective. The two scores are close, so by Theorem~\ref{thm:consistency} the log-odds against $\theta_{NE}$ grow on average by only about $0.03$ per period, roughly ten times more slowly than the population discards $\theta_E$ when interpretations are shared.

Sharing interpretations therefore reverses the outcome: the same population, over the same belief network and with the same data, believes that the vaccine is effective if it shares only beliefs, and that it is not if it also shares interpretations. The reversal comes almost entirely from how Community~1 reads its own data. Community~2's data favor $\theta_E$ by a similar margin whether they are read with its own standard or with the consensus, which is close to it. Community~1's data, however, favor $\theta_{NE}$ by a margin about four times larger once they are read with the demanding consensus standard instead of Community~1's moderate one. Community~1 does not change its belief because it trusts Community~2 more, it changes because it has adopted a standard under which its own data look much more clearly like an ineffective vaccine. In the terms of Theorem~\ref{thm:reversal}, the belief-only ranking and the ranking under sharing have opposite signs.

Figure~\ref{fig:reversal} shows one simulated run of each scenario. Beliefs fluctuate with the random signals, but in each run they eventually settle on a different hypothesis, as the theory predicts.

\begin{figure*}[t]
    \centering
    \includegraphics[width=0.95\textwidth]{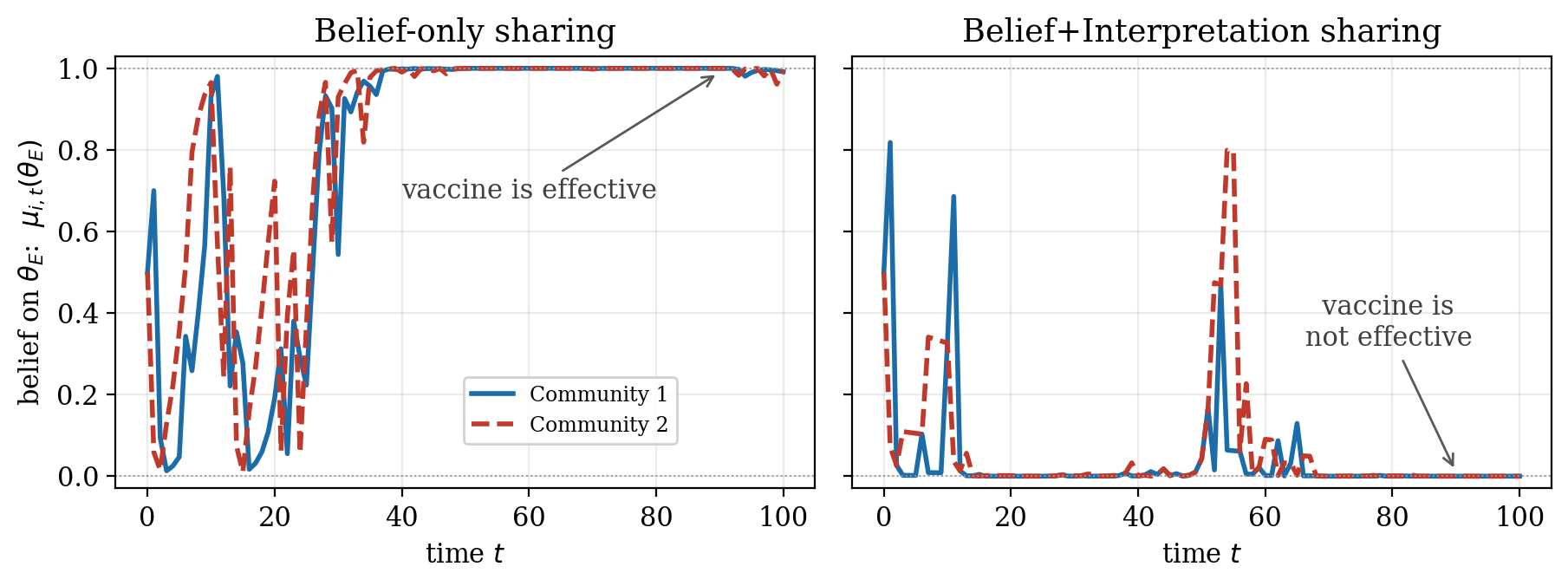}
    \caption{Simulation of the belief dynamics in the running example. Both panels show a run with $\pi^A=(\frac{3}{8}, \frac{5}{8})$, $\pi^{A'}=(\frac{1}{10}, \frac{9}{10})$, the truths $f_1=(0.3,0.7)$ and $f_2=(0.7,0.3)$, and the initial interpretations as in Table~\ref{tab:ex-frameworks}. The two figures differ in whether interpretations are shared. Each curve is a single representative run.}
    \label{fig:reversal}
\end{figure*}

\subsection{The Well-Specified Benchmark}
\label{subsec:wellspecified}

We now analyze the dynamics in the model that the social-learning literature usually adopts. We assume a common true hypothesis $\theta^\star$ under which every agent is initially correctly specified, i.e. would read the data exactly:

\begin{equation}
  f_i=l_{i,0}(\cdot\mid\theta^\star)\qquad\text{for every }i\in N.
  \label{eq:wellspec}
\end{equation}

Under \eqref{eq:wellspec} the classical model ($A'=I$) learns $\theta^\star$ whenever the network is collectively informative, that is, whenever:

\begin{equation}
  \text{for every }\theta\neq\theta^\star\text{ there is }i\in N\text{ with }
  l_{i,0}(\cdot\mid\theta)\neq f_i .
  \label{eq:collid}
\end{equation}

This means that no single agent needs to distinguish all hypotheses on its own, but the agents' interpretations, taken together, must be able to. Condition \eqref{eq:collid} is an assumption on the initial interpretations and it does not follow from Assumption~\ref{ass:standing}. Strong connectivity is what allows every agent's evidence to influence every other agent's beliefs.

\subsubsection{The no-sharing benchmark recovers the truth}

\begin{proposition}[Belief-only sharing learns $\theta^\star$]\label{prop:ws-benchmark}
Under \eqref{eq:wellspec}, \eqref{eq:collid}, and $A'=I$,

\begin{equation*}
    C^{NS}(\theta^\star)=0=\max_{\theta\in\Theta}C^{NS}(\theta), \qquad \Theta^{NS}=\{\theta^\star\},
\end{equation*}

and hence $\mu_{i,t}(\theta^\star)\to1$ a.s.\ for every $i\in N$.
\end{proposition}

The benchmark achieves the best possible score $C^{NS}(\theta^\star)=0$ because each agent's own model fits its own data perfectly. Keeping the models private fully exploits every agent's locally correct interpretation exactly.

\subsubsection{The self-inflicted misspecification created by sharing}

Under \eqref{eq:wellspec}, when agents share their readings, the consensus interpretation of the true hypothesis is the $\pi^{A'}$-logarithmic pool of the truths themselves,

\begin{equation}
    l^*(\cdot\mid\theta^\star) =\frac{\prod_{j}f_j(\cdot)^{\pi^{A'}_j}}{\Psi^{\theta^\star}_\infty}
    \label{eq:lstar-truth}
\end{equation}

In contrast, the pooled truth is the pooled data defined in Proposition~\ref{prop:effective_single_agent}. The two objects (interpretation and pooled truth) are two different aggregates of the same family $\{f_j\}$. They differ both in the averaging rule (geometric versus arithmetic) and in the weights ($\pi^{A'}$ versus $\pi^A$). Consequently the divergence

\begin{equation}
  \varepsilon^\star:=\Dkl{\bar f}{l^*(\cdot\mid\theta^\star)}\ \ge\ 0
  \label{eq:eps-star}
\end{equation}

is in general strictly positive (see the proof of Proposition~\ref{prop:ws-hurt} for an explicit instance). For example, it vanishes under homogeneity $f_1=\dots=f_n$. The quantity $\varepsilon^\star$ is a misspecification error that sharing interpretations introduces at the true hypothesis. 

\begin{proposition}[Exact survival condition under sharing]\label{prop:ws-survival}
Under \eqref{eq:wellspec}, sharing burdens the true hypothesis with the self-inflicted misspecification $\varepsilon^\star$ of \eqref{eq:eps-star}. The truth survives exactly when no rival hypothesis fits the pooled truth better than this misspecification error:

\begin{equation}
    \theta^\star\in\Theta^* \iff \varepsilon^\star\le\Dkl{\bar f}{l^*(\cdot\mid\theta)} \quad\text{for every }\theta\in\Theta,
    \label{eq:ws-survival-eps}
\end{equation}

and the population learns it uniquely, $\Theta^*=\{\theta^\star\}$, iff the inequality is strict for every rival $\theta\neq\theta^\star$.
\end{proposition}

\subsubsection{Two regimes: safe under agreement, harmful under heterogeneity}

The behavior of the interpretation channel now splits according to whether the agents' signal distributions agree.

\begin{proposition}[Homogeneous truths: sharing is safe]\label{prop:ws-homog}
Under \eqref{eq:wellspec}, if $f_1=\dots=f_n=f$ then $l^*(\cdot\mid\theta^\star)=f=\bar f$, so $\varepsilon^\star=0$ and $\theta^\star\in\Theta^*$. Generically $\Theta^*=\{\theta^\star\}$ and sharing learns $\theta^\star$. Uniqueness can fail only on the measure-zero set of configurations for which some $\theta\neq\theta^\star$ satisfies $l^*(\cdot\mid\theta)=f$.
\end{proposition}

\begin{proposition}[Sharing never rescues in the well-specified case]\label{prop:ws-norescue}
Assume \eqref{eq:wellspec}, let $A$ and $A'$ satisfy (A1), and fix initial beliefs satisfying (A2). If interpretation sharing makes the population learn the true state ($\mu_{i,t}(\theta^\star)\to1$ a.s.\ for every $i\in N$), then belief-only learning, the same dynamics with $A'$ replaced by $I$, also learns the truth ($\mu^{NS}_{i,t}(\theta^\star)\to1$ a.s.\ for every $i\in N$). Equivalently, no configuration exists in which sharing leads every agent to learn $\theta^\star$ while belief-only learning does not. Sharing can only preserve or destroy the benchmark's learning of $\theta^\star$, never create it.
\end{proposition}

\begin{proposition}[Heterogeneous truths: sharing can strictly hurt, robustly]
\label{prop:ws-hurt}
There is a nonempty open set of configurations satisfying \eqref{eq:wellspec} with heterogeneous truths such that $\Theta^{NS}=\{\theta^\star\}$ while $\theta^\star\notin\Theta^*$. In other words, belief-only sharing learns the truth ($\mu_{i,t}(\theta^\star)\to1$ a.s.), whereas likelihood sharing drives every agent toward false hypotheses.
\end{proposition}

\paragraph{In the running example.}

We now assume a ground truth: the vaccine is effective, and each community's standard for an effective vaccine matches exactly what it observes, $f_i=l_{i,0}(\cdot\mid\theta_E)$ with the values of Table \ref{tab:ex-frameworks}. Community~1 then observes $30\%$ of severe cases among the vaccinated, which corresponds to about $82\%$ prevention, and Community~2 observes $5\%$, about $98\%$ prevention. The vaccine works in both communities, though with different strength, for instance because of different variants or risk profiles. The standards for an ineffective vaccine and both networks are those used so far.

If the communities share only their beliefs, each one compares its data with a standard of effectiveness that fits them perfectly, so the true hypothesis gets the best possible score, $C^{NS}(\theta_E)=0$, against $C^{NS}(\theta_{NE})\approx-0.97$, and the population believes the vaccine is effective (Proposition \ref{prop:ws-benchmark}). If they also share interpretations, both communities end up with the consensus standard, close to Community~2's: an effective vaccine should leave about $6\%$ of severe cases among the vaccinated. This standard still fits Community~2's data almost perfectly, but no longer fits Community~1's $30\%$. The truth now carries a cost that sharing itself introduced, the misspecification $\varepsilon^\star$ of \eqref{eq:eps-star}, and its score drops to $C(\theta_E)\approx-0.10$. The cost is not large enough to change the outcome, since Community~1's data are far from the demanding picture of an effective vaccine, but they are much further from the consensus picture of an ineffective one ($79\%$ of severe cases vaccinated), so $C(\theta_{NE})\approx-1.02$ and the population still believes the vaccine is effective. Sharing has weakened the support for the truth without overturning it.

However, the cost can overturn it. For example, suppose instead that Community~1 holds a more moderate standard, under which a vaccine that leaves half of the severe cases among the vaccinated already counts as effective, and that this is exactly what it observes, $f_1=l_{1,0}(\cdot\mid\theta_E)=(0.5,0.5)$. With a coverage of $70\%$, this is still an effective vaccine, preventing about $57\%$ of severe disease, so $\theta_E$ remains the truth in both communities. Keep Community 2 dominant in the interpretation channel, $\pi^{A'}=(0.1,0.9)$, but now let Community~1 dominate the belief channel, $\pi^A=(0.9,0.1)$. With belief-only sharing, each community still reads its own data perfectly, and the population believes the vaccine is effective. With interpretation sharing, observe Figure \ref{fig:well_specified}, the consensus standard remains demanding, so the wall stays near Community~2's, at $p^*\approx0.62$. The pooled data, however, are now dominated by Community~1's observations, with only $54.5\%$ of severe cases among the unvaccinated. They fall on the $\theta_{NE}$ side of the wall, and the population comes to believe the vaccine is not effective. Under the shared standard, the moderate protection that Community~1 observes looks more like a failure than like the near-perfect protection the standard expects. Each community reads its own data correctly before communicating, and it is the consensus between them that leads the population to discard the truth.

\begin{figure*}[th!]
    \centering
    \includegraphics[width=0.8\textwidth]{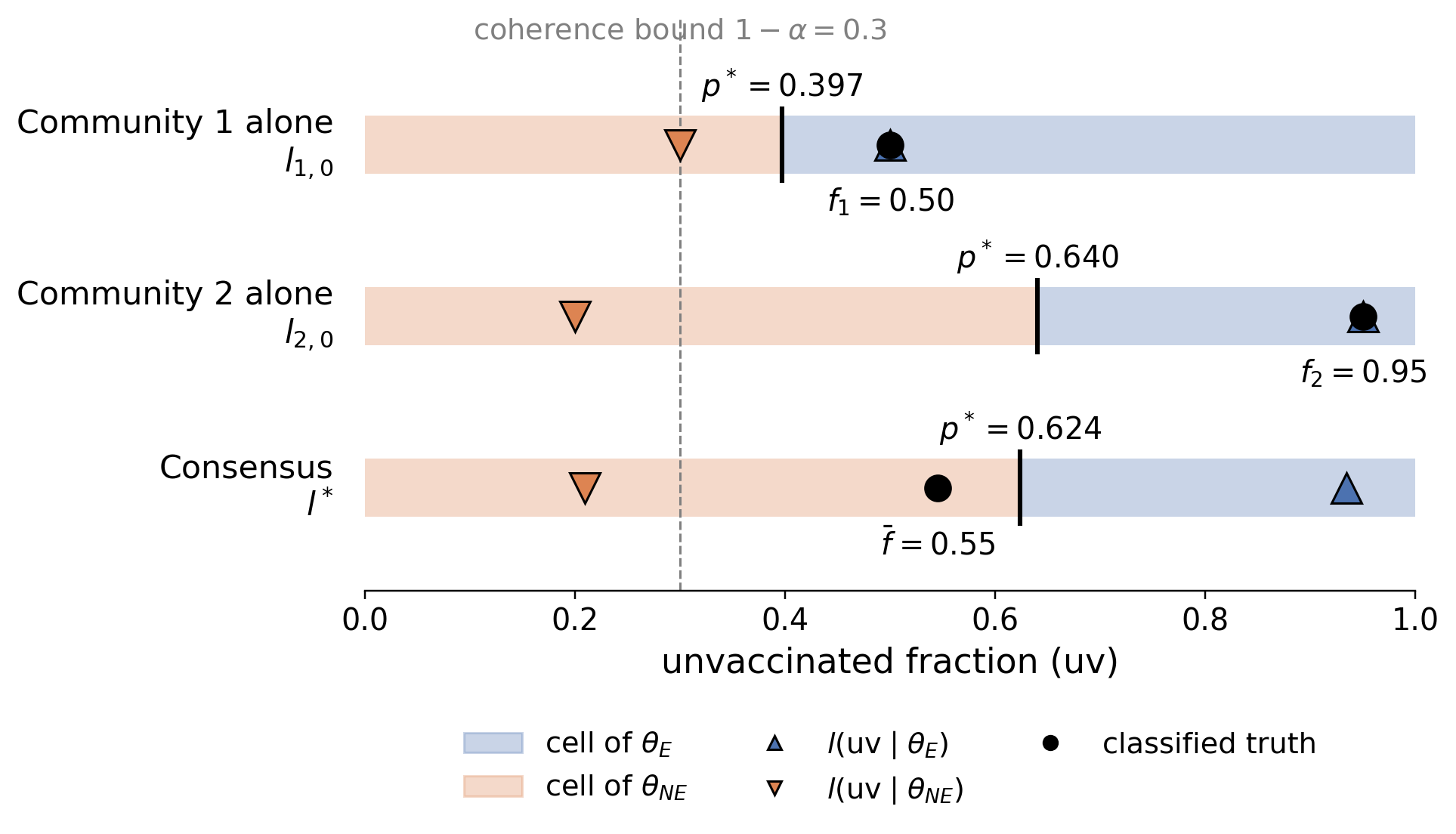}
    \caption{Decision cells on $\Delta(S)$, parametrized by the fraction $p$ of severe cases occurring among unvaccinated individuals. Each row shows the partition induced by a likelihood model: Community~1's initial model $l_{1,0}$, Community~2's initial model $l_{2,0}$, and the consensus model $l^*$ obtained with $\pi^{A'}=(0.1,0.9)$. The shaded intervals are the cells in which $\theta_E$ (blue) or $\theta_{NE}$ (orange) is learned, and the vertical bar marks the wall $p^*$, where both hypotheses fit the data equally well in KL divergence. Triangles mark each model's prediction $l(\mathsf{u}\mid\theta)$ for each hypothesis, and the black dot is the truth being classified: $f_1$ and $f_2$ for isolated communities, and the pooled truth $\bar f=\sum_i\pi^A_i f_i$, with $\pi^A=(0.9,0.1)$, for the interacting population. The dashed line is the coherence bound $1-\alpha=0.3$ implied by the vaccination coverage $\alpha=0.7$.\label{fig:well_specified}}
\end{figure*}

Agreeing on a single common standard, such as one national reading of hospital data, therefore has a price when communities observe different data, even if every community starts with an accurate reading. The common standard cannot match every community's data at once, and the resulting error can be enough to make the whole population abandon the truth. In this well-specified setting, sharing interpretations cannot help, since belief-only sharing already learns the truth (Proposition \ref{prop:ws-norescue}), but it can hurt (Proposition \ref{prop:ws-hurt}). Note that no data are corrupted, the error comes from spreading a reading of the evidence, not false evidence.

\subsection{Coherent Misreading Interpretations: Sharing Can Rescue or Corrupt}
\label{subsec:coherent_sharing_good_or_bad}

In Section~\ref{subsec:wellspecified}, every agent's initial interpretation of the true hypothesis matches the agent's data. We now relax this and simply require that each agent interpret each hypothesis coherently and that a ground truth exist. Coherence is determined by the specific case under study and the definition of the hypotheses. Each hypothesis constrains a set of logically compatible signal distributions. We then assume that every agent holds a coherent interpretation and that all agents' data distributions fall within the coherent set of exactly one hypothesis.

\begin{definition}[Coherence and ground truth]\label{def:coherent}
For each $\theta\in\Theta$, let $\mathcal T(\theta)\subseteq\Delta(S)$ be the set of signal distributions that are logically compatible with $\theta$ (its coherence region). The regions may overlap.

\begin{enumerate}
    \item[(a)] Interpretations are coherent if $l_{j,0}(\cdot\mid\theta)\in\mathcal T(\theta)$ for all $j\in N$ and $\theta\in\Theta$.
    \item[(b)] $\theta^\star\in\Theta$ is the ground truth if $f_i\in\mathcal T(\theta^\star)\setminus\bigcup_{\theta\neq\theta^\star}\mathcal T(\theta)$ for all $i\in N$. That is, every agent's data are compatible with $\theta^\star$ and with no other hypothesis.
\end{enumerate}
\end{definition}

\begin{lemma}[Uniqueness of the ground truth]\label{lem:gt-unique}
If a ground truth exists, it is unique.
\end{lemma}

A ground truth does not require $f_i=l_{i,0}(\cdot\mid\theta^\star)$, so agents may be misspecified. Moreover, coherence constrains each model $l_{j,0}(\cdot\mid\theta)$ separately, whereas learning depends on how the models of different hypotheses compare on the data. A coherent agent can therefore still read incorrectly her data.

\begin{definition}[Accurate and misreading agents]\label{def:readers}
For $j\in N$, $p\in\Delta(S)$ and $\theta\neq\theta^\star$, let

\begin{equation*}
    \Delta_j(\theta;p):=\sum_{s\in S}p(s)\log\frac{l_{j,0}(s\mid\theta^\star)}{l_{j,0}(s\mid\theta)} =\Dkl{p}{l_{j,0}(\cdot\mid\theta)}-\Dkl{p}{l_{j,0}(\cdot\mid\theta^\star)},
\end{equation*}

whose properties are given in Lemma~\ref{lem:Delta}. Then, we define:

\begin{itemize}
    \item Agent $j$ reads $p$ correctly if $\Delta_j(\theta;p)>0$ for every $\theta\neq\theta^\star$.
    \item Agent $j$ misreads $p$ if $\Delta_j(\theta;p)<0$ for some $\theta\neq\theta^\star$.
    \item Agent $i$ is accurate if it reads its own data $f_i$ correctly, and misreading if it misreads $f_i$.
\end{itemize}
\end{definition}

Thus $\Delta_j(\theta;p)>0$ means that agent $j$'s interpretations rank $\theta^\star$ above $\theta$ on data $p$. The borderline case ($\Delta_i(\theta;f_i)\ge0$ for all $\theta\neq\theta^\star$, with equality for some $\theta$) is neither accurate nor misreading and plays no role below.

These notions describe what an agent learns in isolation. For a single agent, Proposition~\ref{prop:effective_single_agent} gives $C(\theta^\star)-C(\theta)=\Delta_i(\theta;f_i)$. An accurate agent therefore learns $\theta^\star$, and a misreading agent discards it. Under \eqref{eq:wellspec} no agent misreads. An agent's interpretations can also be applied to another agent's data, and the main result asks whether agent $j$ reads $f_i$ correctly.

We write $\mu^{NS}_{k,t}$ and $\mu_{k,t}$ for agent $k$'s beliefs under belief-only learning and under interpretation sharing, respectively.

\begin{theorem}[Sharing interpretations can rescue or corrupt]\label{thm:rescue-corrupt}
Let interpretations be coherent, with ground truth $\theta^\star$, and fix the data $\{f_k\}_{k\in N}$ and the initial interpretations $\{l_{k,0}\}_{k\in N}$.

\begin{enumerate}
    \item[\emph{(i)}] \emph{Rescue:} Suppose some agent $i$ misreads its own data $f_i$ but some agent $j\neq i$ would read $f_i$ correctly. Then there exist networks $A,A'$ satisfying (A1) such that, for every agent $k\in N$,
    
    \begin{equation*}
      \mu^{NS}_{k,t}(\theta^\star)\to0\qquad\text{but}\qquad\mu_{k,t}(\theta^\star)\to1\qquad\text{a.s.}
    \end{equation*}
    
    \item[\emph{(ii)}] \emph{Corruption:} Suppose some agent $i$ is accurate and some agent $j\neq i$ would misread $f_i$. Then there exist networks $A,A'$ satisfying (A1) such that, for every agent $k\in N$,
    
    \begin{equation*}
        \mu^{NS}_{k,t}(\theta^\star)\to1\qquad\text{but}\qquad\mu_{k,t}(\theta^\star)\to0\qquad\text{a.s.}
    \end{equation*}
\end{enumerate}
\end{theorem}

\begin{corollary}[Common data]\label{cor:rc-common}
Let interpretations be coherent, with ground truth $\theta^\star$, and let $f_1=\dots=f_n$. If the population contains at least one accurate and at least one misreading agent, then, depending on the networks, sharing interpretations can either rescue the entire population or corrupt it, in the sense of Theorem~\ref{thm:rescue-corrupt}.
\end{corollary}

The implications concern how the whole population learns when part of it misreads her own data. Under belief-only learning, the population learns the true state if $G^{NS}_\theta=\sum_k\pi^A_k\Delta_k(\theta;f_k)>0$ for every $\theta\neq\theta^\star$, and fails to learn it if this quantity is negative for some $\theta$ (Lemma~\ref{lem:Delta} and Remark~\ref{rem:fixed-profile}). Accurate agents must therefore carry enough belief influence, weighted by how sharply each agent's reading separates the hypotheses on its own data. Otherwise, even a single influential misreading agent can lead the entire population to the wrong state. Sharing interpretations can prevent this. An agent who misreads her own data can read another’s data accurately. When such interpretations shape the consensus, the resulting interpretation can read the pooled data correctly, and this can rescue the whole population, including the misreading agents. The opposite can also happen. A misreading interpretation may shape the consensus and corrupt agents who were accurate. If those agents are influential in the belief network, they then lead the population to the wrong state, whereas without sharing interpretations they would have led them to the truth. Therefore, we reframe misinformation as misinterpretation contagion. Note that agents never observe each other's data. Instead, they influence their neighbors to change how they interpret their own data, for better or worse. What the population will learn depends on who is influential in each channel.

\paragraph{In the running example.}

Coherence here comes from the vaccination coverage $\alpha=0.7$, as explained in Section~\ref{subsec:example}. A standard for an effective vaccine is coherent if it expects fewer than $70\%$ of severe cases among the vaccinated, $\mathcal T(\theta_E)=\{p:p(\mathsf v)<0.7\}$, and a standard for an ineffective one is coherent if it expects at least $70\%$, $\mathcal T(\theta_{NE})=\{p:p(\mathsf v)\ge0.7\}$. Both communities' standards in Table~\ref{tab:ex-frameworks} meet these bounds. Suppose now that both communities observe the same data, $f=(0.6,0.4)$, i.e., $40\%$ of severe cases are vaccinated. This is below the coverage, so the data are compatible only with an effective vaccine reading and $\theta_E$ is the ground truth (Definition~\ref{def:coherent}). Objectively, the vaccine prevents about $71\%$ of severe disease.

Coherent standards can still misread these data. Community~1's moderate standard ranks $\theta_E$ above $\theta_{NE}$ on them, $\Delta_1(\theta_{NE};f)\approx0.17>0$, so Community~1 is accurate. Community~2's standard ranks them the other way, $\Delta_2(\theta_{NE};f)\approx-0.17<0$, so it misreads. The reason is how strongly each kind of case counts under its standard. If an effective vaccine should leave only $5\%$ of severe cases among the vaccinated, each vaccinated case is about $16$ times more likely under $\theta_{NE}$ than under $\theta_E$, while each unvaccinated case is only about $5$ times more likely under $\theta_E$. Observing $40\%$ of cases vaccinated, the evidence against effectiveness outweighs the evidence for it, and Community~2 reads a vaccine that prevents most severe disease as a failure. Geometrically, the common data ($60\%$ of severe cases among the unvaccinated) lie between the two communities' walls, $0.5$ and about $0.64$, on the $\theta_E$ side of Community~1's and the $\theta_{NE}$ side of Community~2's.

Since the population contains one accurate and one misreading community, Corollary~\ref{cor:rc-common} says that, depending on the networks, sharing interpretations can either rescue or corrupt it. The two channels act differently here. With belief-only sharing, each community keeps its own standard, and the population follows the community that dominates the belief network, weighted by how strongly its standard separates the two hypotheses on the data. With interpretation sharing, every community reads the same data through the same consensus standard, so the pooled data equal $f$ whatever the belief network, and the outcome depends only on who dominates the interpretation channel (Corollary~\ref{cor:singletruth}), setting the walls in the geometric representation.

\emph{Rescue.} Let Community~2 dominate beliefs and Community~1 dominate interpretations, $\pi^A=(0.1,0.9)$ and $\pi^{A'}=(0.9,0.1)$. With belief-only sharing, the misreading community leads the population, $C^{NS}(\theta_E)-C^{NS}(\theta_{NE})\approx-0.14$, and everyone comes to believe the vaccine is not effective. With interpretation sharing, the consensus standard is close to Community~1's, its wall sits at about $0.52$, and the data fall on the $\theta_E$ side: $C(\theta_E)-C(\theta_{NE})\approx0.16$. Community~2 still leads beliefs, but it now reads its data through a standard under which they support effectiveness, so the whole population, Community~2 included, comes to believe the vaccine is effective.

\emph{Corruption.} Swap the roles, $\pi^A=(0.9,0.1)$ and $\pi^{A'}=(0.1,0.9)$. With belief-only sharing, the accurate community leads the population to the truth, $C^{NS}(\theta_E)-C^{NS}(\theta_{NE})\approx0.14$. With interpretation sharing, the consensus standard is close to Community~2's, its wall moves to about $0.63$, and the data fall on the $\theta_{NE}$ side: $C(\theta_E)-C(\theta_{NE})\approx-0.12$. Community~1 still leads beliefs, but it has adopted a standard under which its own data look like a failure, and it leads the population away from the truth.

Figure~\ref{fig:rescue_corrupt} shows both cases, (a) corruption and (b) rescue. In both, the data and the initial standards are the same, and no community ever sees the other's data. What moves the outcome is which reading of those data spreads through the interpretation channel, and whether the community that leads beliefs has adopted it.

\begin{figure*}[th!]
    \centering
    \begin{subfigure}[t]{0.49\textwidth}
        \centering
        \includegraphics[width=\textwidth]{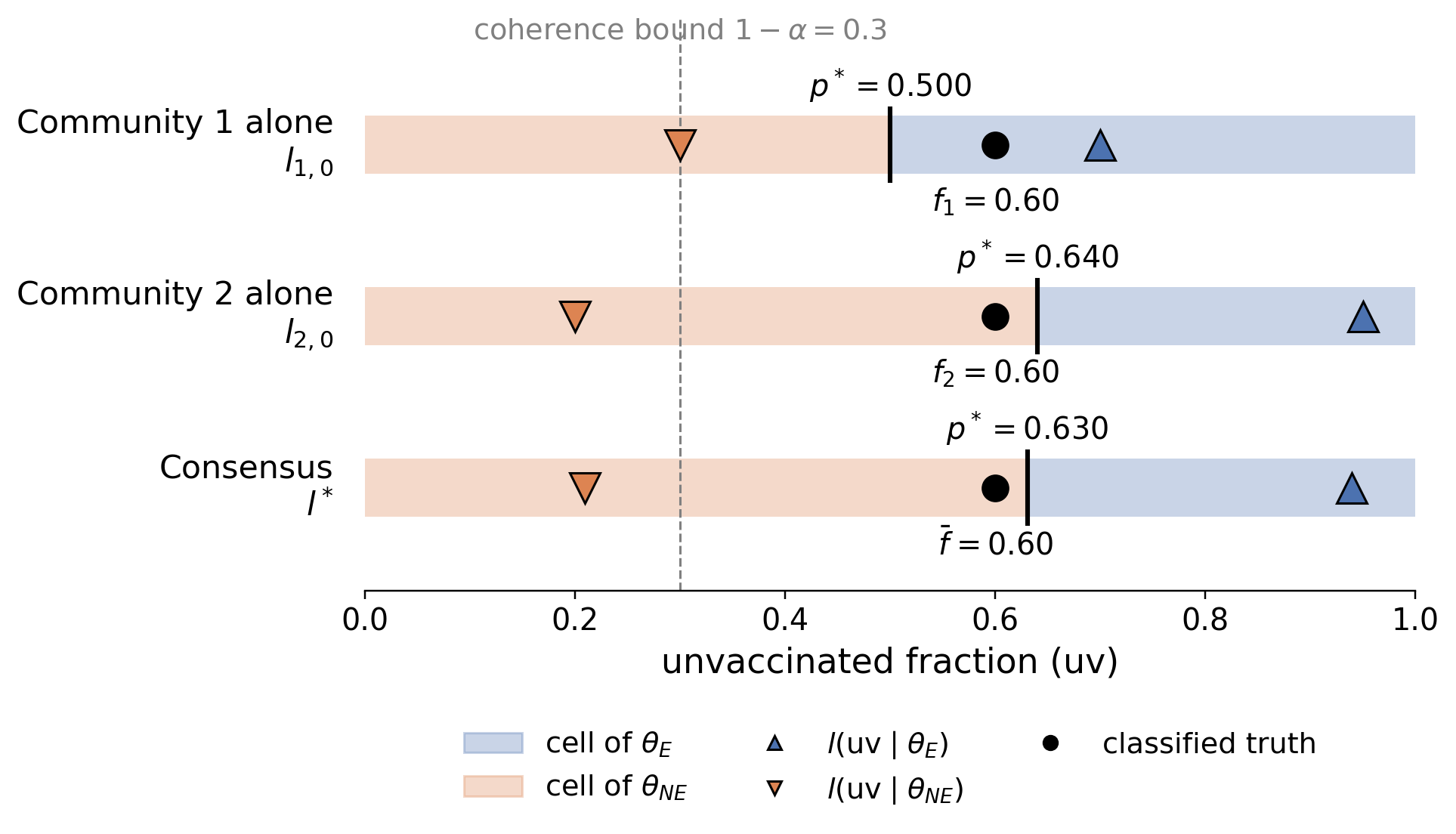}
        \caption{}
        \label{fig:corruption}
    \end{subfigure}
    \hfill
    \begin{subfigure}[t]{0.49\textwidth}
        \centering
        \includegraphics[width=\textwidth]{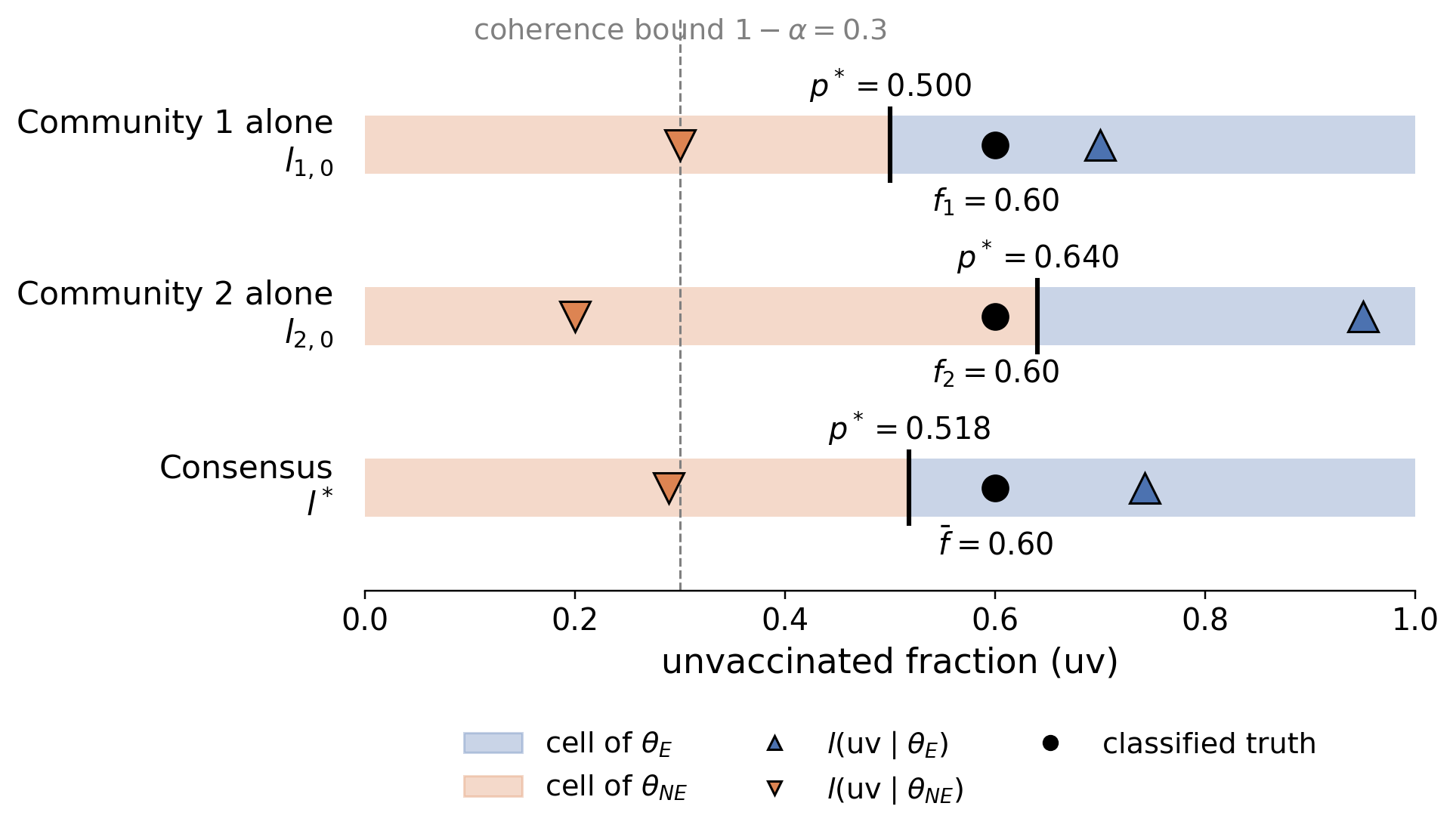}
        \caption{}
        \label{fig:rescue}
    \end{subfigure}

    \caption{Decision cells on $\Delta(S)$, parametrized by the fraction $p$ of severe cases occurring among unvaccinated individuals. Each panel shows the partition induced by a likelihood model: Community~1's initial model $l_{1,0}$, Community~2's initial model $l_{2,0}$, and the consensus model $l^*$ obtained with the corresponding $\pi^{A'}$. The shaded intervals are the cells in which $\theta_E$ (blue) or $\theta_{NE}$ (orange) is learned, and the vertical bar marks the wall $p^*$, where both hypotheses fit the data equally well in KL divergence. Triangles mark each model's prediction $l(\mathsf{u}\mid\theta)$ for each hypothesis, where the coherence bound is the separation between the possible coherent likelihoods in each hypothesis (given by the vaccination coverage $\alpha\%$). The black dot is the truth being classified: $f_1$ and $f_2$ for isolated communities, and the pooled truth $\bar f=\sum_i\pi^A_i f_i$ for the interacting population, with the corresponding $\pi^A$.
    \textbf{(a)} $\pi^A=(0.9,0.1)$ and $\pi^{A'}=(0.1,0.9)$ (corruption); 
    \textbf{(b)} $\pi^A=(0.1,0.9)$ and $\pi^{A'}=(0.9,0.1)$ (rescue).}
    \label{fig:rescue_corrupt}
\end{figure*}

\section{Discussion}
\label{sec:discussion}

The results rest on the fact that, in the long run, the interpretation channel decides what each hypothesis is taken to predict, and the belief channel decides whose data are held against those predictions. We discuss what this division implies, why it makes the effect of sharing interpretations depend on the setting, and which of its features follow from our modeling choices.

\paragraph{Two channels}

In the geometric visualization of Section~\ref{sec:characterizing_population}, the two channels act on different objects. The changes in the data agents observe, or how influential are their beliefs when sharing them, moves the pooled data $\bar f$, whereas anything that changes which readings circulate, or how influential is a given reading, moves the decision boundaries. Moreover, providing better evidence cannot change what the population counts as support for each hypothesis, while amplifying a different reading of the same evidence can change the conclusion without anyone observing new data. In the running example, the population that concludes the vaccine is not effective could be brought to $\theta_E$ either by giving Community 2 enough weight in the belief network ($\pi^A_2>0.826$) or by sufficiently reducing the reach of its demanding standard, which moves the wall below the pooled data. Mechanisms such as platform amplification, community notes, or the reach of a confident voice act on the second lever. This clean separation, however, holds only because interpretations do not respond to data, a point we return to below.

\paragraph{Why sharing helps or hurts}

Replacing private readings with a single consensus reading has two effects, and the comparisons of Section~\ref{sec:comparison} differ only in which one prevails. First, sharing breaks the match between each agent's reading and its own data. When every agent starts with a reading that matches its data, that match is what makes belief-only learning succeed, so can only affect negatively to the learning, since one common reading cannot fit several different data distributions at once, and the resulting misspecification can be large enough to discard the truth. Second, sharing overrides individual misreadings. When some agents misread their own data, the consensus can replace a misreading with an accurate reading, or carry a misreading to agents who were reading their data correctly. Which happens depends on how influence is distributed across the two channels, not on either network alone. The case of common data makes this especially sharp, in which under sharing, the belief network then plays no role in what is learned, and the outcome is decided by who dominates the interpretation channel. Without sharing, it is decided by who dominates the belief channel, weighted by how sharply each agent's reading separates the hypotheses. In both cases, errors spread without any false data, and instead arise from incorrect way of reading the private evidence.

\paragraph{Consensus is not evidence of correctness}

Whenever a single hypothesis best fits the pooled data under the consensus reading, the population ends up unanimous and its confidence grows exponentially. The rate depends only on how much better that hypothesis fits than its rivals, not on whether it is true. In the misspecified running example, sharing interpretations leads the population to reject the effective hypothesis about ten times faster than belief-only sharing leads it to accept it. A confident, unanimous population therefore says nothing, on its own, about whether the shared reading helped or harmed it.

\paragraph{Limitations and future work}

The main limitation is that interpretations are revised only through conversation and never confronted with data. This isolates the effect of socially transmitted readings, and it is reasonable when agents cannot tell which hypothesis's reading should be pulled toward what they observe. Still, people who repeatedly see data that contradict their standards may revise them. Several of our conclusions depend on this assumption, notably the deterministic consensus $l^*$ and the strict separation between moving the boundaries and moving the point. If interpretations also responded to data, changing what agents observe would move the boundaries as well, and some misreadings might correct themselves, as in models where likelihoods are learned from data \cite{hare2020non,hare2021general}. Letting interpretations respond to signals or beliefs is the natural next step. Our results are also asymptotic, the networks are static, and agents transmit beliefs and interpretations without strategic motives. Finite-time behavior, time-varying networks and strategic transmission remain open. Finally, the coherence regions are taken as given, and the running example is a stylized two-community illustration, not a calibrated model of vaccine beliefs.

\section{Conclusions}\label{sec:conclusions}

We proposed a model of non-Bayesian social learning in which agents share both their beliefs about which hypothesis is most likely true, and their interpretations of the hypotheses, meaning how likely each observation is under each hypothesis. These two objects are shared over possibly two different networks. We can recover the classical model as the special case with the interpretation adjacency matrix being $A'=I$.

Under standard connectivity assumptions, interpretations converge geometrically to a consensus that is the geometric pool of the initial readings weighted by each agent's influence in the interpretation network. Beliefs then concentrate almost surely and exponentially fast on the hypotheses that, under the consensus interpretation, best fit the population's pooled data, weighting each agent's data by her influence in the belief network. We can understand this result as the inference of one effective agent who would observe evidence coming from these pooled data and hold an interpretation of the hypotheses given by the consensus likelihoods. This gives a clear division of labor between the two networks and allows for an easy geometric visualization of both effects and the final learning outcome. The interpretation side, which includes the initial likelihoods and the interpretation network, fixes what each hypothesis is taken to mean and divides the geometric space into sections, each one representing the learned hypothesis outcome if the pooled data were placed in that section. The belief side, given by the individual and possibly heterogeneous signal distributions and the belief network, only determines which data become relevant to the learning process and places the point inside the geometry, determining which hypothesis will be learned in the long run.

Compared with sharing beliefs alone, sharing interpretations is neither always helpful nor always harmful. When every agent initially reads the true hypothesis exactly as her own data, sharing cannot help, and can hurt if agents observe different data. When interpretations are only required to be coherent and part or all of the population misreads its own data, the effect becomes two-sided. Sharing can rescue a population whose belief network is dominated by misreading agents, if correct readers dominate the interpretation network. It can also corrupt a population whose belief network is dominated by accurate agents, when misreading agents dominate interpretations. In neither case are the agents misinformed (i.e., they do not observe corrupted data), and what spreads is a reading of the evidence, not the evidence itself. In every one of these cases, when there exists a unique hypothesis that best explains the global data under the common reading, the population ends up confident and unanimous, so agreement on its own is never a sign that the interpretation channel did no harm. What an agent concludes therefore depends not only on what she observes and whose beliefs she listens to, but also on the different interpretations of the evidence her neighbors have.

\clearpage

\appendix

\section{Perron-Frobenius Estimate}\label{app:pf}

Let $M\in\R^{n\times n}$ be row-stochastic and primitive (irreducible and aperiodic). By the Perron-Frobenius theorem, $1$ is a simple eigenvalue of $M$, every other eigenvalue $\lambda$ satisfies $|\lambda|<1$, and $M$ possesses a unique stationary distribution $\pi_M\in\ri\Delta(N)$ with $\pi_M^{\top}M=\pi_M^{\top}$ and $\pi_M^{\top}\one=1$. Set
\begin{equation*}
    \rho_M:=\max\{|\lambda|:\ \lambda\text{ an eigenvalue of }M,\ \lambda\neq1\}<1,
\end{equation*}
with $\rho_M:=0$ if $1$ is the only eigenvalue.

\begin{lemma}[Geometric mixing]\label{lem:pf-rate}
For every $\hat\rho\in(\rho_M,1)$ there is a constant $C_{M,\hat\rho}<\infty$ such that

\begin{equation}
    \bigl\|M^t-\one\pi_M^{\top}\bigr\|_\infty\le C_{M,\hat\rho}\,\hat\rho^{\,t}\qquad\text{for all }t\ge0 .
    \label{eq:PF-mix}
\end{equation}

If $M$ is diagonalizable, \eqref{eq:PF-mix} also holds with $\hat\rho=\rho_M$ (with the convention $0^0=1$). Consequently $\|M^t-\one\pi_M^{\top}\|_\infty=O(\rho_M^{\,t})$ in the sense of Section~\ref{subsec:notation}, and $\sum_{t\ge0}\|M^t-\one\pi_M^{\top}\|_\infty<\infty$.
\end{lemma}

\begin{proof}
Let $P:=\one\pi_M^{\top}$ and $R:=M-P$. From $M\one=\one$, $\pi_M^\top M=\pi_M^\top$ and $\pi_M^\top\one=1$ we get $P^2=P$ and $MP=PM=P$, hence $PR=RP=0$. Since $P$ and $R$ commute, the binomial expansion gives $M^t=(P+R)^t=P+R^t$ for every $t\ge1$, that is, $M^t-P=R^t$ for $t\ge1$.

The decomposition $\mathbb C^n=\operatorname{span}(\one)\oplus\ker\pi_M^\top$ is invariant under $M$: indeed $M\one=\one$, and $\pi_M^\top w=0$ implies $\pi_M^\top Mw=\pi_M^\top w=0$. On $\operatorname{span}(\one)$, $R$ acts as $0$. On $\ker\pi_M^\top$, $P$ vanishes, so $R$ acts as $M$. Hence the spectrum of $R$ is $\{0\}$ together with the spectrum of $M|_{\ker\pi_M^\top}$. The latter is the spectrum of $M$ with the simple eigenvalue $1$ removed, so the spectral radius of $R$ is $\rho_M$.

By Gelfand's formula $\|R^t\|_\infty^{1/t}\to\rho_M$. Given $\hat\rho\in(\rho_M,1)$, there is therefore $t_0$ with $\|R^t\|_\infty\le\hat\rho^{\,t}$ for all $t\ge t_0$, and \eqref{eq:PF-mix} holds with $C_{M,\hat\rho}:=\max\bigl\{1,\max_{0\le t<t_0}\|M^t-P\|_\infty\,\hat\rho^{-t}\bigr\}$.

If $M$ is diagonalizable, so is its restriction to the invariant subspace $\ker\pi_M^\top$, and therefore $R=V\operatorname{diag}(0,\lambda_2,\dots,\lambda_n)V^{-1}$ for some invertible $V$, where $\lambda_2,\dots,\lambda_n$ are the eigenvalues of $M|_{\ker\pi_M^\top}$. Then $\|R^t\|_\infty\le\|V\|_\infty\|V^{-1}\|_\infty\,\rho_M^{\,t}$ for $t\ge1$, and $t=0$ is absorbed in the constant since $\|I-P\|_\infty\le2$.

Summability follows from $\sum_{t\ge0}C_{M,\hat\rho}\hat\rho^{\,t}=C_{M,\hat\rho}/(1-\hat\rho)<\infty$.
\end{proof}

If $M$ is not diagonalizable, the bound with $\hat\rho=\rho_M$ itself may fail. The sharp bound then carries a polynomial factor $t^{d-1}$, where $d$ is the largest Jordan-block size among the eigenvalues of modulus $\rho_M$. This is why the rates in the paper are stated with the convention of Section~\ref{subsec:notation}.

\section{Finite-Time Expressions for the Interpretation Channel}\label{app:rates}

This appendix records finite-time expressions for the KL functionals of the interpretation dynamics. They are not needed for the proof of Theorem~\ref{thm:consistency}, which uses only the limits $l_{j,t}\to l^*$ and $Z^\theta_{i,t}\to1$. They are included for readers interested in non-asymptotic bounds.

\begin{definition}[KL-gap matrix]\label{def:D}
For an ordered pair $(\theta_v,\theta_w)$ and $t\ge0$, the matrix $D_t(\theta_v,\theta_w)\in\R^{n\times n}$ has entries

\begin{equation}
    [D_t(\theta_v,\theta_w)]_{ij}
    :=\Dkl{f_i}{l_{j,t}(\cdot\mid\theta_w)}-\Dkl{f_i}{l_{j,t}(\cdot\mid\theta_v)} .
    \label{eq:Ddef}
\end{equation}

The row $i$ designates the agent whose truth $f_i$ is the reference, the column $j$, the agent whose model is evaluated.
\end{definition}

\begin{proposition}[Closed form for the KL-gap matrix]\label{thm:d-closed}
For every pair $(\theta_v,\theta_w)$ and every $t\ge0$,

\begin{equation}
    D_t(\theta_v,\theta_w) = D_0(\theta_v,\theta_w)\,\bigl(A'^{\,t}\bigr)^{\!\top} +\one\Bigl(\log\tfrac{\bm\Psi^{\theta_w}_t}{\bm\Psi^{\theta_v}_t}\Bigr)^{\!\top},
    \label{eq:d-closed}
\end{equation}

where $\bigl[\log\tfrac{\bm\Psi^{\theta_w}_t}{\bm\Psi^{\theta_v}_t}\bigr]_j :=\log\Psi^{\theta_w}_{j,t}-\log\Psi^{\theta_v}_{j,t}$.
\end{proposition}

\begin{corollary}[Limits and rates]\label{cor:d-conv}
Define $\bm d^*(\theta_v,\theta_w)\in\R^n$ by $d^*_i:=\Dkl{f_i}{l^*(\cdot\mid\theta_w)}-\Dkl{f_i}{l^*(\cdot\mid\theta_v)}$. Then $D_t(\theta_v,\theta_w)\to\bm d^*\one^{\top}$ with $\max_{i,j}|[D_t]_{ij}-d^*_i|=O(\rho_{A'}^{\,t})$. The drift vector $\bm e_t:=(A'\circ D_t)\one$ satisfies $\|\bm e_t-\bm d^*\|_\infty=O(\rho_{A'}^{\,t})$, and the normalizer-gap vector with entries $[\bm\delta_t]_i:=\log(Z^{\theta_w}_{i,t}/Z^{\theta_v}_{i,t})$ satisfies $\|\bm\delta_t\|_\infty=O(\rho_{A'}^{\,t})$.
\end{corollary}

\section{Proofs}
\label{app:proofs}

Several proofs below use the Shannon entropy $H(p):=-\sum_{x}p(x)\log p(x)$ of $p\in\Delta(X)$, through the identity $\Dkl{p}{q}=-H(p)-\sum_{x}p(x)\log q(x)$, valid whenever $q$ has full support.

\subsection{Proofs of Section \ref{sec:likelihood}}

\begin{proof}[Proof of Lemma \ref{lem:logratio}]
Taking logarithms in \eqref{eq:S1} for $s$ and $s'$ and subtracting, the normalizer $Z^\theta_{i,t-1}$ (independent of the signal argument) cancels, leaving $\log l_{i,t}(s\mid\theta)-\log l_{i,t}(s'\mid\theta) =\sum_{j\in\Nip}a'_{ij}\bigl(\log l_{j,t-1}(s\mid\theta)-\log l_{j,t-1}(s'\mid\theta)\bigr)$, which is the $i$-th coordinate of $A'\bm r_{t-1}$. The limit follows from the Perron-Frobenius estimate for the primitive matrix $A'$ (Appendix~\ref{app:pf}).
\end{proof}

\begin{proof}[Proof of Proposition \ref{prop:l-closed}]
By induction on $t$. The base case $t=0$ holds since $[A'^{0}]_{ij}=\delta_{ij}$, giving $\prod_j l_{j,0}(s\mid\theta)^{\delta_{ij}}=l_{i,0}(s\mid\theta)$, and $\Psi^\theta_{i,0}=\sum_s l_{i,0}(s\mid\theta)=1$, so \eqref{eq:l-closed} reads $l_{i,0}(s\mid\theta)=l_{i,0}(s\mid\theta)$.

Assume \eqref{eq:l-closed} holds at time $t-1$ for every $i,\theta,s$. By (A2) the initial profiles have full support, so every $l_{j,t-1}(\cdot\mid\theta)\in\ri\Delta(S)$ and all expressions below are well defined. Setting $a'_{ij}=0$ for $j\notin\Nip$ lets us write $\prod_{j\in\Nip}=\prod_{j=1}^n$. Substituting the inductive hypothesis into \eqref{eq:S1},

\begin{align*}
    l_{i,t}(s\mid\theta)
    &= \frac{\displaystyle\prod_{j}l_{j,t-1}(s\mid\theta)^{a'_{ij}}}{Z^\theta_{i,t-1}} = \frac{1}{Z^\theta_{i,t-1}}\prod_{j} \Biggl(\frac{\prod_{k}l_{k,0}(s\mid\theta)^{[A'^{\,t-1}]_{jk}}} {\Psi^\theta_{j,t-1}}\Biggr)^{\!a'_{ij}}\\
    &= \frac{1}{Z^\theta_{i,t-1}}\cdot \frac{\displaystyle\prod_{k}l_{k,0}(s\mid\theta)^{\sum_j a'_{ij}[A'^{\,t-1}]_{jk}}}{\displaystyle\prod_{j}\bigl(\Psi^{\theta}_{j,t-1}\bigr)^{a'_{ij}}} = \frac{1}{Z^\theta_{i,t-1}}\cdot \frac{\displaystyle\prod_{k}l_{k,0}(s\mid\theta)^{[A'^{\,t}]_{ik}}} {\displaystyle\prod_{j}\bigl(\Psi^{\theta}_{j,t-1}\bigr)^{a'_{ij}}},
\end{align*}

where the third equality uses $\sum_j a'_{ij}[A'^{\,t-1}]_{jk}=[A'A'^{\,t-1}]_{ik}=[A'^{\,t}]_{ik}$.

Write $\Gamma_i:=\prod_{j}\bigl(\Psi^{\theta}_{j,t-1}\bigr)^{a'_{ij}}$, which does not depend on $s$. Then

\begin{equation*}
    l_{i,t}(s\mid\theta)=\frac{1}{Z^\theta_{i,t-1}\,\Gamma_i}\, \prod_{k}l_{k,0}(s\mid\theta)^{[A'^{\,t}]_{ik}}.
\end{equation*}

Both $Z^\theta_{i,t-1}$ and $\Gamma_i$ are constants independent of $s$, so their product is the unique constant that normalizes $s\mapsto\prod_{k}l_{k,0}(s\mid\theta)^{[A'^{\,t}]_{ik}}$ to a probability distribution on $S$. Summing the last display over $s$ and using $\sum_s l_{i,t}(s\mid\theta)=1$,

\begin{equation*}
    1=\frac{1}{Z^\theta_{i,t-1}\,\Gamma_i}\sum_{s\in S}\prod_{k}l_{k,0}(s\mid\theta)^{[A'^{\,t}]_{ik}} =\frac{\Psi^{\theta}_{i,t}}{Z^\theta_{i,t-1}\,\Gamma_i},
\end{equation*}

where the second equality is exactly the normalization term in Proposition \ref{prop:l-closed}. Hence $Z^\theta_{i,t-1}\,\Gamma_i=\Psi^{\theta}_{i,t}$, that is,

\begin{equation}
    Z^\theta_{i,t-1}=\frac{\Psi^{\theta}_{i,t}}{\prod_{j}\bigl(\Psi^{\theta}_{j,t-1}\bigr)^{a'_{ij}}}\qquad\text{for all }i\in N,\ \theta\in\Theta,\ t\ge1,
    \label{eq:Z-Psi}
\end{equation}

and substituting back gives

\begin{equation*}
    l_{i,t}(s\mid\theta)=\frac{\prod_{k}l_{k,0}(s\mid\theta)^{[A'^{\,t}]_{ik}}}{\Psi^{\theta}_{i,t}},
\end{equation*}

which is \eqref{eq:l-closed} at time $t$. By induction the closed form holds for all $t\ge 0$.
\end{proof}

\begin{proof}[Proof of Theorem \ref{thm:l-consensus}]
Fix $\theta\in\Theta$ and $\hat\rho\in(\rho_{A'},1)$. If $A'$ is diagonalizable, $\hat\rho=\rho_{A'}$ is also allowed. By Lemma~\ref{lem:pf-rate} applied to the primitive matrix $A'$ of (A1), there is a constant $C_0<\infty$ with

\begin{equation*}
    \bigl|[A'^{\,t}]_{ij}-\pi^{A'}_j\bigr|\le\bigl\|A'^{\,t}-\one\pi^{A'\top}\bigr\|_\infty\le C_0\,\hat\rho^{\,t} \qquad\text{for all }i,j\text{ and }t\ge 0.
\end{equation*}

By (A2) every $l_{j,0}(s\mid\theta)>0$, so the numbers $\log l_{j,0}(s\mid\theta)$ are finite. Set

\begin{equation*}
    L:=\max_{j\in N,\,s\in S}\bigl|\log l_{j,0}(s\mid\theta)\bigr|<\infty.
\end{equation*}

\emph{Step 1: convergence of the log-numerators.} For $i\in N$, $s\in S$, define the exponents

\begin{equation*}
    E_{i,t}(s):=\sum_{j=1}^n[A'^{\,t}]_{ij}\log l_{j,0}(s\mid\theta), \qquad E_\infty(s):=\sum_{j=1}^n\pi^{A'}_j\log l_{j,0}(s\mid\theta).
\end{equation*}

Then, using the Perron-Frobenius bound termwise,

\begin{equation*}
    \bigl|E_{i,t}(s)-E_\infty(s)\bigr| \le\sum_{j=1}^n\bigl|[A'^{\,t}]_{ij}-\pi^{A'}_j\bigr|\,\bigl|\log l_{j,0}(s\mid\theta)\bigr| \le n\,C_0 L\,\hat\rho^{\,t} =:C_1\hat\rho^{\,t},
\end{equation*}

uniformly in $i,s$.

\emph{Step 2: convergence of the numerators.} Write the unnormalized numerators

\begin{equation*}
    N_{i,t}(s):=\prod_{j=1}^n l_{j,0}(s\mid\theta)^{[A'^{\,t}]_{ij}}=e^{E_{i,t}(s)}, \qquad N_\infty(s):=e^{E_\infty(s)}.
\end{equation*}

Every exponent satisfies $|E_{i,t}(s)|\le L$ and $|E_\infty(s)|\le L$, since the weights $[A'^{\,t}]_{ij}$ and $\pi^{A'}_j$ are nonnegative and sum to $1$. On the interval $[-L,L]$ the exponential is Lipschitz with constant $e^{L}$, so by Step 1

\begin{equation*}
    \bigl|N_{i,t}(s)-N_\infty(s)\bigr| \le e^{L}\,\bigl|E_{i,t}(s)-E_\infty(s)\bigr| \le e^{L}C_1\,\hat\rho^{\,t} =:C_2\hat\rho^{\,t},
\end{equation*}

uniformly in $i,s$.

\emph{Step 3: convergence of the normalizers.} Summing over the finitely many $s\in S$,

\begin{equation*}
    \bigl|\Psi^\theta_{i,t}-\Psi^\theta_\infty\bigr| =\Bigl|\sum_{s\in S}N_{i,t}(s)-\sum_{s\in S}N_\infty(s)\Bigr| \le\sum_{s\in S}\bigl|N_{i,t}(s)-N_\infty(s)\bigr| \le|S|\,C_2\,\hat\rho^{\,t},
\end{equation*}

uniformly in $i$. Moreover $\Psi^\theta_\infty=\sum_{s}N_\infty(s)\ge|S|\,e^{-L}>0$, and each $\Psi^\theta_{i,t}\ge|S|\,e^{-L}>0$ as well, since every $N_{i,t}(s)\ge e^{-L}$.

\emph{Step 4: convergence of the quotient.} By \eqref{eq:l-closed} and \eqref{eq:lstar},

\begin{equation*}
    l_{i,t}(s\mid\theta)-l^*(s\mid\theta) =\frac{N_{i,t}(s)}{\Psi^\theta_{i,t}}-\frac{N_\infty(s)}{\Psi^\theta_\infty}.
\end{equation*}

Using the elementary identity

\begin{equation*}
    \frac{a_t}{b_t}-\frac{a}{b} =\frac{(a_t-a)\,b-a\,(b_t-b)}{b_t\,b}
\end{equation*}

with $a_t=N_{i,t}(s)$, $b_t=\Psi^\theta_{i,t}$, $a=N_\infty(s)$, $b=\Psi^\theta_\infty$, together with the bounds $N_\infty(s)\le e^{L}$, $b=\Psi^\theta_\infty\ge|S|e^{-L}$, and $b_t=\Psi^\theta_{i,t}\ge|S|e^{-L}$, we obtain

\begin{equation*}
    \bigl|l_{i,t}(s\mid\theta)-l^*(s\mid\theta)\bigr| \le\frac{\bigl|N_{i,t}(s)-N_\infty(s)\bigr|\,\Psi^\theta_\infty +N_\infty(s)\,\bigl|\Psi^\theta_{i,t}-\Psi^\theta_\infty\bigr|} {\Psi^\theta_{i,t}\,\Psi^\theta_\infty} \le\frac{C_2\hat\rho^{\,t}\,\Psi^\theta_\infty +e^{L}\,|S|C_2\hat\rho^{\,t}} {\bigl(|S|e^{-L}\bigr)\,\Psi^\theta_\infty}.
\end{equation*}

The right-hand side equals $C_3\,\hat\rho^{\,t}$ with $C_3:=\bigl(C_2\Psi^\theta_\infty+e^{L}|S|C_2\bigr)/\bigl(|S|e^{-L}\Psi^\theta_\infty\bigr)$, which does not depend on $i$, $s$ or $t$. Hence

\begin{equation*}
    \max_{i\in N,\,s\in S}\bigl|l_{i,t}(s\mid\theta)-l^*(s\mid\theta)\bigr| \le C_3\,\hat\rho^{\,t}\qquad\text{for all }t\ge0.
\end{equation*}

Since $\hat\rho\in(\rho_{A'},1)$ was arbitrary, and $\hat\rho=\rho_{A'}$ is admissible when $A'$ is diagonalizable, this is \eqref{eq:l-conv} in the sense of Section~\ref{subsec:notation}. Finally, each factor $l_{j,0}(s\mid\theta)^{\pi^{A'}_j}$ in \eqref{eq:lstar} is strictly positive by (A2) and $\Psi^\theta_\infty>0$, so $l^*(\cdot\mid\theta)\in\ri\Delta(S)$. Since $\theta\in\Theta$ was arbitrary, the claim holds for every $\theta$.
\end{proof}

\subsection{Proofs of Section \ref{sec:beliefs}}\label{app:proofs_beliefs}

This appendix first derives the linear recursion of the log-beliefs ratios (Appendix~\ref{app:recursion}) and computes the expectation and long-run of the drift (Appendix~\ref{app:drift}). It then proves Theorem~\ref{thm:consistency} and Corollary~\ref{cor:singleton} (Appendix~\ref{app:consistency-proof}). The recursion inherits its skeleton from \cite{nedic2017fast}. There, the drift is a static function of fixed models. Here, the likelihoods are the evolving pool $l_{j,t}$, so the drift is a function of a moving target and its limit must be carried through the Bayesian step explicitly. The rest of the argument combines a strong law for independent summands with network mixing.

\subsubsection{The log-belief-ratio recursion}\label{app:recursion}

Both social steps act on beliefs multiplicatively, so taking logarithms turns their products into sums and the whole update becomes additive. This is the same coordinate change that turned the geometric likelihood pool into linear averaging in Lemma~\ref{lem:logratio}. Applied to beliefs it produces a linear recursion, as we now show. Fix an ordered pair of hypotheses $(\theta_v,\theta_w)$, let $\phi_{i,t}=\phi_{i,t}(\theta_v,\theta_w):=\log\frac{\mu_{i,t}(\theta_v)}{\mu_{i,t}(\theta_w)},
  \qquad \theta_v,\theta_w\in\Theta,$ be the log-belief ratio, and stack the agents' ratios into

\begin{equation*}
    \bm\phi_t(\theta_v,\theta_w):=\bigl(\phi_{1,t},\dots,\phi_{n,t}\bigr)^{\!\top}\in\R^n .
\end{equation*}
 
By (A2) every belief has full support, so $\bm\phi_t$ is finite for all $t$. Take the log-ratio of the Bayesian step \eqref{eq:S3} and substitute the log-linear pool \eqref{eq:S2} and the pooled model \eqref{eq:S1}. The normalizers of the interim belief and of the geometric pool cancel in every ratio, and we obtain, for each $i$,
 
\begin{equation}
  \phi_{i,t}
  =\sum_{j=1}^n a_{ij}\,\phi_{j,t-1}+[\bm x_t]_i,
  \label{eq:onestep}
\end{equation}
 
where the scalar increment collects the fresh evidence seen at agent $i$,
 
\begin{equation}
  [\bm x_t]_i
  :=\sum_{j=1}^n a'_{ij}\,
       \log\frac{l_{j,t-1}(\omega_{it}\mid\theta_v)}
                {l_{j,t-1}(\omega_{it}\mid\theta_w)}
     \;+\;\log\frac{Z^{\theta_w}_{i,t-1}}{Z^{\theta_v}_{i,t-1}} .
  \label{eq:increment}
\end{equation}
 
The increment has one term from each network. The first is an $a'_{ij}$-weighted average, over the models agent $i$ attends to, of how strongly those models tell $\theta_v$ and $\theta_w$ apart at the realized signal $\omega_{it}$. The second is the residual left by the geometric-pool normalizers. All the randomness of the update is contained in $\bm x_t$. The deterministic social averaging is contained in the matrix $A$.
 
Stacking the increments and iterating \eqref{eq:onestep} gives the matrix form and its unrolling,
 
\begin{equation}
  \bm\phi_t(\theta_v,\theta_w)=A\,\bm\phi_{t-1}(\theta_v,\theta_w)+\bm x_t ,
  \qquad
  \bm\phi_t
  =A^{t}\bm\phi_0+\sum_{\tau=1}^{t}A^{\,t-\tau}\bm x_\tau .
  \label{eq:unrolled}
\end{equation}
 
The unrolled form is worth reading slowly, because it fixes what remains to be done. The state splits into two parts: an initial-condition term $A^{t}\bm\phi_0$, which decays as the network mixes ($A^{t}\to\one\pi^{A\top}$, see Section~\ref{subsec:notation}), and an accumulated-evidence term $\sum_\tau A^{\,t-\tau}\bm x_\tau$, which carries the long-run behavior. The first part forgets where the beliefs started. The second is where learning happens. Everything that follows is an analysis of the second part.
 
Two properties of the forcing $\bm x_t$ make that analysis possible. Because the models $l_{j,t-1}(\cdot\mid\theta)$ and the normalizers $Z^{\theta}_{i,t-1}$ are deterministic (Section~\ref{sec:likelihood}) and $\omega_{it}\sim f_i$ is drawn independently of the past $\F_{t-1}$, the vector $\bm x_t$ is $\F_t$-measurable, and the $\{\bm x_t\}_{t\ge1}$ are mutually independent: each is a deterministic function of the disjoint, time-indexed signal family $\{\omega_{it}\}_{i\in N}$. An independent forcing sequence can be split into its mean and a mean-zero fluctuation, and the two can be handled separately, which is what we do next.

\subsubsection{The drift of the increment}\label{app:drift}

Since the accumulated forcing governs the dynamics, we separate each increment into its conditional mean, the drift, and a mean-zero fluctuation. The fluctuation will later be shown to average out. The drift is where the direction of learning is decided, so we compute it here.
 
Because $l_{j,t-1}(\cdot\mid\theta)$ is $\F_{t-1}$-measurable and $\omega_{it}\sim f_i$ is independent of $\F_{t-1}$,
 
\begin{align}
    \E\!\left[\log\frac{l_{j,t-1}(\omega_{it}\mid\theta_v)}{l_{j,t-1}(\omega_{it}\mid\theta_w)}\,\Big|\,\F_{t-1}\right]
    &=\sum_{s\in S} f_i(s)\log\frac{l_{j,t-1}(s\mid\theta_v)}{l_{j,t-1}(s\mid\theta_w)} .
    \notag
\end{align}
 
Adding and subtracting $\sum_s f_i(s)\log f_i(s)$ rewrites this expectation as a difference of two KL divergences,
 
\begin{equation}
    \E\!\left[\log\frac{l_{j,t-1}(\omega_{it}\mid\theta_v)}{l_{j,t-1}(\omega_{it}\mid\theta_w)}\,\Big|\,\F_{t-1}\right]
    = \Dkl{f_i}{l_{j,t-1}(\cdot\mid\theta_w)}-\Dkl{f_i}{l_{j,t-1}(\cdot\mid\theta_v)},
    \label{eq:EL}
\end{equation}
 
and this form makes the meaning plain: in expectation, evidence pushes agent $i$'s belief toward whichever hypothesis its model explains better, and the strength of the push is exactly how much better $\theta_v$ fits agent $i$'s own truth $f_i$ than $\theta_w$ does, measured in KL divergence.
 
The divergences in \eqref{eq:EL} are evaluated at the moving models $l_{j,t-1}$, and this is where the work of Section~\ref{sec:likelihood} pays off. By Theorem~\ref{thm:l-consensus} the models converge geometrically to the common limit $l^{*}$, and KL divergence is continuous on $\ri\Delta(S)$ (the models stay in a compact subset of $\ri\Delta(S)$ by Lemma~\ref{lem:unifbound}), so $\Dkl{f_i}{l_{j,t-1}(\cdot\mid\theta)}\to\Dkl{f_i}{l^{*}(\cdot\mid\theta)}$ for every $j$. Because the weights satisfy $\sum_j a'_{ij}=1$ and the limit no longer depends on $j$, the $a'_{ij}$-weighting inside the increment averages a single number and thus disappears in the limit. Then, the $a'_{ij}$-weighted sum of \eqref{eq:EL} converges to
 
\begin{equation}
    [\bm d^*]_i := \Dkl{f_i}{l^*(\cdot\mid\theta_w)} - \Dkl{f_i}{l^*(\cdot\mid\theta_v)}.
    \label{eq:dstar}
\end{equation}
 
The normalizer term of \eqref{eq:increment} vanishes as well, because $Z^\theta_{i,t-1}\to\sum_s l^*(s\mid\theta)=1$. Hence the conditional mean of the increment converges to a deterministic vector,
 
\begin{equation}
    \bar{\bm x}_t := \E[\bm x_t\mid\F_{t-1}] = \E[\bm x_t] \longrightarrow \bm d^*,
    \qquad t\to\infty.
    \label{eq:meanconv}
\end{equation}

The limiting drift $\bm d^{*}$ has a clean reading: its $i$-th entry is how much better $\theta_v$ explains agent $i$'s truth than $\theta_w$ does, now judged through the single shared model $l^{*}$ instead of through agent $i$'s original one. The individual models have dropped out, only the consensus remains.
 
\begin{remark}[Heterogeneous truths before consensus]\label{rem:hadamard}
One feature of \eqref{eq:EL} is visible before consensus but not after.

The expectation depends on two indices at once: on the row $i$, because the signal $\omega_{it}$ is drawn from the agent-specific truth $f_i$, and on the column $j$, because the likelihood is read off agent $j$'s still-distinct model. The mean increment is therefore a row-sum of a Hadamard product plus the deterministic normalizer gap,
 
\begin{equation*}
    \E[\bm x_t] = (A'\circ D_{t-1}(\theta_v,\theta_w))\,\one + \bm\delta_{t-1},
    \qquad
    [D_t(\theta_v,\theta_w)]_{ij} = \Dkl{f_i}{l_{j,t}(\cdot\mid\theta_w)} - \Dkl{f_i}{l_{j,t}(\cdot\mid\theta_v)},
\end{equation*}
 
where $[\bm\delta_{t-1}]_i:=\log\bigl(Z^{\theta_w}_{i,t-1}/Z^{\theta_v}_{i,t-1}\bigr)$ is the normalizer term of \eqref{eq:increment}.

Notice this is not the simpler product $A'\bm d$ that would appear if all agents shared one truth $f_1=\dots=f_n$. The two indices separate only in the limit, once every column of $D_t$ becomes identical because $l_{j,t}(\cdot\mid\theta)\to l^*(\cdot\mid\theta)$. This collapse is what removes the index $j$ from $\bm d^{*}$ and lets it act as a single drift.
\end{remark}

\subsubsection{Proof of Theorem~\ref{thm:consistency} and Corollary~\ref{cor:singleton}}\label{app:consistency-proof}

The proof is the standard strong-law-plus-mixing scheme. It is specialized to the present setting only through the drift limit \eqref{eq:meanconv}, and it rests on three facts, one for each concern raised by the unrolled recursion \eqref{eq:unrolled}.

\begin{enumerate}
    \item The increments are uniformly bounded, because the pooled models never approach the boundary of the simplex (Lemma~\ref{lem:unifbound}). The mean-zero fluctuations of the forcing therefore obey a strong law, and their time-average vanishes, so the noise does not build up.
    \item The belief network mixes geometrically, $A^{t}\to\one\pi^{A\top}$ (Lemma~\ref{lem:mixing}). This erases the initial term $A^{t}\bm\phi_0$ and, in the accumulated-evidence sum, drives the weight on each agent's past forcing to its stationary value $\pi^{A}$. This is where $\pi^{A}$ becomes the weighting that matters.
    \item Combining the two with the drift limit \eqref{eq:meanconv} yields a Cesàro law (Lemma~\ref{lem:cesaro}): the time-average of the accumulated forcing converges almost surely to the network drift $\one\,\pi^{A\top}\bm d^{*}$, which gives \eqref{eq:phi-rate}.
\end{enumerate}

\begin{lemma}[Uniform bound on the increments]\label{lem:unifbound}
There exists $\beta\in(0,1)$ with $l_{j,t}(s\mid\theta)\ge\beta$ for all $t\ge0$, $j\in N$, $s\in S$, $\theta\in\Theta$. Consequently $\|\bm x_t\|_\infty\le B_x:=\log\tfrac1\beta+\log\tfrac{|S|}{\beta}$ a.s.\ for all $t\ge1$.
\end{lemma}

\begin{proof}[Proof of Lemma \ref{lem:unifbound}]
Let $\underline l:=\min_{j\in N,\,s\in S,\,\theta\in\Theta}l_{j,0}(s\mid\theta)$, which is $>0$ by (A2) and the finiteness of $N,S,\Theta$, and set $\beta:=\underline l/|S|$. (The bounds below are not sharp: by the AM-GM inequality one even has $\Psi^\theta_{i,t}\le1$ and $Z^\theta_{i,t-1}\le1$; only finiteness matters.) Since $\underline l\le\sum_{s}l_{j,0}(s\mid\theta)=1$ and $|S|\ge 2$, we have $\beta\in(0,1)$.

\emph{Uniform lower bound.} Fix $i\in N$, $\theta\in\Theta$, $s\in S$, $t\ge 0$. In the closed form \eqref{eq:l-closed} the exponents $[A'^{\,t}]_{ij}$ are nonnegative and sum to $1$, because $A'^{\,t}$ is row-stochastic. Hence each weighted geometric mean obeys

\begin{equation*}
    \underline l=\underline l^{\sum_j[A'^{\,t}]_{ij}} \le\prod_{j=1}^n l_{j,0}(s\mid\theta)^{[A'^{\,t}]_{ij}} \le 1^{\sum_j[A'^{\,t}]_{ij}}=1,
\end{equation*}

using $\underline l\le l_{j,0}(s\mid\theta)\le 1$ for every $j$. Thus the numerator of \eqref{eq:l-closed} is $\ge\underline l$, and the denominator $\Psi^\theta_{i,t}=\sum_{s'}\prod_j l_{j,0}(s'\mid\theta)^{[A'^{\,t}]_{ij}}$ is a sum of $|S|$ such means, so $\Psi^\theta_{i,t}\le|S|$. Therefore

\begin{equation*}
    l_{i,t}(s\mid\theta)=\frac{\prod_j l_{j,0}(s\mid\theta)^{[A'^{\,t}]_{ij}}}{\Psi^\theta_{i,t}} \ge\frac{\underline l}{|S|}=\beta,
\end{equation*}

and $l_{i,t}(s\mid\theta)\le 1$ since it is a probability. Hence $\beta\le l_{i,t}(s\mid\theta)\le 1$ for all $i,\theta,s,t\ge 0$.

\emph{Bound on the increment.} By definition each coordinate of $\bm x_t$ is the sum of a signal-dependent log-likelihood-ratio term and a log-normalizer-ratio term (see \eqref{eq:increment}). For the first, every entry $\log l_{i,t}(s\mid\theta)$ lies in $[\log\beta,0]$ by the bound just proved, so a difference of two such logarithms lies in $[\log\beta,-\log\beta]=[-\log\tfrac1\beta,\log\tfrac1\beta]$. The $a'_{ij}$-weighted average of such differences (a convex combination) lies in the same interval, giving a contribution of absolute value $\le\log\tfrac1\beta$. For the second, each normalizer $Z^\theta_{i,t-1}=\sum_{s}\prod_j l_{j,t-1}(s\mid\theta)^{a'_{ij}}$ is a sum of $|S|$ weighted geometric means of numbers in $[\beta,1]$, hence $Z^\theta_{i,t-1}\in[\beta,|S|]$. Therefore

\begin{equation*}
    \Bigl|\log\frac{Z^{\theta_w}_{i,t-1}}{Z^{\theta_v}_{i,t-1}}\Bigr| \le\log\frac{|S|}{\beta}.
\end{equation*}

Adding the two contributions, $\|\bm x_t\|_\infty\le\log\tfrac1\beta+\log\tfrac{|S|}{\beta}=B_x$ almost surely for all $t\ge 1$.
\end{proof}

\begin{lemma}[Mixing of the aggregation matrix]\label{lem:mixing}
Let $\Lambda_t:=A^{t}-\one\pi^{A\top}$. Then $\|\Lambda_t\|_\infty=O(\rho_A^{\,t})$ and $S_A:=\sum_{t\ge0}\|\Lambda_t\|_\infty<\infty$.
\end{lemma}

\begin{proof}[Proof of Lemma \ref{lem:mixing}]
Fix $\hat\rho\in(\rho_A,1)$; if $A$ is diagonalizable, $\hat\rho=\rho_A$ is also allowed. Since $A$ is primitive and row-stochastic, Lemma~\ref{lem:pf-rate} gives a constant $C_A<\infty$ with $\|\Lambda_t\|_\infty=\|A^{t}-\one\pi^{A\top}\|_\infty\le C_A\,\hat\rho^{\,t}$ for all $t\ge 0$. This is $\|\Lambda_t\|_\infty=O(\rho_A^{\,t})$ in the sense of Section~\ref{subsec:notation}. Since $\hat\rho<1$, the geometric series converges, so $S_A=\sum_{t\ge 0}\|\Lambda_t\|_\infty\le C_A\sum_{t\ge0}\hat\rho^{\,t}=\frac{C_A}{1-\hat\rho}<\infty$.
\end{proof}

\begin{lemma}[Ces\`aro limit of the driving term]\label{lem:cesaro}
With $\bm d^*$ as in \eqref{eq:dstar},

\begin{equation}
    \lim_{t\to\infty}\Bigl(\frac1t\sum_{\tau=1}^{t}A^{\,t-\tau}\bm x_\tau
    -\one\,\pi^{A\top}\bm d^*\Bigr)=\bm 0\qquad\text{a.s.}
    \label{eq:cesaro}
\end{equation}
\end{lemma}

\begin{proof}[Proof of Lemma \ref{lem:cesaro}]
Split the average using $\Lambda_t:=A^{t}-\one\pi^{A\top}$:

\begin{equation}
    \frac1t\sum_{\tau=1}^{t}A^{\,t-\tau}\bm x_\tau
    =\underbrace{\frac1t\sum_{\tau=1}^{t}\Lambda_{t-\tau}\bm x_\tau}_{(\mathrm I)}
    +\one\pi^{A\top}
    \underbrace{\frac1t\sum_{\tau=1}^{t}\bm x_\tau}_{(\mathrm{II})} ,
    \label{eq:split}
\end{equation}

where we used $A^{\,t-\tau}=\Lambda_{t-\tau}+\one\pi^{A\top}$ and pulled the constant matrix $\one\pi^{A\top}$ out of the sum.

\emph{Term $(\mathrm I)$ vanishes.} By Lemma~\ref{lem:unifbound}, $\|\bm x_\tau\|_\infty\le B_x$ a.s.\ for all $\tau\ge 1$, and by Lemma~\ref{lem:mixing}, $\sum_{s\ge0}\|\Lambda_s\|_\infty=S_A<\infty$. By the triangle inequality for $\|\cdot\|_\infty$, sub-multiplicativity of the operator $\infty$-norm (i.e.\ $\|Mv\|_\infty\le\|M\|_\infty\|v\|_\infty$, where $\|M\|_\infty=\max_i\sum_j|M_{ij}|$ is the matrix norm induced by the vector $\infty$-norm), and reindexing $s=t-\tau$:

\begin{equation*}
    \|(\mathrm I)\|_\infty
    =\Bigl\|\tfrac1t\sum_{\tau=1}^{t}\Lambda_{t-\tau}\bm x_\tau\Bigr\|_\infty
    \le\frac1t\sum_{\tau=1}^{t}\|\Lambda_{t-\tau}\bm x_\tau\|_\infty
    \le\frac1t\sum_{\tau=1}^{t}\|\Lambda_{t-\tau}\|_\infty\,\|\bm x_\tau\|_\infty
    \le\frac{B_x}{t}\sum_{s=0}^{t-1}\|\Lambda_{s}\|_\infty
    \le\frac{B_x\,S_A}{t}\xrightarrow[t\to\infty]{}0\qquad\text{a.s.}
\end{equation*}

Hence $(\mathrm I)\to\bm 0$ almost surely.

\emph{Term $(\mathrm{II})$ converges to $\bm d^*$.} Write $\bm x_\tau=\bar{\bm x}_\tau+\bm\xi_\tau$ with $\bar{\bm x}_\tau:=\E[\bm x_\tau]$ and $\bm\xi_\tau:=\bm x_\tau-\bar{\bm x}_\tau$.

\emph{(a) Fluctuations.} Fix a coordinate $i$. By Section~\ref{sec:likelihood} the models $l_{j,\tau-1}(\cdot\mid\theta)$ and the normalizers $Z^\theta_{i,\tau-1}$ are deterministic, so $[\bm x_\tau]_i$ is a fixed measurable function of the single signal $\omega_{i\tau}$ alone. Since the signals $\{\omega_{i\tau}\}_{\tau\ge1}$ are independent across $\tau$, the variables $\{[\bm\xi_\tau]_i\}_{\tau\ge1}$ are independent. Each has mean zero, since $\bar{\bm x}_\tau=\E[\bm x_\tau]$ implies $\E[[\bm\xi_\tau]_i]=0$. From $\|\bm x_\tau\|_\infty\le B_x$ we get $|[\bm\xi_\tau]_i|=|[\bm x_\tau]_i-\E[\bm x_\tau]_i|\le 2B_x$, hence $\operatorname{Var}([\bm\xi_\tau]_i)\le 4B_x^2$. Therefore

\begin{equation*}
    \sum_{\tau\ge1}\frac{\operatorname{Var}([\bm\xi_\tau]_i)}{\tau^2} \le 4B_x^2\sum_{\tau\ge1}\frac1{\tau^2}<\infty,
\end{equation*}

and Kolmogorov's criterion for the strong law of large numbers for independent (not necessarily identically distributed) summands yields $\frac1t\sum_{\tau=1}^t[\bm\xi_\tau]_i\to 0$ a.s. Intersecting the finitely many almost-sure events over $i$, $\frac1t\sum_{\tau=1}^t\bm\xi_\tau\to\bm0$ a.s.

\emph{(b) Means.} By \eqref{eq:meanconv} the deterministic means converge, $\bar{\bm x}_\tau\to\bm d^*$. The Ces\`aro average of a convergent sequence converges to the same limit, so $\frac1t\sum_{\tau=1}^t\bar{\bm x}_\tau\to\bm d^*$.

Adding (a) and (b), $(\mathrm{II})=\frac1t\sum_{\tau=1}^t\bm x_\tau\to\bm d^*$ a.s., then $\one\pi^{A\top}(\mathrm{II})\to\one\pi^{A\top}\bm d^*$. Combined with $(\mathrm I)\to\bm0$ in \eqref{eq:split}, this proves \eqref{eq:cesaro}.
\end{proof}

\begin{proof}[Proof of Theorem~\ref{thm:consistency}]
It suffices to show $\mu_{i,t}(\theta_v)\to0$ a.s.\ for an arbitrary fixed $\theta_v\notin\Theta^*$ and every $i\in N$. Since $\Theta$ is finite, summing over the finitely many such $\theta_v$ then forces the belief mass to concentrate on $\Theta^*$.

Fix $\theta_v\notin\Theta^*$ and any $\theta_w\in\Theta^*$, and consider the log-belief-ratio vector $\bm\phi_t=\bm\phi_t(\theta_v,\theta_w)$ with coordinates $\phi_{i,t}=\log\bigl(\mu_{i,t}(\theta_v)/\mu_{i,t}(\theta_w)\bigr)$. Dividing the unrolled recursion \eqref{eq:unrolled} by $t$,

\begin{equation*}
    \frac1t\,\bm\phi_t =\frac1t\,A^{t}\bm\phi_0 +\frac1t\sum_{\tau=1}^{t}A^{\,t-\tau}\bm x_\tau .
\end{equation*}

For the first term, $A$ row-stochastic gives $\|A^{t}\bm\phi_0\|_\infty\le\|\bm\phi_0\|_\infty<\infty$ (the initial beliefs lie in $\ri\Delta(\Theta)$ by (A2), so $\bm\phi_0$ is finite), then $\|\tfrac1t A^t\bm\phi_0\|_\infty\le\|\bm\phi_0\|_\infty/t\to0$ deterministically. By Lemma~\ref{lem:cesaro} the second term converges a.s.\ to $\one\pi^{A\top}\bm d^*$. Therefore

\begin{equation*}
    \frac1t\,\bm\phi_t\xrightarrow[t\to\infty]{\mathrm{a.s.}}\one\,\pi^{A\top}\bm d^* .
\end{equation*}

By \eqref{eq:dstar} and Definition~\ref{def:confidence},

\begin{equation*}
    \pi^{A\top}\bm d^* =\sum_{i=1}^n\pi^A_i\Bigl(\Dkl{f_i}{l^*(\cdot\mid\theta_w)} -\Dkl{f_i}{l^*(\cdot\mid\theta_v)}\Bigr) =C(\theta_v)-C(\theta_w)=C(\theta_v)-C^*,
\end{equation*}

using $\theta_w\in\Theta^*$, i.e.\ $C(\theta_w)=C^*$. Reading this coordinatewise gives \eqref{eq:phi-rate}: for every $i$,

\begin{equation*}
    \frac{\phi_{i,t}}{t}\xrightarrow[t\to\infty]{\mathrm{a.s.}}C(\theta_v)-C^* .
\end{equation*}

Since $\theta_v\notin\Theta^*$ and $\Theta^*$ is the set of maximizers of $C$, we have $C^*-C(\theta_v)>0$, so the a.s.\ limit $C(\theta_v)-C^*$ is strictly negative. Consequently $\phi_{i,t}\to-\infty$ a.s.\ for every $i$. In detail, for almost every sample path, there exist a $T$ with $\phi_{i,t}/t\le\tfrac12(C(\theta_v)-C^*)<0$ for all $t\ge T$, so $\phi_{i,t}\le\tfrac{t}{2}(C(\theta_v)-C^*)\to-\infty$.

Finally, beliefs remain in $\ri\Delta(\Theta)$, so $\mu_{i,t}(\theta_w)\in(0,1]$ and

\begin{equation*}
    0\le\mu_{i,t}(\theta_v)\le\frac{\mu_{i,t}(\theta_v)}{\mu_{i,t}(\theta_w)} =\exp(\phi_{i,t})\xrightarrow[t\to\infty]{\mathrm{a.s.}}0 .
\end{equation*}

Hence $\mu_{i,t}(\theta_v)\to0$ a.s. As $\theta_v\notin\Theta^*$ and $i\in N$ were arbitrary, and $\Theta\setminus\Theta^*$ is finite, $\sum_{\theta_v\notin\Theta^*}\mu_{i,t}(\theta_v)\to0$ a.s., i.e.\ the belief mass concentrates on $\Theta^*$ for every $i$.
\end{proof}

\begin{proof}[Proof of Corollary \ref{cor:singleton}]
By Theorem~\ref{thm:consistency}, $\mu_{i,t}(\theta)\to0$ a.s.\ for each of the finitely many $\theta\neq\theta^\star$. Since $\sum_\theta\mu_{i,t}(\theta)=1$, $\mu_{i,t}(\theta^\star)\to1$ a.s.
\end{proof}

\begin{proof}[Proof of Lemma~\ref{lem:twins}]
First, $l_{j,t}(\cdot\mid\theta)=l_{j,t}(\cdot\mid\theta')$ for all $j$ and $t$, by induction on $t$: it holds at $t=0$ by assumption, and \eqref{eq:S1} maps identical inputs to identical outputs. For the same reason $Z^{\theta}_{i,t}=Z^{\theta'}_{i,t}$ for all $i,t$.

The derivation of \eqref{eq:onestep}-\eqref{eq:increment} in Appendix~\ref{app:recursion} uses only \eqref{eq:S1}-\eqref{eq:S3} and the row-stochasticity of $A$ and $A'$, not their primitivity. It therefore applies here with $(\theta_v,\theta_w)=(\theta,\theta')$. Both terms of \eqref{eq:increment} vanish identically, so $\bm x_t=\bm 0$ and $\bm\phi_t=A^{t}\bm\phi_0$ for all $t$. Since $A^t$ has nonnegative entries and unit row sums, $|[A^t\bm\phi_0]_i|\le\|\bm\phi_0\|_\infty=c$, which is the first claim.

For the second, $\mu_{i,t}(\theta)=\mu_{i,t}(\theta')e^{\phi_{i,t}}$ and $\mu_{i,t}(\theta)+\mu_{i,t}(\theta')\le1$ give $\mu_{i,t}(\theta')\bigl(1+e^{\phi_{i,t}}\bigr)\le1$. Hence $\mu_{i,t}(\theta')\le(1+e^{\phi_{i,t}})^{-1}\le(1+e^{-c})^{-1}<1$.
\end{proof}

\subsection{Proofs of Section \ref{sec:characterizing_population}}

\begin{proof}[Proof of Proposition \ref{prop:effective_single_agent}]
Expanding $\Dkl{f_i}{l^*(\cdot\mid\theta)}=-H(f_i)-\sum_s f_i(s)\log l^*(s\mid\theta)$, substituting into \eqref{eq:confidence}, and using $\sum_i\pi^A_i f_i(s)=\bar f(s)$,

\begin{equation*}
    C(\theta) =\underbrace{\sum_i\pi^A_iH(f_i)-H(\bar f)}_{=:K,\ \theta\text{-independent}} -\Dkl{\bar f}{l^*(\cdot\mid\theta)} .
\end{equation*}

Differencing across $\theta_v,\theta_w$ cancels $K$ and gives \eqref{eq:confidence-gap-pooled}. Taking the argmax gives \eqref{eq:effective-argmin}. The divergence is finite since $\bar f$ has full support by (A3) and $\pi^A_i>0$, and $l^*(\cdot\mid\theta)\in\ri\Delta(S)$.
\end{proof}

\begin{proof}[Proof of Corollary \ref{cor:concrete}]
By \eqref{eq:effective-argmin}, $\Theta^*=\{\theta^\star\}$ holds if and only if $\Dkl{\bar f}{l^*(\cdot\mid\theta^\star)}<\Dkl{\bar f}{l^*(\cdot\mid\theta)}$ for all $\theta\neq\theta^\star$. Since the entropy $H(\bar f)$ cancels in the difference, this equals \eqref{eq:concrete-llr}. Inserting the identity $\log l^*(s\mid\theta)=\sum_j\pi^{A'}_j\log l_{j,0}(s\mid\theta)-\log\Psi^\theta_\infty$ from \eqref{eq:lstar}, and using $\sum_s\bar f(s)=1$, gives \eqref{eq:concrete-l0}. Full learning of $\theta^\star$ then follows from Corollary~\ref{cor:singleton}.
\end{proof}

\begin{proof}[Proof of Proposition~\ref{prop:geometry} and Remark~\ref{rem:cells}]
Throughout, $|\Theta|=m\ge2$. For $p\in\Delta(S)$ and $\theta\in\Theta$, write

\begin{equation*}
    \Lambda_\theta(p):=\sum_{s}p(s)\log l^*(s\mid\theta)=-H(p)-\Dkl{p}{l^*(\cdot\mid\theta)},
\end{equation*}

where $H(p)=-\sum_s p(s)\log p(s)$, with the convention $0\log0=0$. Both terms are finite because $l^*(\cdot\mid\theta)\in\ri\Delta(S)$ by \eqref{eq:lstar}. For any $\theta,\theta'\in\Theta$,

\begin{equation}
    g_{\theta,\theta'}(p):=\sum_{s}p(s)\,\nu_{\theta,\theta'}(s)=\Lambda_\theta(p)-\Lambda_{\theta'}(p)
    =\Dkl{p}{l^*(\cdot\mid\theta')}-\Dkl{p}{l^*(\cdot\mid\theta)}.
    \label{eq:geom-identity}
\end{equation}

Each $g_{\theta,\theta'}$ is the restriction to $\Delta(S)$ of a linear map on $\R^S$, hence continuous.

\emph{Learning in a cell.} If $\bar f\in\mathcal R(\theta)$, then $\theta$ is the unique minimizer of $\theta'\mapsto\Dkl{\bar f}{l^*(\cdot\mid\theta')}$ over $\Theta$. By \eqref{eq:effective-argmin}, $\Theta^*=\{\theta\}$, and Corollary~\ref{cor:singleton} gives $\mu_{i,t}(\theta)\to1$ a.s.\ for every $i\in N$.

\emph{Proposition, (i).} By \eqref{eq:geom-identity}, $\Dkl{p}{l^*(\cdot\mid\theta)}<\Dkl{p}{l^*(\cdot\mid\theta')}$ holds if and only if $g_{\theta,\theta'}(p)>0$. This gives the linear description of $\mathcal R(\theta)$. For $\theta'\neq\theta$, let $H_{\theta'}:=\{p\in\R^S:\sum_sp(s)\nu_{\theta,\theta'}(s)>0\}$. Each $H_{\theta'}$ is convex, as the strict superlevel set of a linear map, so

\begin{equation}
    \mathcal R(\theta)=\Delta(S)\cap\bigcap_{\theta'\neq\theta}H_{\theta'}
    \label{eq:cell-intersection}
\end{equation}

is convex, as a finite intersection of convex sets. The vectors $\nu_{\theta,\theta'}$ are functions of $l^*$ only, hence of $\{l_{j,0}\}$ and $\pi^{A'}$ only, by \eqref{eq:lstar}.

\emph{Proposition, (ii).} \emph{Disjointness.} If $p\in\mathcal R(\theta)\cap\mathcal R(\theta')$ with $\theta\neq\theta'$, then $g_{\theta,\theta'}(p)>0$ and $g_{\theta',\theta}(p)=-g_{\theta,\theta'}(p)>0$, a contradiction.

\emph{Points outside the walls.} Let $p\in\Delta(S)\setminus\mathcal W$. Then $g_{\theta,\theta'}(p)\neq0$ for all $\theta\neq\theta'$, so by \eqref{eq:geom-identity} the $m$ numbers $\Lambda_\theta(p)$ are pairwise distinct. They have a unique maximizer $\hat\theta$, which satisfies $g_{\hat\theta,\theta'}(p)>0$ for all $\theta'\neq\hat\theta$, so $p\in\mathcal R(\hat\theta)$. By disjointness, $p$ lies in no other cell.

\emph{The walls.} Assume $l^*(\cdot\mid\theta)\neq l^*(\cdot\mid\theta')$ for all $\theta\neq\theta'$, and fix such a pair, with $\nu:=\nu_{\theta,\theta'}$. Then $\nu$ is not constant on $S$: if $\nu(s)=c$ for all $s$, then $l^*(s\mid\theta)=e^{c}\,l^*(s\mid\theta')$ for all $s$, summing over $s$ gives $e^{c}=1$, so $c=0$ and the two interpretations coincide, a contradiction. Let $\mathcal H:=\{p\in\R^S:\sum_sp(s)=1\}$ be the affine hull of $\Delta(S)$, an affine space of dimension $|S|-1$ whose direction space is $\mathcal H_0:=\{v\in\R^S:\sum_sv(s)=0\}$. Choose $s_1,s_2\in S$ with $\nu(s_1)\neq\nu(s_2)$. The vector $v:=e_{s_1}-e_{s_2}$ lies in $\mathcal H_0$ and satisfies $\sum_sv(s)\nu(s)=\nu(s_1)-\nu(s_2)\neq0$. Hence the map $p\mapsto\sum_sp(s)\nu(s)$ is affine and nonconstant on $\mathcal H$, and its zero set in $\mathcal H$ is an affine hyperplane of $\mathcal H$, of dimension $|S|-2$. Such a set has zero $(|S|-1)$-dimensional Lebesgue measure in $\mathcal H$. The wall between $\theta$ and $\theta'$ is its intersection with $\Delta(S)$, a piece of a hyperplane (possibly empty), and also has zero measure. Since $\mathcal W$ is the finite union of these pieces over pairs $\theta\neq\theta'$, it has zero $(|S|-1)$-dimensional Lebesgue measure.

\emph{Proposition, (iii).} By (i), the cells depend only on $l^*$. By Proposition~\ref{prop:effective_single_agent}, the data $\{f_i\}$ and the weights $\pi^A$ affect the set of learned hypotheses only through $\bar f$.

\emph{Remark: relative openness and form.} Each $H_{\theta'}$ in \eqref{eq:cell-intersection} is open in $\R^S$, as the preimage of $(0,\infty)$ under a continuous map. Hence $\mathcal R(\theta)$ is the trace on $\Delta(S)$ of an open subset of $\R^S$, so it is relatively open in $\Delta(S)$. If $\nu_{\theta,\theta'}\neq0$ for every $\theta'\neq\theta$, each $H_{\theta'}$ is an open half-space, and \eqref{eq:cell-intersection} expresses $\mathcal R(\theta)$ as the intersection of $\Delta(S)$ with $m-1$ open half-spaces. If $\nu_{\theta,\theta'}=0$ for some $\theta'\neq\theta$, then $H_{\theta'}=\emptyset$ and $\mathcal R(\theta)=\emptyset$. In either case, the stated alternative holds.

\emph{Remark: shared interpretations give empty cells.} If $l^*(\cdot\mid\theta)=l^*(\cdot\mid\theta')$ for some $\theta\neq\theta'$, then $\nu_{\theta,\theta'}=0$, so $\mathcal R(\theta)=\emptyset$ by the previous step, and symmetrically $\mathcal R(\theta')=\emptyset$.

\emph{Remark: nonempty cells are not closed.} First, $\mathcal R(\theta)\neq\Delta(S)$. Indeed, fix any $\theta'\neq\theta$, which exists since $m\ge2$. For each vertex $e_s$ of $\Delta(S)$, $g_{\theta,\theta'}(e_s)=\log\bigl(l^*(s\mid\theta)/l^*(s\mid\theta')\bigr)$. If every vertex belonged to $\mathcal R(\theta)$, we would have $l^*(s\mid\theta)>l^*(s\mid\theta')$ for all $s\in S$, and summing over $s$ would give $1>1$, a contradiction. Next, $\Delta(S)$ is convex, hence connected, so its only subsets that are both relatively open and relatively closed are $\emptyset$ and $\Delta(S)$. Since $\mathcal R(\theta)$ is relatively open and different from $\Delta(S)$, it is not relatively closed in $\Delta(S)$ unless it is empty. Because $\Delta(S)$ is closed in $\R^S$, a subset of $\Delta(S)$ is relatively closed if and only if it is closed in $\R^S$. Hence every nonempty cell is not closed. A polytope, the convex hull of finitely many points, is compact and therefore closed, so a nonempty cell is not a polytope.

\emph{Remark: cells may contain points on other walls.} We give an example. Any family of full-support distributions $\{q_\theta\}_{\theta\in\Theta}$ can arise as the consensus interpretation: if $l_{j,0}(\cdot\mid\theta)=q_\theta$ for every $j\in N$, then \eqref{eq:lstar} gives $l^*(\cdot\mid\theta)=q_\theta$, since the weights $\pi^{A'}_j$ sum to one. Take $S=\{s_1,s_2\}$, $\Theta=\{\theta_1,\theta_2,\theta_3\}$, and

\begin{equation*}
    l^*(\cdot\mid\theta_1)=(0.5,0.5),\qquad
    l^*(\cdot\mid\theta_2)=(0.9,0.1),\qquad
    l^*(\cdot\mid\theta_3)=(0.1,0.9),
\end{equation*}

and let $p=(0.5,0.5)$. Then $\Lambda_{\theta_1}(p)=\log0.5$ and $\Lambda_{\theta_2}(p)=\Lambda_{\theta_3}(p)=\tfrac12\log0.09=\log0.3$. Hence $g_{\theta_1,\theta_2}(p)=g_{\theta_1,\theta_3}(p)=\log(5/3)>0$, so $p\in\mathcal R(\theta_1)$, while $g_{\theta_2,\theta_3}(p)=0$, so $p$ lies on the wall between $\theta_2$ and $\theta_3$, both of which fit $p$ worse than $\theta_1$.
\end{proof}

\subsection{Proofs of Section \ref{sec:comparison}}

\subsubsection{Proofs of Subsection \ref{subsec:misspecified_benchmark}}

\begin{proof}[Proof of Remark~\ref{rem:fixed-profile}]\label{app:fixed-profile}
With $A'=I$ we have $l_{i,t}=l_{i,0}$ and $Z^\theta_{i,t-1}=1$, so \eqref{eq:increment} reduces to $[\bm x_t]_i=\log\bigl(l_{i,0}(\omega_{it}\mid\theta_v)/l_{i,0}(\omega_{it}\mid\theta_w)\bigr)$. These increments are i.i.d.\ in $t$ with constant mean $[\bm d^{NS}]_i:=\Dkl{f_i}{l_{i,0}(\cdot\mid\theta_w)}-\Dkl{f_i}{l_{i,0}(\cdot\mid\theta_v)}$. Lemma~\ref{lem:unifbound} holds unchanged, since its proof only uses that $A'^{\,t}$ is row-stochastic. Lemma~\ref{lem:mixing} involves only $A$. In Lemma~\ref{lem:cesaro}, step (b) is trivial because the means are constant. Hence $\frac1t\bm\phi_t\to\one\pi^{A\top}\bm d^{NS}$ a.s., with $\pi^{A\top}\bm d^{NS}=C^{NS}(\theta_v)-C^{NS}(\theta_w)$, and the proof of Theorem~\ref{thm:consistency} applies with $C^{NS}$ in place of $C$.
\end{proof}

\begin{proof}[Proof of Theorem \ref{thm:reversal}]
\emph{Construction.} Let $n=2$, $S=\{a,b\}$, $\Theta=\{\theta_1,\theta_2\}$, and

\begin{equation}
    A=\begin{pmatrix}3/4&1/4\\[2pt]3/4&1/4\end{pmatrix},
    \qquad
    A'=\begin{pmatrix}1/4&3/4\\[2pt]1/4&3/4\end{pmatrix}.
\end{equation}

Both are row-stochastic with strictly positive entries, hence primitive, so (A1) holds (with $\rho_A=\rho_{A'}=0$, since each matrix has rank one). Since all rows of $A$ (resp.\ $A'$) coincide, each stationary distribution equals the common row: $\pi^A=(3/4,1/4)$ and $\pi^{A'}=(1/4,3/4)$. Let both agents share the truth $f_1=f_2=f:=(0.9,0.1)\in\ri\Delta(S)$, and let the initial profiles be the swapped full-support models

\begin{equation}
    \begin{aligned}
    l_{1,0}(\cdot\mid\theta_1)&=(0.9,0.1), & l_{1,0}(\cdot\mid\theta_2)&=(0.1,0.9),\\
    l_{2,0}(\cdot\mid\theta_1)&=(0.1,0.9), & l_{2,0}(\cdot\mid\theta_2)&=(0.9,0.1),
    \end{aligned}
\end{equation}

with any $\mu_{i,0}\in\ri\Delta(\Theta)$, so (A2)-(A3) hold. Thus agent $1$'s interpretation of $\theta_1$ (and agent $2$'s interpretation of $\theta_2$) matches the truth exactly, while the other two interpretations are maximally swapped.

\emph{Benchmark ranking.} Let $\kappa:=\Dkl{(0.9,0.1)}{(0.1,0.9)}=0.8\log 9=1.6\log 3$. The four benchmark divergences follow directly from the swapped construction: each agent's own model fits the truth perfectly, $\Dkl{f}{l_{1,0}(\cdot\mid\theta_1)}=\Dkl{f}{l_{2,0}(\cdot\mid\theta_2)}=0$, while the cross terms cost $\Dkl{f}{l_{1,0}(\cdot\mid\theta_2)}=\Dkl{f}{l_{2,0}(\cdot\mid\theta_1)}=\kappa$. By \eqref{eq:CNS},

\begin{equation}
    C^{NS}(\theta_1)=-\tfrac14\kappa,\quad
    C^{NS}(\theta_2)=-\tfrac34\kappa,\quad
    C^{NS}(\theta_1)-C^{NS}(\theta_2)=\tfrac12\kappa=0.8\log3>0 .
    \label{eq:CNS-ineq}
\end{equation}

The benchmark favors $\theta_1$: the agent whose interpretation of $\theta_1$ is correct carries three quarters of the belief-network influence.

\emph{Pooled ranking.} By \eqref{eq:lstar} with weights $\pi^{A'}=(1/4,3/4)$, the odds of $l^*(\cdot\mid\theta_1)$ are $l^*(a\mid\theta_1):l^*(b\mid\theta_1)=0.9^{1/4}\,0.1^{3/4}:0.1^{1/4}\,0.9^{3/4}=9^{1/4-3/4}=9^{-1/2}=1:3$, so after normalization $l^*(\cdot\mid\theta_1)=(1/4,3/4)$ and, by the symmetry $a\leftrightarrow b,\ \theta_1\leftrightarrow\theta_2$, $l^*(\cdot\mid\theta_2)=(3/4,1/4)$. Since $f_1=f_2=f$, the pooled truth is $\bar f=f=(0.9,0.1)$, whatever $\pi^A$. By Proposition~\ref{prop:effective_single_agent} in log-likelihood-ratio form \eqref{eq:concrete-llr},

\begin{equation}
    C(\theta_2)-C(\theta_1)
    =\sum_s\bar f(s)\log\frac{l^*(s\mid\theta_2)}{l^*(s\mid\theta_1)}
    =0.9\log\frac{3/4}{1/4}+0.1\log\frac{1/4}{3/4}
    =0.8\log3>0 .
      \label{eq:C-ineq}
\end{equation}

(Numerically, $\Dkl{\bar f}{l^*(\cdot\mid\theta_1)}\approx0.9514$ and $\Dkl{\bar f}{l^*(\cdot\mid\theta_2)}\approx0.0725$, difference $0.8\log3\approx0.8789$.)

Inequalities \eqref{eq:CNS-ineq} and \eqref{eq:C-ineq} give \eqref{eq:reversal}. With $|\Theta|=2$ they give $\Theta^{NS}=\{\theta_1\}$ and $\Theta^*=\{\theta_2\}$, and Corollary~\ref{cor:singleton} yields the stated a.s.\ concentration.

\emph{Characterization \eqref{eq:reversal-box}.} Fix any configuration satisfying Assumption~\ref{ass:standing}, and let $\mathfrak D(\theta)$ be the matrix with entries $\mathfrak D_{ij}(\theta)$. By \eqref{eq:CNS}, $C^{NS}(\theta_1)-C^{NS}(\theta_2)=\sum_i\pi^A_i\bigl(\mathfrak D_{ii}(\theta_2)-\mathfrak D_{ii}(\theta_1)\bigr)$, which is the left side of \eqref{eq:reversal-box}. 

Inserting $\log l^*(s\mid\theta)=\sum_j\pi^{A'}_j\log l_{j,0}(s\mid\theta)-\log\Psi^\theta_\infty$ from \eqref{eq:lstar} into the KL divergence, and using $\sum_j\pi^{A'}_j=1$, gives:

\begin{equation*}
    \Dkl{f_i}{l^*(\cdot\mid\theta)}=\sum_j\pi^{A'}_j\mathfrak D_{ij}(\theta)+\log\Psi^\theta_\infty.
\end{equation*}

Hence $C(\theta)=-\pi^{A\top}\mathfrak D(\theta)\pi^{A'}-\log\Psi^\theta_\infty$, and $C(\theta_1)<C(\theta_2)$ is equivalent to $\pi^{A\top}\bigl(\mathfrak D(\theta_2)-\mathfrak D(\theta_1)\bigr)\pi^{A'}+\log\Psi^{\theta_2}_\infty-\log\Psi^{\theta_1}_\infty<0$, the right side of \eqref{eq:reversal-box}.
\end{proof}

\subsubsection{Proofs of Subsection \ref{subsec:wellspecified}}

\begin{proof}[Proof of Proposition \ref{prop:ws-benchmark}]
By \eqref{eq:wellspec}, $\Dkl{f_i}{l_{i,0}(\cdot\mid\theta^\star)}=\Dkl{f_i}{f_i}=0$, so $C^{NS}(\theta^\star)=0$ in \eqref{eq:CNS}. Since every KL term is nonnegative and $\pi^A_i>0$, $C^{NS}(\theta)\le0$ for all $\theta$, so $\theta^\star$ is a maximizer. For $\theta\neq\theta^\star$, \eqref{eq:collid} gives some $i$ with $\Dkl{f_i}{l_{i,0}(\cdot\mid\theta)}>0$, whence $C^{NS}(\theta)<0$. Thus $\Theta^{NS}=\{\theta^\star\}$, and Remark~\ref{rem:fixed-profile} yields the a.s.\ limit.
\end{proof}

\begin{proof}[Proof of Proposition \ref{prop:ws-survival}]
By \eqref{eq:effective-argmin}, $\theta^\star\in\Theta^*$ if and only if $\Dkl{\bar f}{l^*(\cdot\mid\theta^\star)}\le\Dkl{\bar f}{l^*(\cdot\mid\theta)}$ for every $\theta$, and $\Theta^*=\{\theta^\star\}$ if and only if this is strict for every $\theta\neq\theta^\star$. The left-hand side is $\varepsilon^\star$ by \eqref{eq:eps-star}, with $l^*(\cdot\mid\theta^\star)$ given by \eqref{eq:lstar-truth} under \eqref{eq:wellspec}.
\end{proof}

\begin{proof}[Proof of Proposition \ref{prop:ws-homog}]
Homogeneity and \eqref{eq:wellspec} give $l_{j,0}(\cdot\mid\theta^\star)=f$ for every $j$, so the pool \eqref{eq:lstar-truth} equals $f$, also $\bar f=f$. Hence $\Dkl{\bar f}{l^*(\cdot\mid\theta^\star)}=0$ attains the global minimum in \eqref{eq:effective-argmin}, so $\theta^\star\in\Theta^*$, with a tie only if $\Dkl{\bar f}{l^*(\cdot\mid\theta)}=0$, i.e.\ $l^*(\cdot\mid\theta)=f$, for some $\theta\neq\theta^\star$.

This tie happens only in exceptional cases. Fix $f$, $\pi^{A'}$, a rival $\theta\neq\theta^\star$, and a reference signal $s_0\in S$. By \eqref{eq:lstar}, for every $s$,

\begin{equation*}
    \log\frac{l^*(s\mid\theta)}{l^*(s_0\mid\theta)} =\sum_{j=1}^n\pi^{A'}_j\log\frac{l_{j,0}(s\mid\theta)}{l_{j,0}(s_0\mid\theta)},
\end{equation*}

since the normalizer $\Psi^\theta_\infty$ cancels. So $l^*(\cdot\mid\theta)=f$ requires

\begin{equation*}
    \log\frac{l_{1,0}(s\mid\theta)}{l_{1,0}(s_0\mid\theta)} =\frac{1}{\pi^{A'}_1}\Bigl(\log\frac{f(s)}{f(s_0)}-\sum_{j\ge2}\pi^{A'}_j\log\frac{l_{j,0}(s\mid\theta)}{l_{j,0}(s_0\mid\theta)}\Bigr) \qquad\text{for all }s,
\end{equation*}

which is well defined because $\pi^{A'}_1>0$. Once the models of agents $j\ge2$ are fixed, this determines $l_{1,0}(\cdot\mid\theta)$ uniquely. The tie therefore occurs for a single choice of $l_{1,0}(\cdot\mid\theta)$, a set of Lebesgue measure zero in $\ri\Delta(S)$ (as $|S|\ge2$). By Fubini's theorem, the set of profiles producing a tie has measure zero, and so does its finite union over $\theta\neq\theta^\star$.
\end{proof}

\begin{proof}[Proof of Proposition \ref{prop:ws-norescue}]
We prove the contrapositive. Suppose belief-only learning does not give $\mu^{NS}_{i,t}(\theta^\star)\to1$ a.s.\ for every $i$. By Proposition~\ref{prop:ws-benchmark}, condition \eqref{eq:collid} must then fail: there is $\theta\neq\theta^\star$ with $l_{i,0}(\cdot\mid\theta)=f_i$ for every $i\in N$. By \eqref{eq:wellspec}, $l_{i,0}(\cdot\mid\theta)=f_i=l_{i,0}(\cdot\mid\theta^\star)$ for every $i$, so $\theta$ and $\theta^\star$ are twins in the sense of Lemma~\ref{lem:twins}. Applying that lemma to the sharing dynamics with $(A,A')$ gives $\mu_{i,t}(\theta^\star)\le(1+e^{-c})^{-1}<1$ for every $i$, every $t$ and every realization of the signals. Hence $\mu_{i,t}(\theta^\star)$ converges to $1$ for no $i$ and on no sample path, so sharing does not learn $\theta^\star$ either.
\end{proof}

\begin{proof}[Proof of Proposition \ref{prop:ws-hurt}]
Take $n=2$, $S=\{a,b\}$, $\Theta=\{\theta^\star,\theta_1\}$, and \emph{equal} influence in both channels, $\pi^A=\pi^{A'}=(\tfrac12,\tfrac12)$ (e.g.\ $A=A'$ with identical rows). Let the truths be heterogeneous,

\begin{equation*}
    f_1=(0.9,0.1),\qquad f_2=(0.6,0.4),
\end{equation*}

and set $l_{i,0}(\cdot\mid\theta^\star)=f_i$ (so \eqref{eq:wellspec} holds) and $l_{1,0}(\cdot\mid\theta_1)=l_{2,0}(\cdot\mid\theta_1)=(0.75,0.25)$. Then, all models have full support and $\mu_{i,0}\in\ri\Delta(\Theta)$ is arbitrary, so (A1)-(A3) hold.

\emph{Effective objects.} The pooled truth is $\bar f=\tfrac12 f_1+\tfrac12 f_2=(0.75,0.25)$. The consensus model at $\theta^\star$ is the equal-weight geometric pool of $f_1,f_2$,

\begin{equation*}
    l^*(\cdot\mid\theta^\star) \propto\bigl(\sqrt{0.9\cdot0.6},\,\sqrt{0.1\cdot0.4}\bigr) =(\sqrt{0.54},\,0.2)\ \Longrightarrow\ l^*(\cdot\mid\theta^\star)\approx(0.7861,0.2139),
\end{equation*}

while $l^*(\cdot\mid\theta_1)=(0.75,0.25)=\bar f$ (identical pooled inputs).

\emph{Sharing learns the falsehood.} By \eqref{eq:effective-argmin},

\begin{equation*}
    \Dkl{\bar f}{l^*(\cdot\mid\theta_1)}=\Dkl{\bar f}{\bar f}=0
    \;<\;
    \Dkl{\bar f}{l^*(\cdot\mid\theta^\star)}\approx0.0037,
\end{equation*}

so $\Theta^*=\{\theta_1\}$ and Corollary~\ref{cor:singleton} gives $\mu_{i,t}(\theta_1)\to1$ a.s.

\emph{Belief-only sharing learns the truth.} By \eqref{eq:CNS}, $C^{NS}(\theta^\star)=0$ (from \eqref{eq:wellspec}), while

\begin{equation*}
    C^{NS}(\theta_1)
    =-\tfrac12\Dkl{f_1}{(0.75,0.25)}-\tfrac12\Dkl{f_2}{(0.75,0.25)}
    \approx-\tfrac12(0.0725+0.0541)\approx-0.0633<0,
\end{equation*}

so $\Theta^{NS}=\{\theta^\star\}$ and Proposition~\ref{prop:ws-benchmark} gives $\mu_{i,t}(\theta^\star)\to1$ a.s.

Both strict inequalities are exact, not only numerical: $(0.75,0.25)\neq f_i$ gives $C^{NS}(\theta_1)<0$ by Gibbs' inequality; and $l^*(\cdot\mid\theta^\star)=\bar f$ would require $\sqrt{0.54}=0.75(\sqrt{0.54}+0.2)$, i.e.\ $0.54=0.36$, so $\Dkl{\bar f}{l^*(\cdot\mid\theta^\star)}>0$.

\emph{Openness.} Regard configurations with $n=2$, $|S|=2$, $\Theta=\{\theta^\star,\theta_1\}$ as points $(f_1,f_2,l_{1,0}(\cdot\mid\theta_1),l_{2,0}(\cdot\mid\theta_1),A,A')$, with $l_{i,0}(\cdot\mid\theta^\star)=f_i$ imposed by \eqref{eq:wellspec}, where the distributions lie in $\ri\Delta(S)$ and $A,A'$ are $2\times2$ row-stochastic matrices with positive entries (hence satisfy (A1)). For such matrices $\pi^A=(a_{21},a_{12})/(a_{12}+a_{21})$, which is continuous in $A$, and likewise for $A'$. The quantities $C^{NS}(\theta^\star)-C^{NS}(\theta_1)$ and $C(\theta^\star)-C(\theta_1)$ are continuous in the configuration, since they are built from KL divergences, $\bar f$ and $l^*$ via \eqref{eq:lstar}. The set of configurations with $f_1\neq f_2$, $C^{NS}(\theta^\star)>C^{NS}(\theta_1)$ and $C(\theta^\star)<C(\theta_1)$ is therefore open. It contains the configuration above, so it is nonempty.
\end{proof}

\begin{proof}[Proof of Proposition \ref{thm:d-closed}]
Using $\Dkl{f_i}{l_{j,t}(\cdot\mid\theta)}=-H(f_i)-\sum_s f_i(s)\log l_{j,t}(s\mid\theta)$, the entropy term cancels in \eqref{eq:Ddef}, leaving $[D_t]_{ij}=-\sum_s f_i(s)\bigl(\log l_{j,t}(s\mid\theta_w)-\log l_{j,t}(s\mid\theta_v)\bigr)$. By Proposition~\ref{prop:l-closed}, $\log l_{j,t}(s\mid\theta)=\sum_k[A'^{\,t}]_{jk}\log l_{k,0}(s\mid\theta)-\log\Psi^\theta_{j,t}$, so

\begin{equation*}
    \log l_{j,t}(s\mid\theta_w)-\log l_{j,t}(s\mid\theta_v)
    =\sum_k[A'^{\,t}]_{jk}\bigl(\log l_{k,0}(s\mid\theta_w)-\log l_{k,0}(s\mid\theta_v)\bigr)
    -\bigl(\log\Psi^{\theta_w}_{j,t}-\log\Psi^{\theta_v}_{j,t}\bigr).
\end{equation*}

Multiplying by $-f_i(s)$, summing over $s$, and using $\sum_k[A'^{\,t}]_{jk}=1$ identifies the first term with $\sum_k[A'^{\,t}]_{jk}[D_0]_{ik}$ and the second with the rank-one correction, giving \eqref{eq:d-closed}.
\end{proof}

\begin{proof}[Proof of Corollary \ref{cor:d-conv}]
Fix the pair $(\theta_v,\theta_w)$ and $\hat\rho\in(\rho_{A'},1)$; if $A'$ is diagonalizable, $\hat\rho=\rho_{A'}$ is also allowed. Apply Steps~1-3 of the proof of Theorem~\ref{thm:l-consensus} to $\theta_v$ and to $\theta_w$, and let $L$ and $C_\Psi$ be the larger of the two values of $L$ and $|S|C_2$ obtained there. Then, for $\theta\in\{\theta_v,\theta_w\}$, all $j\in N$ and all $t\ge0$,
\begin{equation*}
    \bigl|\Psi^\theta_{j,t}-\Psi^\theta_\infty\bigr|\le C_\Psi\,\hat\rho^{\,t},\qquad \Psi^\theta_{j,t}\ge|S|e^{-L},\qquad \Psi^\theta_\infty\ge|S|e^{-L}.
\end{equation*}
Since $\log$ is $(e^{L}/|S|)$-Lipschitz on $[|S|e^{-L},\infty)$,
\begin{equation}
    \bigl|\log\Psi^\theta_{j,t}-\log\Psi^\theta_\infty\bigr|\le C_\ell\,\hat\rho^{\,t},\qquad C_\ell:=e^{L}C_\Psi/|S| .
    \label{eq:logpsi-rate}
\end{equation}

\emph{Limit of $D_t$.} Inserting $\log l^*(s\mid\theta)=\sum_k\pi^{A'}_k\log l_{k,0}(s\mid\theta)-\log\Psi^\theta_\infty$ from \eqref{eq:lstar} into the KL divergence and using $\sum_k\pi^{A'}_k=1$ gives $\Dkl{f_i}{l^*(\cdot\mid\theta)}=\sum_k\pi^{A'}_k\Dkl{f_i}{l_{k,0}(\cdot\mid\theta)}+\log\Psi^\theta_\infty$. Hence
\begin{equation*}
    d^*_i=\sum_k\pi^{A'}_k[D_0]_{ik}+\log\Psi^{\theta_w}_\infty-\log\Psi^{\theta_v}_\infty .
\end{equation*}
Subtracting this from the $(i,j)$ entry of \eqref{eq:d-closed},
\begin{equation*}
    [D_t]_{ij}-d^*_i=\sum_k\bigl([A'^{\,t}]_{jk}-\pi^{A'}_k\bigr)[D_0]_{ik}
    +\bigl(\log\Psi^{\theta_w}_{j,t}-\log\Psi^{\theta_w}_\infty\bigr)-\bigl(\log\Psi^{\theta_v}_{j,t}-\log\Psi^{\theta_v}_\infty\bigr).
\end{equation*}
By Lemma~\ref{lem:pf-rate} and \eqref{eq:logpsi-rate},
\begin{equation*}
    \max_{i,j}\bigl|[D_t]_{ij}-d^*_i\bigr|\le\max_{i,k}\bigl|[D_0]_{ik}\bigr|\,\bigl\|A'^{\,t}-\one\pi^{A'\top}\bigr\|_\infty+2C_\ell\hat\rho^{\,t}
    \le\Bigl(C_{A',\hat\rho}\max_{i,k}\bigl|[D_0]_{ik}\bigr|+2C_\ell\Bigr)\hat\rho^{\,t}.
\end{equation*}
In particular $D_t\to\bm d^*\one^{\top}$.

\emph{Drift.} Since $\sum_ja'_{ij}=1$, $[\bm e_t]_i-d^*_i=\sum_ja'_{ij}\bigl([D_t]_{ij}-d^*_i\bigr)$ is a convex combination. Hence $\|\bm e_t-\bm d^*\|_\infty\le\max_{i,j}|[D_t]_{ij}-d^*_i|$, and the previous bound applies.

\emph{Normalizer gap.} Identity \eqref{eq:Z-Psi} at time $t+1$ reads $Z^\theta_{i,t}=\Psi^\theta_{i,t+1}/\prod_j(\Psi^\theta_{j,t})^{a'_{ij}}$. Using $\sum_ja'_{ij}=1$,
\begin{equation*}
    \log Z^\theta_{i,t}=\bigl(\log\Psi^\theta_{i,t+1}-\log\Psi^\theta_\infty\bigr)-\sum_ja'_{ij}\bigl(\log\Psi^\theta_{j,t}-\log\Psi^\theta_\infty\bigr),
\end{equation*}
so by \eqref{eq:logpsi-rate} $|\log Z^\theta_{i,t}|\le C_\ell\hat\rho^{\,t+1}+C_\ell\hat\rho^{\,t}\le2C_\ell\hat\rho^{\,t}$. Therefore $\|\bm\delta_t\|_\infty\le4C_\ell\hat\rho^{\,t}$.

Since $\hat\rho$ was arbitrary, with $\hat\rho=\rho_{A'}$ admissible when $A'$ is diagonalizable, all three claims hold in the sense of Section~\ref{subsec:notation}.
\end{proof}

\subsubsection{Proofs for Subsection~\ref{subsec:coherent_sharing_good_or_bad}}

\begin{lemma}[Properties of $\Delta$]\label{lem:Delta}
Fix $\theta\neq\theta^\star$ and $j\in N$.

\begin{enumerate}
  \item[(i)] The map $p\mapsto\Delta_j(\theta;p)$ is finite and linear on $\Delta(S)$. In particular, $\Delta_j\bigl(\theta;\sum_m w_mp_m\bigr)=\sum_m w_m\,\Delta_j(\theta;p_m)$ for any weights $w\in\Delta(N)$.
  
  \item[(ii)] $\Delta_j(\theta;p)=\Dkl{p}{l_{j,0}(\cdot\mid\theta)}-\Dkl{p}{l_{j,0}(\cdot\mid\theta^\star)}$ for every $p\in\Delta(S)$.
  
  \item[(iii)] The belief-only and sharing gaps $G^{NS}_\theta:=C^{NS}(\theta^\star)-C^{NS}(\theta)$ and $G_\theta:=C(\theta^\star)-C(\theta)$ satisfy
  
  \begin{align}
    G^{NS}_\theta&=\sum_k\pi^A_k\,\Delta_k(\theta;f_k), \label{eq:GNS}\\
    G_\theta&=\sum_k\sum_m\pi^{A'}_k\,\pi^A_m\,\Delta_k(\theta;f_m)+\log\Psi^\theta_\infty-\log\Psi^{\theta^\star}_\infty . \label{eq:Gshare}
  \end{align}
\end{enumerate}
\end{lemma}

\begin{proof}[Proof of Lemma~\ref{lem:Delta}]
(i) $\Delta_j(\theta;p)=\sum_sp(s)v_j(s)$ with the fixed vector $v_j(s):=\log\bigl(l_{j,0}(s\mid\theta^\star)/l_{j,0}(s\mid\theta)\bigr)$, which is finite by (A2).

(ii) Write both divergences with the entropy identity of Appendix~\ref{app:proofs}. The entropy $H(p)$ cancels in the difference, leaving $\sum_sp(s)v_j(s)$.

(iii) By (ii) with $p=f_k$, $\Dkl{f_k}{l_{k,0}(\cdot\mid\theta)}-\Dkl{f_k}{l_{k,0}(\cdot\mid\theta^\star)}=\Delta_k(\theta;f_k)$, and \eqref{eq:CNS} gives \eqref{eq:GNS}. For the sharing gap, \eqref{eq:confidence-gap-pooled} and the same entropy cancellation give $G_\theta=\sum_s\bar f(s)\log\bigl(l^*(s\mid\theta^\star)/l^*(s\mid\theta)\bigr)$. Substitute $\log l^*(s\mid\theta)=\sum_k\pi^{A'}_k\log l_{k,0}(s\mid\theta)-\log\Psi^\theta_\infty$ from \eqref{eq:lstar} and use $\sum_s\bar f(s)=1$. This gives $G_\theta=\sum_k\pi^{A'}_k\Delta_k(\theta;\bar f)+\log\Psi^\theta_\infty-\log\Psi^{\theta^\star}_\infty$. Finally, $\bar f=\sum_m\pi^A_mf_m$ and (i) yield \eqref{eq:Gshare}.
\end{proof}

\begin{proof}[Proof of Lemma~\ref{lem:gt-unique}]
Suppose $\theta_1\neq\theta_2$ are both ground truths, and fix any $i\in N$. Since $\theta_1$ is a ground truth, $f_i\notin\bigcup_{\theta\neq\theta_1}\mathcal T(\theta)$. This union contains $\mathcal T(\theta_2)$, so $f_i\notin\mathcal T(\theta_2)$. Since $\theta_2$ is a ground truth, $f_i\in\mathcal T(\theta_2)$. This is a contradiction.
\end{proof}

\begin{proof}[Proof of Theorem~\ref{thm:rescue-corrupt}] 
Belief-only learning is the dynamics \eqref{eq:S1}-\eqref{eq:S3} on $(A,I)$, and interpretation sharing is the same dynamics on $(A,A')$, both from the same initial conditions satisfying (A2). Let $i$ and $j\neq i$ be the two agents in the hypothesis, so $n\ge2$.

\emph{Step 1 (the gaps in terms of $\Delta$).} For $\theta\neq\theta^\star$, the gaps $G^{NS}_\theta$ and $G_\theta$ are given by \eqref{eq:GNS} and \eqref{eq:Gshare} of Lemma~\ref{lem:Delta}(iii).

\emph{Step 2 (from the signs of the gaps to beliefs).} Fix any initial beliefs satisfying (A2) and any $k\in N$.

\emph{Sharing.} If $G_\theta>0$ for every $\theta\neq\theta^\star$, then $\Theta^*=\{\theta^\star\}$. Corollary~\ref{cor:singleton} gives $\mu_{k,t}(\theta^\star)\to1$ a.s. If $G_\theta<0$ for some $\theta$, then $\theta^\star\notin\Theta^*$, and Theorem~\ref{thm:consistency} gives $\mu_{k,t}(\theta^\star)\to0$ a.s.

\emph{Belief-only.} By Remark~\ref{rem:fixed-profile}, $\mu^{NS}_{k,t}(\theta)\to0$ a.s.\ for every $\theta\notin\Theta^{NS}$.
If $G^{NS}_\theta>0$ for every $\theta\neq\theta^\star$, then $\Theta^{NS}=\{\theta^\star\}$. The finitely many beliefs $\mu^{NS}_{k,t}(\theta)$, $\theta\neq\theta^\star$, tend to $0$ a.s. Since beliefs sum to one, $\mu^{NS}_{k,t}(\theta^\star)\to1$ a.s. If $G^{NS}_\theta<0$ for some $\theta$, then $\theta^\star\notin\Theta^{NS}$ and $\mu^{NS}_{k,t}(\theta^\star)\to0$ a.s.

\emph{It therefore suffices to find networks $A,A'$ satisfying (A1) such that:}
\begin{itemize}
  \item in case (i), $G^{NS}_\theta<0$ for some $\theta\neq\theta^\star$ and $G_\theta>0$ for all $\theta\neq\theta^\star$;
  \item in case (ii), $G^{NS}_\theta>0$ for all $\theta\neq\theta^\star$ and $G_\theta<0$ for some $\theta\neq\theta^\star$.
\end{itemize}

\emph{Step 3 (a family of networks).}
For $\eta\in(0,1)$ let $u:=\one/n$ and

\begin{equation*}
    \pi^A(\eta):=(1-\eta)e_i+\eta u,\qquad \pi^{A'}(\eta):=(1-\eta)e_j+\eta u,
\end{equation*}

with $e_i,e_j$ the unit vectors. Every entry is at least $\eta/n>0$, so both lie in $\ri\Delta(N)$. Define

\begin{equation*}
    A_\eta:=\tfrac12I+\tfrac12\one\,\pi^A(\eta)^\top,\qquad
    A'_\eta:=\tfrac12I+\tfrac12\one\,\pi^{A'}(\eta)^\top .
\end{equation*}

We check (A1) for $M=\tfrac12I+\tfrac12\one\pi^\top$ with $\pi\in\ri\Delta(N)$:

\begin{itemize}
  \item every entry is positive and each row sums to $\tfrac12+\tfrac12\sum_l\pi_l=1$, so $M$ is row-stochastic and primitive;
  \item $\pi^\top M=\tfrac12\pi^\top+\tfrac12(\pi^\top\one)\pi^\top=\pi^\top$, so $\pi$ is the (unique) stationary distribution;
  \item $M\one=\one$ and $Mv=\tfrac12v$ for every $v$ with $\pi^\top v=0$, so $\rho_M=\tfrac12\in(0,1)$.
\end{itemize}

Hence $(A_\eta,A'_\eta)$ satisfy (A1) with stationary distributions $\pi^A(\eta)$ and $\pi^{A'}(\eta)$. Denote by $G^{NS}_\theta(\eta)$ and $G_\theta(\eta)$ the gaps \eqref{eq:GNS} and \eqref{eq:Gshare} for these networks.

\emph{Step 4 (the limit $\eta\to0$).}
As $\eta\to0$, $\pi^A_m(\eta)\to\mathbb 1\{m=i\}$ and $\pi^{A'}_k(\eta)\to\mathbb 1\{k=j\}$. Hence, by \eqref{eq:GNS},

\begin{equation*}
    G^{NS}_\theta(\eta)\to\Delta_i(\theta;f_i),
\end{equation*}

and the double sum in \eqref{eq:Gshare} tends to $\Delta_j(\theta;f_i)$. For the normalizers,

\begin{equation*}
    \Psi^\theta_\infty(\eta)=\sum_s\prod_kl_{k,0}(s\mid\theta)^{\pi^{A'}_k(\eta)} .
\end{equation*}

Each factor $x^{a}$ with $x=l_{k,0}(s\mid\theta)>0$ is continuous in the exponent $a$. So $\Psi^\theta_\infty(\eta)\to\sum_sl_{j,0}(s\mid\theta)=1$, and $\log\Psi^\theta_\infty(\eta)\to0$ for every $\theta$. Therefore

\begin{equation*}
    G^{NS}_\theta(\eta)\to\Delta_i(\theta;f_i),\qquad G_\theta(\eta)\to\Delta_j(\theta;f_i),\qquad\theta\neq\theta^\star .
\end{equation*}

\emph{Step 5 (conclusion).}
\begin{itemize}
  \item In case (i), agent $i$ misreads $f_i$: $\Delta_i(\theta_1;f_i)<0$ for some $\theta_1\neq\theta^\star$. Agent $j$ reads $f_i$ correctly: $\Delta_j(\theta;f_i)>0$ for all $\theta\neq\theta^\star$.
  \item In case (ii), $\Delta_i(\theta;f_i)>0$ for all $\theta\neq\theta^\star$ and $\Delta_j(\theta_1;f_i)<0$ for some $\theta_1\neq\theta^\star$.
\end{itemize}

By Step~4, each of these finitely many strict inequalities also holds for $G^{NS}_\theta(\eta)$ and $G_\theta(\eta)$ once $\eta$ is small enough. Taking the smallest of the finitely many thresholds, all of them hold simultaneously for some $\eta\in(0,1)$. The networks $(A_\eta,A'_\eta)$ then satisfy the conditions at the end of Step~2, which proves (i) and (ii).
\end{proof}

\begin{remark}\label{rem:rc-scope}
The proof uses only the signs of $\Delta_i(\theta;f_i)$ and $\Delta_j(\theta;f_i)$. Coherence and the ground-truth property (Definition~\ref{def:coherent}) serve only to identify $\theta^\star$ as the true hypothesis (unique by Lemma~\ref{lem:gt-unique}). The construction takes $\pi^A,\pi^{A'}$ near unit vectors, which suffices for existence. For arbitrary networks, Step~2 with \eqref{eq:GNS} and \eqref{eq:Gshare} shows that sharing rescues whenever $\min_{\theta\neq\theta^\star}G^{NS}_\theta<0<\min_{\theta\neq\theta^\star}G_\theta$, and corrupts whenever $\min_{\theta\neq\theta^\star}G_\theta<0<\min_{\theta\neq\theta^\star}G^{NS}_\theta$.
\end{remark}

\begin{proof}[Proof of Corollary~\ref{cor:rc-common}]
Let $f_1=\dots=f_n=:f$, let $a$ be accurate and $b$ misreading.

Then $a\neq b$: accuracy requires $\Delta_a(\theta;f)>0$ for all $\theta\neq\theta^\star$, which excludes $\Delta_a(\theta;f)<0$ for some $\theta$.

Since $f_a=f_b=f$:

\begin{itemize}
  \item agent $b$ misreads its own data and agent $a$ reads $f_b=f_a$ correctly, so Theorem~\ref{thm:rescue-corrupt}(i) applies with $i=b$, $j=a$;
  \item agent $a$ reads its own data correctly and agent $b$ misreads $f_a=f_b$, so Theorem~\ref{thm:rescue-corrupt}(ii) applies with $i=a$, $j=b$.
\end{itemize}
\end{proof}

\bibliographystyle{abbrv}
\bibliography{bibliography}

\end{document}